\documentclass[11pt]{article}

\usepackage[utf8]{inputenc}

\usepackage{macro_yihan}

\title{Zero-Rate Thresholds and New Capacity Bounds\\for List-Decoding and List-Recovery}

\author{
    Nicolas Resch\thanks{University of Amsterdam. Email: \href{mailto:n.a.resch@uva.nl}{\texttt{n.a.resch@uva.nl}}.}
    \and
    Chen Yuan\thanks{Shanghai Jiao Tong University. Email: \href{mailto:chen_yuan@sjtu.cn.edu}{\texttt{chen\_yuan@sjtu.cn.edu}}.}
    \and
    Yihan Zhang\thanks{Institute of Science and Technology Austria. Email: \href{mailto:zephyr.z798@gmail.com}{\texttt{zephyr.z798@gmail.com}}.}
}

\begin{document}

\maketitle

\begin{abstract}

In this work we consider the list-decodability and list-recoverability of arbitrary $q$-ary codes, for all integer values of $q\geq 2$. A code is called $(p,L)_q$-list-decodable if every radius $pn$ Hamming ball contains less than $L$ codewords; $(p,\ell,L)_q$-list-recoverability is a generalization where we place radius $pn$ Hamming balls on every point of a combinatorial rectangle with side length $\ell$ and again stipulate that there be less than $L$ codewords. 


Our main contribution is to precisely calculate the maximum value of $p$ for which there exist infinite families of positive rate $(p,\ell,L)_q$-list-recoverable codes, the quantity we call the \emph{zero-rate threshold}. Denoting this value by $p_*$, we in fact show that codes correcting a $p_*+\varepsilon$ fraction of errors must have size $O_{\varepsilon}(1)$, i.e., independent of $n$. Such a result is typically referred to as a ``Plotkin bound.'' To complement this, a standard random code with expurgation construction shows that there exist positive rate codes correcting a $p_*-\varepsilon$ fraction of errors. We also follow a classical proof template (typically attributed to Elias and Bassalygo) to derive from the zero-rate threshold other tradeoffs between rate and decoding radius for list-decoding and list-recovery. 

Technically, proving the Plotkin bound boils down to demonstrating the Schur convexity of a certain function defined on the $q$-simplex as well as the convexity of a univariate function derived from it.  
We remark that an earlier argument claimed similar results for $q$-ary list-decoding; however, we point out that this earlier proof is flawed. 
\end{abstract}
\tableofcontents

\section{Introduction}
\label{sec:intro}


Given a code $\cC \subset [q]^n$, a fundamental problem of coding-theory is to determine how ``well-spread'' $\cC$ can be if we also insist that $\cC$ have large rate $R = \frac{\log_q|\cC|}{n}$. The most basic way of quantifying ``well-spread'' is by insisting that all pairs of codewords are far apart. That is, we hope that the minimum distance $d := \min\{\disth(\vec c,\vec c'): \vec c \neq \vec c' \in \cC\}$ is large, where $\disth(\cdot,\cdot)$ denotes Hamming distance, i.e., the number of coordinates on which the two strings differ. Equivalently, given any word $\vec y \in [q]^n$, we have that $|\bham(\vec y,r) \cap \cC| \leq 1$, where $r = \lfloor d/2 \rfloor$ and $\bham(\vec y,r) = \{\vec x \in [q]^n:\disth(\vec x,\vec y)\leq r\}$ denotes the Hamming ball of radius $r$ centered at $\vec y$. 

One can naturally relax this requirement to the notion of list-decodability: instead of upper-bounding $|\bham(\vec y,r) \cap \cC|$ by $1$, we upper bound it by a larger integer $L-1$.\footnote{We find it most convenient to let $L$ denote 1 more than the list-size, which is admittedly nonstandard, but will make our computations much cleaner.} Equivalently, if we place Hamming balls of radius $r$ on each codeword of $\cC$, no vector in $[q]^n$ is covered by $L$ or more balls. If $\cC$ satisfies this property we call it $(p,L)_q$-list-decodable. Initially introduced by Elias and Wozencraft in the 1950's~\cite{Elias57,Wozencraft58,elias1991error}, this relaxed notion of decoding has been intensively studied in recent years, in part motivated by purely coding-theoretic concerns, but also due to its connections with theoretical computer science more broadly~\cite{goldreich1989hard,babai1990bpp,lipton1990efficient,kushilevitz1993learning,jackson1997efficient,sudan2001pseudorandom}. 

A further generalization of list-decoding is provided by \emph{list-recoverability}. In this case, one considers tuples of input lists $\vcY = (\cY_1,\dots,\cY_n)$ where each $\cY_i \subset [q]$ is of size at most $\ell$, and the requirement is that the number of codewords $\vc$ satisfying $|\{i \in [n]:c_i \in \cY_i\}| \leq p n$ is at most $L-1$. Such a code is deemed $(p,\ell,L)_q$-list-recoverable. Note that $(p,1,L)_q$-list-recoverability is the same as $(p,L)_q$-list-decoding, demonstrating that list-recoverability is a more general notion. While it was originally defined as an abstraction required for the task of uniquely-/list-decoding concatenated codes~\cite{GuruswamiI01,GuruswamiI02,GuruswamiI03,GuruswamiI04}, it has since found myriad further applications in computer science more broadly, e.g., in cryptography~\cite{HIOS15,holmgren2021fiat}, randomness extraction~\cite{guruswami2009unbalanced}, hardness amplification~\cite{doron2020nearly}, group testing~\cite{INR10,NPR12}, streaming algorithms~\cite{doron2022high}, and beyond.

When it comes to list-decoding and list-recovery, the optimal tradeoff between decoding-radius $p$ and rate $R$ is well-understood if one is satisfied with list-sizes $L = O(1)$.\footnote{Or indeed, if we insist on $L$ just being subexponential.} That is, there exist $(p,\ell,O(\ell/\eps))_q$-list-recoverable codes of rate $1-H_{q,\ell}(p)-\eps$ where\footnotemark{} 
\[
    H_{q,\ell}(p) \coloneqq p\log_q\paren{\frac{q-\ell}{p}}+(1-p)\log_q\paren{\frac{\ell}{1-p}};
\]
conversely, if the rate is at least $1-H_{q,\ell}(p)+\eps$ then it will not be list-recoverable for any $L = o\paren{q^{\eps n}}$~\cite[Theorem~2.4.12]{resch2020thesis}. (Note that setting $\ell=1$ recovers the more well-known list-decoding capacity theorem.) While this already provides some ``coarse-grained'' information concerning the list-decodability/-recoverability of codes, it leaves many questions unanswered. 
\footnotetext{For $\ell = 1$, $ H_{q,1} $ reduces to the $q$-ary entropy function denoted by $ H_q $. }

For example, one can ask about the maximum rate of a $(p,3)_q$-list-decodable code. That is, what is the maximum rate of a code that never contains more than 2 points from a Hamming ball of radius $pn$? However, this question as stated appears to be quite difficult to solve: any improvement for the special case of $L=2$ and $q=2$ would require improving either on the Gilbert-Varshamov bound~\cite{gilbert1952comparison,varshamov1957estimate} (on the ``possibility'' side) or the linear programming bounds~\cite{mrrw1,mrrw2,delsarte-1973} (on the ``impossibility'' side). Unfortunately, despite decades of interest in this basic question hardly any asymptotic improvements on these bounds have been provided in the past fifty years. 

\paragraph{Zero-rate thresholds for list-decoding and -recovery} We therefore begin by targeting a more modest question: what is the maximum $p_* = p_*(q,\ell,L)$ such that for any $p<p_*$ there exist infinite families of $q$-ary $(p,\ell,L)_q$-list-recoverable codes of positive rate? That is, imagining the curve describing the achievable tradeoffs with the rate $R$ on the $y$-axis and decoding radius $p$ on the $x$-axis, instead of asking to describe this entire curve, we simply seek to determine the point where this curve crosses the $x$-axis (clearly, this curve is monotonically decreasing). 

Over the binary alphabet, setting $\ell=1$ and $L=2$ in this question we recover a famous result of Plotkin~\cite{plotkin-1960}: the maximum fraction of errors that can be uniquely-decoded by an infinite family of positive rate binary codes is $1/4$. Over general $q$-ary alphabets, this value is similarly known to be $\frac{q-1}{2q}$. The value of $p_*(2,1,L)$ has been computed by Blinovsky~\cite{blinovsky-1986-ls-lb-binary} for all $L$, and is known to be 
\begin{align*}
    p_*(2,1,L) = \frac{1}{2} - \frac{\binom{2k}{k}}{2^{2k+1}} \text{ if } L=2k \text{ or } L=2k+1 .
\end{align*}
While this expression is quite impenetrable at first glance, here is a natural probabilistic interpretation: given $x_1,\dots,x_L \in \{0,1\}$ let $\pl(x_1,\dots,x_L)$ denote the number of times the more popular bit appears.\footnote{We use $\pl$ to stand for ``plurality''. However, we caution that this function does not output the more popular symbol (as is perhaps more in line with the standard meaning of plurality), but the number of $i \in [L]$ for which $x_i$ equals the most popular symbol.} We then have 
\begin{align*}
    p_*(2,1,L) = 1 - \frac{1}{L}\exptover{(X_1,\cdots,X_L)\sim \bern(1/2)^{\ot L}}{\pl(X_1,\cdots,X_L)} ,
\end{align*}
where the notation $(X_1,\cdots,X_L)\sim \bern(1/2)^{\ot L}$ denotes that $L$ independent unbiased bits are sampled. 

It is then not difficult to conjecture the value for $p_*(q,\ell,L)$: if $\plur_\ell(x_1,\dots,x_L)$ denotes the top-$\ell$-plurality value of $x_1,\dots,x_L \in [q]$, i.e., $\plur_\ell(x_1,\dots,x_L) = \max_{\Sigma \subseteq [q]:|\Sigma|=\ell} |\{i \in [L]:x_i \in \Sigma\}|$, then it should be that 
\begin{align} \label{eq:intro-q-ary-plotkin}
    p_*(q,\ell,L) = 1 - \frac{1}{L}\exptover{(X_1,\cdots,X_L)\sim \unif([q])^{\ot L}}{\plur_\ell(X_1,\cdots,X_L)} .
\end{align}
For the case of $\ell=1$, i.e., $q$-ary list-decoding, a proof is claimed in~\cite{blinovsky-2005-ls-lb-qary,blinovsky-2008-ls-lb-qary-supplementary}; however, as we outline in \Cref{sec:discussion-blinovsky} this proof is flawed. In this work we provide a rigorous derivation of \Cref{eq:intro-q-ary-plotkin} for all values of $\ell$, $L$ and $q$ with $1\leq \ell \leq q$. 

More precisely, we obtain the following results:
\begin{itemize}
    \item A proof that $(p,\ell,L)_q$-list-recoverable $q$-ary codes with $p > p_*(q,\ell,L)$ have \emph{constant-size}, i.e., independent of $n$. This should be interpreted as a generalization of the Plotkin bound~\cite{plotkin-1960}, which states that binary codes uniquely-decodable from a $1/4+\eps$ fraction of errors have size at most $O(1/\eps)$. For this reason we call our result a ``Plotkin bound for list-recovery.''
    \item Adapting the Elias-Bassalygo argument~\cite{bassalygo-pit1965}, we subsequently derive upper bounds on the rate of $(p,\ell,L)_q$-list-recoverable $q$-ary codes when $p < p_*(q,\ell,L)$.
    \item To complement this, we show that there exist infinite families of positive rate $q$-ary codes that are $(p,\ell,L)_q$-list-recoverable whenever $p<p_*(q,\ell,L)$. We are therefore justified in calling $p_*(q,\ell,L)$ the zero-rate threshold for list-recovery. 
\end{itemize}
We now describe our techniques in more detail. 

\subsection{Our techniques}

\paragraph{Schur convexity of the function $f_{q,L,\ell}$} Following prior work~\cite{blinovsky-2005-ls-lb-qary},\footnote{In fact, \cite{blinovsky-2005-ls-lb-qary} only considers list-decoding, so a slight adaptation of this argument is required for list-recovery.} our task requires us to answer the following question. Consider the function on distributions $P$ over the alphabet $[q]$ defined as 
\begin{align*}
    f_{q,L,\ell}(P) &\coloneqq \exptover{(X_1,\cdots,X_L)\sim P^{\ot L}}{\plur_\ell(X_1,\cdots,X_L)} .
\end{align*}
Analogously to before, the notation $(X_1,\cdots,X_L)\sim P^{\ot L}$ means that $L$ independent samples are taken from the distribution $P$. A crucial ingredient for deriving the Plotkin bound is a demonstration that this function is minimized by the uniform distribution. 

There is a well-studied class of functions on finite distributions with the property that they are minimized by the uniform distribution: \emph{Schur convex} functions. These are the functions that are monotonically-increasing with respect to the \emph{majorization}-ordering, which compares vectors of real numbers by first sorting the vectors in descending order and then checking to see if all the prefix sums of one vector is greater than or equal to the prefix sums of the other. The important detail for us is that the uniform vector $(1/q,\dots,1/q) \in \mathbb R^q$, corresponding to the uniform distribution, is majorized by \emph{every} other vector corresponding to a distribution over $[q]$. 

To demonstrate the Schur convexity of this function, we use the Schur-Ostrowski criterion, which states that Schur-convexity is equivalent to the non-negativity of a certain expression involving partial derivatives. Showing that this expression is non-negative boils down to a combinatorial accounting game, where we can show that the positive contributions arising from certain terms exceed the negative contributions arising from others. 

\paragraph{Convexity of the univariate function $g_{q,L,\ell}$} Another important technical ingredient that we need for the proof of the Plotkin bound is the convexity of the univariate function
\[
    g_{q,L,\ell}(w) \coloneqq f_{q,L,\ell}(P_{q,\ell,w}),
\]
where the distribution $P_{q,\ell,w} = (p_1,\dots,p_q)$ is defined as 
\[
    p_i = \begin{cases}
        \frac{w}{q-\ell} & \text{if } i \leq q-\ell\\
        \frac{1-w}{\ell} & \text{if }i \geq q-\ell+1
    \end{cases} . 
\]
In order to show the function is convex, we prove the second derivative is non-negative. In differentiating, we use the expression for $g_{q,L,\ell}$ in terms of $f_{q,L,\ell}$ and apply the chain rule. Showing the resulting expression is positive is again a sort of combinatorial accounting game: we can show the positive terms contribute more than the negative terms. 

Quite interestingly, for $\ell = 1$ (i.e., the case relevant for list-decoding) we only prove the convexity of the function $f_{q,1,L}$ on the interval $[0,(q-1)/q]$. Fortunately, as we can also easily show that $g_{q,1,L}$ decreases on the interval $[0,(q-1)/q]$ and then increases on the interval $[(q-1)/q,1]$,\footnote{This is in fact an easy corollary of the Schur convexity of $f_{q,1,L}$.} convexity of $f_{q,1,L}$ on $[0,(q-1)/q]$ suffices for our purposes. And indeed, this is not an artifact of the proof: Blinovsky had already observed that convexity of $f_{q,1,L}$ does not hold on the entire interval $[0,1]$~\cite{blinovsky-2005-ls-lb-qary,blinovsky-2008-ls-lb-qary-supplementary}. However, for $\ell \geq 2$ we obtain that convexity of $f_{q,\ell,L}$ does indeed hold on the entire interval $[0,1]$. We note that the second derivative does behave qualitatively differently, so this is perhaps not too surprising in hindsight; we comment on this further in \Cref{rem:ell-2-convexity}. 


\paragraph{Plotkin bound} Armed with these (Schur-)convexity results, we aim to prove a Plotkin bound for list-decoding/-recovery. That is, if a $q$-ary code is $(p,\ell,L)_q$-list-recoverable with $p \geq p_*(q,\ell,L)+\eps$, how large can the code be? Following the template of the standard argument (although certain subtleties arise when generalizing to list-recovery), we can show that such a code must be of constant size, i.e., independent of $n$. 

Informally, the argument begins with a ``preprocessing step'' that prunes away some (but, crucially, not too many) codewords and yields a more structured subcode that we can subsequently analyze. We start with a code that is almost constant weight (if necessary, we may pass to a subcode of codewords whose weights lie in a narrow interval), and then pass to an almost ``equi-coupled'' subcode $\cC'$, which is proven to exist via Ramsey-theoretic results. Informally, equi-coupled means that every $L$-subset has roughly the same type, i.e., roughly the same empirical distribution of length $L$-vectors if we imagine creating an $L \times n$ matrix from the $L$ codewords in the subset and then randomly sampling a column. One can then argue that either this type is symmetric or the obtained subcode is already very small (and we are done). 

To analyze this subcode $\cC'$ (which we may now assume to be symmetric) we apply a double-counting argument to the average radius to cover $L$-subsets (where for list-recoverability, this radius is measured via the distance to a tuple of input lists). The lower bound on this quantity follows quite naturally from the list-decodability/-recoverability of the code, together with the equi-coupled property of the subcode. For the upper bound, we compute the radius of an $L$-subset in terms of the empirical distribution of a coordinate $k \in [n]$, i.e., each $x \in [q]$ is assigned probability mass $P_k(x) = \frac{1}{M}\sum_{\vec x \in\cC'} \indicator{x_k=x}$. By the Schur convexity of the function $f_{q,L,\ell}$ and the convexity of the univariate function $g_{q,L,\ell}$, we can bound this in terms of a distribution placing total mass $w \leq \frac{q-\ell}{q}$ on the last $\ell$ elements of $[q]$ and mass $\frac{1-w}{q-\ell}$ on each of the others. The result then follows.

We remark that, due to our use of Ramsey-theoretic arguments, the precise bound we obtain on the code size is quite poor. We have made no effort to optimize this constant. However, we do believe it would be interesting to improve this bound; we discuss this further in \Cref{sec:open-q}. 

\paragraph{Elias-Bassalygo-style bound} After deriving this Plotkin bound, a well-known argument template (typically attributed to Elias and Bassalygo~\cite{bassalygo-pit1965}) allows one to derive more general tradeoffs between the rate $R$ and the noise-resilience parameters $(p,\ell,L)_q$. Informally, this proceeds by covering the space $[q]^n$ by a bounded number of list-recovery balls. 
The radius of these balls is carefully chosen to allow one to apply the Plotkin bound to the subcodes obtained by taking the intersection of the code with these balls.
On the other hand, the number of list-recovery balls needed to cover $[q]^n$, known as the covering number, can be sharply estimated. 
From the above two bounds (the Plotkin bound and the covering number), a bound on the size of the whole code can be derived. 

\paragraph{Possibility result: random code with expurgation} To complement the Plotkin bound, we show that if the decoding radius $p$ is less than $p_*(q,\ell,L)$ then there exist infinite families of $(p,\ell,L)_q$-list-recoverable $q$-ary codes. This justifies our ``zero-rate threshold'' terminology for $p_*(q,\ell,L)$. The argument is completely standard, obtained by sampling a random code and subsequently expurgating.
In fact, the lower bound on achievable rate is derived from the exact large deviation exponent of a certain quantity known as the average radius (cf.\ \Cref{def:avg-rad-list-rec}) of a tuple of random vectors.
Therefore the bound holds under a stronger notion called average-radius list-recovery: namely, for any subset of $L$ codewords $\vec x_1,\dots,\vec x_j$ and any tuple of input lists $(\mathcal{Y}_1,\dots,\mathcal{Y}_n)$, we have 
\[
    \sum_{j=1}^L |\{i\in[n]:x_{j,i} \notin \mathcal{Y}_i\}| > Lpn \ .
\]

\subsection{Discussion on related work}
\label{sec:related-work}

\paragraph{Lower bounds for small $q$ and/or $L$}
For the case of $ (p,3)_2 $-list-decoding, it was shown in \cite[Theorem 6.1]{gmrsw-2020-sharp-threshold} that the \emph{threshold rate}\footnote{We warn the reader not not to confuse this concept with that of the zero-rate threshold.} of random binary \emph{linear} codes equals
\begin{align}
    \frac{1}{2}(2 - H_2(3p) - 3p\log_2(3)) . \label{eqn:gmrsw-lin-thr}
\end{align}
The term \emph{threshold} refers to the critical rate below which a random binary linear code is $ (p,3)_2 $-list-decodable with high probability and above which it is not with high probability. 
This result was recently extended to the following two cases \cite{resch-yuan-arxiv}. 
For $ (p,4)_2 $-list-decoding, the threshold rate of random binary linear code is lower bounded by \cite[Theorem 1.3]{resch-yuan-arxiv}
\begin{align}
    \frac{1}{3} \min_{\substack{x_1,x_2\ge0 \\ x_1 + 2x_2 \le 4p \\ x_1 + x_2 \le 1}} {3 - \eta_2(x_1,x_2) - 2x_1 - x_2\log_2(3)} . \label{eqn:ry-1} 
\end{align}
Here we use the notation 
\begin{align}
    \eta_q(x_1,\dots,x_t) &\coloneqq \sum_{i=1}^t x_i \log_q \frac{1}{x_i} + \paren{1 - \sum_{i = 1}^t x_i} \log_q\frac{1}{1 - \sum_{i = 1}^t x_i} \notag 
\end{align}
for a partial probability vector $ (x_1,\dots,x_t)\in\bbR_{\ge0}^t $ satisfying $ t\le q $ and $ x_1 + \dots + x_t\le1 $. 
Note that $ \eta_2(x) = H_2(x) $, however, this is no longer the case for $q>2$. 
Moreover, for $ (p,3)_q $-list-decoding, \cite[Theorem 1.5]{resch-yuan-arxiv} showed that the threshold rate of random linear code is at least
\begin{align}
    \frac{1}{2} \min_{\substack{x_1,x_2\ge0 \\ x_1 + 2x_2 \le 3p \\ x_1 + x_2 \le 1}} {2 - \eta_q(x_1,x_2) - x_1\log_q(3(q-1)) - x_2\log_q(q-1)(q-2)} . \label{eqn:ry-2} 
\end{align}
Our general lower bound (cf.\ \Cref{thm:lb}) for list-recovery (numerically) matches \Cref{eqn:gmrsw-lin-thr,eqn:ry-1,eqn:ry-2} upon particularizing the parameters $ q,\ell,L $ suitably. 
See \Cref{fig:gmrsw,fig:ry1,fig:ry2}. 
It is possible to analytically prove this observation, though we do not pursue it in the current paper. 
The rationale underlying this phenomenon is that the threshold rate of random linear codes for list-recovery is expected to match the rate achieved by random codes with expurgation (with the notable exception of zero-error list-recovery~\cite{glmrsw-2020-bounds-list-dec-rand-lin}). 
This conjecture, in its full generality, remains unproved, although it is partially justified in several recent works \cite{mosheiffetal-2019-ldpc,glmrsw-2020-bounds-list-dec-rand-lin,gmrsw-2020-sharp-threshold,resch-yuan-arxiv}. 

\paragraph{Hash codes}
One may note that for $\ell\ge2$, our upper and lower bounds typically exhibit a large gap even at $ p = 0 $. 
See \Cref{fig:this-hash1,fig:this-hash2,fig:this-bk-hash,fig:this-listrec}. 
We provide evidence below indicating that closing this gap is in general a rather challenging task and necessarily requires significantly new ideas. 
Let us focus on the vertical axis $p=0$, known as \emph{zero-error list-recovery}. 
We observe that some configurations of $q,\ell,L$ in this regime encode several longstanding open questions in combinatorics. 
Indeed, consider $\ell = q-1,L = q$. 
The $(0,q-1,q)_q$-list recoverability condition can then be written as: for any $ \cY_1,\cdots,\cY_n\in\binom{[q]}{q-1} $,
\begin{align}
    \card{\curbrkt{\vx\in\cC : \card{\curbrkt{j\in[n] : x_j\notin\cY_j}} = 0}} &\le L-1 , \notag 
\end{align}
i.e., 
\begin{align}
    \card{\curbrkt{\vx\in\cC : \forall j\in[n], \, x_j\in\cY_j}} &\le L-1 . \notag 
\end{align}
Taking the contrapositive, we note that this condition is further equivalent to: for any $ \{\vx_1,\cdots,\vx_L\}\in\binom{\cC}{L} $, there exists $j\in[n]$ such that 
$ \card{\curbrkt{x_{1,j},\cdots,x_{L,j}}} = q $. 
In words, for any $ q $-tuple of codewords in a $(0,q-1,q)_q$-list-recoverable code, there must exist one coordinate such that the corresponding $q$-ary symbols in the tuple are all distinct. 
Such a code is also known as a \emph{$q$-hashing} in combinatorics. 
It is well-known \cite{FK,K86} that a probabilistic construction yields such codes of rate\footnotemark{} at least 
\begin{align}
    C_{(0,q-1,q)_q} &\ge \frac{1}{q-1}\log_q\frac{1}{1 - \frac{q!}{q^q}} . \label{eqn:fredman-komlos-lower} 
\end{align}
\footnotetext{The bounds in \cite{KM,FK} are slightly adjusted so that they are consistent with our definition of code rate which adopts a $ \log_q $ normalization (cf.\ \Cref{def:cap}).}
In the same paper \cite{FK} also proved an upper bound 
\begin{align}
    C_{(0,q-1,q)_q} &\le \frac{q!}{q^{q-1}}\log_q(2) . \label{eqn:fredman-komlos-upper}
\end{align}
Another upper bound 
\begin{align}
    C_{(0,q-1,q)_q} &\le \log_q\frac{q}{q - 1}  \label{eqn:korner-marton-upper}
\end{align}
can be proved using either a double-counting argument (a.k.a.\ first moment method), or (hyper)graph entropy \cite{K86,KM,korner-graph-ent}. 
\Cref{eqn:fredman-komlos-upper} is much better than \Cref{eqn:korner-marton-upper} for $q\ge4$. 
However, the latter bound $\log_3\frac{3}{2}$ remains the best known for $q = 3$ (called the \emph{trifference problem} by \korner). 
For larger $q$, both lower \cite{XY21} and upper bounds \cite{Arikan,DGR17,GR,CD20,dalai-bk-hashing-2022} can be improved. 
However, improving the bound for $q=3$ is recognized as a formidable challenge. 
We will show in \Cref{rk:lb-lr-hashing} in \Cref{sec:list-rec-lb} that our lower bound for list-recovery (cf.\ \Cref{thm:lb}) recovers \Cref{eqn:fredman-komlos-lower} for $q$-hashing upon setting $ \ell = q-1,L=q $. 
Furthermore, our upper bound \Cref{thm:eb-list-rec-main} recovers \Cref{eqn:korner-marton-upper} for $q$-hashing (cf.\ \Cref{rk:eb-hashing}). 

A generalization of $q$-hashing known as \emph{$(q,L)$-hashing} ($q\ge L$) can also be cast as zero-error list-recoverable codes with more general values of $\ell,L$. 
Indeed, taking $ L = \ell + 1 $ and $ \ell\le q-1 $, we can write $ (0,\ell,\ell+1)_q $-list-recoverability alternatively as: for any $ \{\vx_1,\cdots,\vx_{\ell+1}\}\in\binom{\cC}{\ell+1} $, there exists $ j\in[n] $ such that $ \card{\curbrkt{x_{1,j},\cdots,x_{\ell+1,j}}} = \ell+1 $. 
This is in turn the precise definition of $(q,\ell+1)$-hashing. 
It can be immediately seen that $(q,q)$-hashing is nothing but $q$-hashing. 
The upper and lower bounds in \cite{FK} also extend to $(q,\ell+1)$-hashing and read as follows:
\begin{align}
    \frac{1}{\ell}\log_q\frac{1}{1 - \frac{\binom{q}{\ell+1} (\ell+1)!}{q^{\ell+1}}} \le 
    C_{(0,\ell,\ell+1)_q} &\le \frac{\binom{q}{\ell}\ell!}{q^{\ell}} \log_q(q - \ell + 1) . \label{eqn:fredman-komlos-bk}
\end{align}
Our lower bound for list-recovery in \Cref{thm:lb} also recovers the above lower bound for $(q,\ell+1)$-hashing by \cite{FK} upon setting $L = \ell+1$ (see \Cref{rk:list-rec-lb-recover}). 
The upper bound was later improved in \cite{KM} for $q>L$ using the notion of hypergraph entropy:
\begin{align}
    C_{(0,\ell,\ell+1)_q} &\le \min_{0\le j\le \ell-1} \frac{\binom{q}{j+1}(j+1)!}{q^{j+1}} \log_q\frac{q-j}{\ell-j} , \label{eqn:korner-marton-upper-bk}
\end{align}
though it coincides with \Cref{eqn:fredman-komlos-upper} when $\ell = q-1$. 
Some improved upper bounds in \cite{GR,dalai-bk-hashing-2022} apply to $(q,\ell+1)$-hashing as well. 
To the best of our knowledge, no improvement on lower bounds is known for $ \ell<q-1 $. 




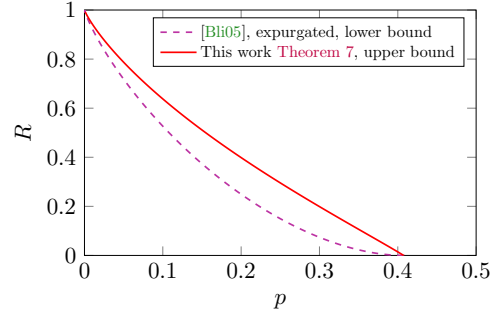
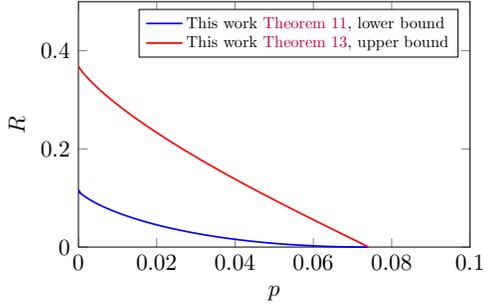
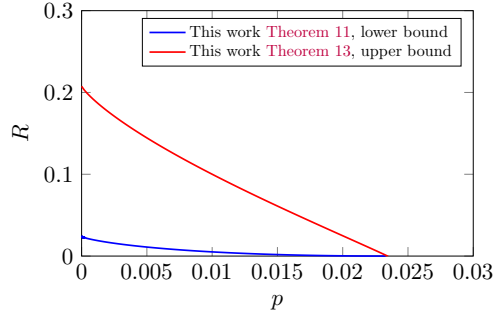
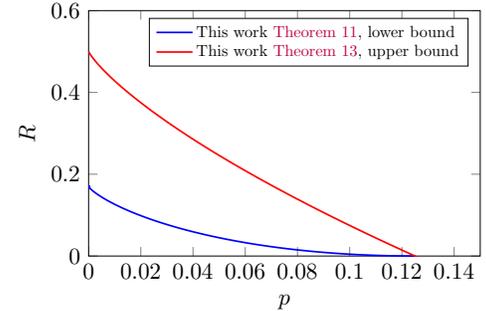
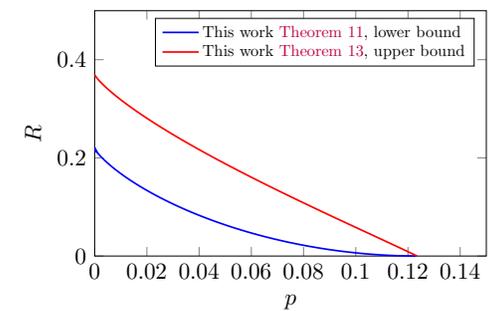
\begin{figure}[htbp]
	\centering
	\begin{subfigure}[b]{0.49\linewidth}
	\centering
	\begin{tikzpicture}[thick,scale=0.8]
	\pgfplotsset{compat = 1.3}
	\begin{axis}[
		legend style={nodes={scale=0.7, transform shape}},
		legend cell align={left},
		width = \columnwidth,
		height = 0.7\columnwidth,
		xlabel = {$p$},
		xlabel style={nodes={scale=0.8, transform shape}},
		ylabel = {$R$},
		ylabel style={nodes={scale=0.8, transform shape}},
		xmin = 0,
		xmax = 0.3,
		ymin = 0.0,
		ymax = 1.0,
		legend pos = north east,
		y tick label style={
	        /pgf/number format/.cd,
	            fixed,
	            fixed,
	            precision=2,
	        /tikz/.cd
	    },
	    x tick label style={
	        /pgf/number format/.cd,
	            fixed,
	            fixed,
	            precision=2,
	        /tikz/.cd
	    }
	]

	\addplot[color=cyan, mark=none, thick] table {code-and-data/lb_rate_gmrsw_bin_2.txt};
	\addlegendentry{\cite{gmrsw-2020-sharp-threshold}, linear, threshold};

	\addplot[color=byzantine, mark=none, thick, dashed] table {code-and-data/lb_rate_bli_bin_even_2.txt};
	\addlegendentry{\cite{blinovsky-1986-ls-lb-binary}, expurgated, lower bound};

	\addplot[color=orange, mark=none, thick, dashed] table {code-and-data/ub_rate_bli_bin_2.txt};
	\addlegendentry{\cite{blinovsky-1986-ls-lb-binary}, upper bound};


	\end{axis}
	\end{tikzpicture}
    \caption{$ q=2,\ell=1,L=3 $ (list-decoding).}
	\label{fig:gmrsw}
	\end{subfigure}
	\begin{subfigure}[b]{0.49\linewidth}
    \centering
    \begin{tikzpicture}[thick,scale=0.8]
	\pgfplotsset{compat = 1.3}
	\begin{axis}[
		legend style={nodes={scale=0.7, transform shape}},
		legend cell align={left},
		width = \columnwidth,
		height = 0.7\columnwidth,
		xlabel = {$p$},
		xlabel style={nodes={scale=0.8, transform shape}},
		ylabel = {$R$},
		ylabel style={nodes={scale=0.8, transform shape}},
		xmin = 0,
		xmax = 0.4,
		ymin = 0.0,
		ymax = 1.0,
		legend pos = north east,
		y tick label style={
	        /pgf/number format/.cd,
	            fixed,
	            fixed,
	            precision=2,
	        /tikz/.cd
	    },
	    x tick label style={
	        /pgf/number format/.cd,
	            fixed,
	            fixed,
	            precision=2,
	        /tikz/.cd
	    }
	]

	\addplot[color=green, mark=none, thick] table {code-and-data/lb_rate_ry_bin_3.txt};
	\addlegendentry{\cite{resch-yuan-arxiv}, linear, lower bound};

	\addplot[color=byzantine, mark=none, thick, dashed] table {code-and-data/lb_rate_bli_bin_odd_3.txt};
	\addlegendentry{\cite{blinovsky-1986-ls-lb-binary}, expurgated, lower bound};

	\addplot[color=orange, mark=none, thick, dashed] table {code-and-data/ub_rate_bli_bin_3.txt};
	\addlegendentry{\cite{blinovsky-1986-ls-lb-binary}, upper bound};


	\end{axis}
	\end{tikzpicture}
	\caption{$ q=2,\ell=1,L=4 $ (list-decoding).}
	\label{fig:ry1}
    \end{subfigure}
    \begin{subfigure}[b]{0.49\linewidth}
    \centering
    \begin{tikzpicture}[thick,scale=0.8]
	\pgfplotsset{compat = 1.3}
	\begin{axis}[
		legend style={nodes={scale=0.7, transform shape}},
		legend cell align={left},
		width = \columnwidth,
		height = 0.7\columnwidth,
		xlabel = {$p$},
		xlabel style={nodes={scale=0.8, transform shape}},
		ylabel = {$R$},
		ylabel style={nodes={scale=0.8, transform shape}},
		xmin = 0,
		xmax = 0.4,
		ymin = 0.0,
		ymax = 1.0,
		legend pos = north east,
		y tick label style={
	        /pgf/number format/.cd,
	            fixed,
	            fixed,
	            precision=2,
	        /tikz/.cd
	    },
	    x tick label style={
	        /pgf/number format/.cd,
	            fixed,
	            fixed,
	            precision=2,
	        /tikz/.cd
	    }
	]

	\addplot[color=green, mark=none, thick] table {code-and-data/lb_rate_ry_q_2.txt};
	\addlegendentry{\cite{resch-yuan-arxiv}, linear, lower bound}; 

	\addplot[color=byzantine, mark=none, thick, dashed] table {code-and-data/lb_rate_bli_q3_2.txt};
	\addlegendentry{\cite{blinovsky-2005-ls-lb-qary}, expurgated, lower bound};


	\addplot[color=orange, mark=none, thick, dashed] table {code-and-data/ub_rate_bli_q3_2.txt}; 
	\addlegendentry{\cite{blinovsky-2005-ls-lb-qary}, upper bound};

	\end{axis}
	\end{tikzpicture}
	\caption{$ q=3,\ell=1,L=3 $ (list-decoding).}
	\label{fig:ry2}
    \end{subfigure}
    \begin{subfigure}[b]{0.49\linewidth}
    \centering
    \begin{tikzpicture}[thick,scale=0.8]
	\pgfplotsset{compat = 1.3}
	\begin{axis}[
		legend style={nodes={scale=0.7, transform shape}},
		legend cell align={left},
		width = \columnwidth,
		height = 0.7\columnwidth,
		xlabel = {$p$},
		xlabel style={nodes={scale=0.8, transform shape}},
		ylabel = {$R$},
		ylabel style={nodes={scale=0.8, transform shape}},
		xmin = 0,
		xmax = 0.5,
		ymin = 0.0,
		ymax = 1.0,
		legend pos = north east,
		y tick label style={
	        /pgf/number format/.cd,
	            fixed,
	            fixed,
	            precision=2,
	        /tikz/.cd
	    },
	    x tick label style={
	        /pgf/number format/.cd,
	            fixed,
	            fixed,
	            precision=2,
	        /tikz/.cd
	    }
	]

	\addplot[color=byzantine, mark=none, thick, dashed] table {code-and-data/lb_rate_bli_q3_3.txt};
	\addlegendentry{\cite{blinovsky-2005-ls-lb-qary}, expurgated, lower bound};

	\addplot[color=red, mark=none, thick] table {code-and-data/ub_rate_bli_q3_3.txt}; 
	\addlegendentry{This work \Cref{thm:eb-qary-list-dec-main}, upper bound}; 

	\end{axis}
	\end{tikzpicture}
	\caption{$ q=3,\ell=1,L=4 $ (list-decoding).}
	\label{fig:bli-qary}
    \end{subfigure}
    \begin{subfigure}[b]{0.49\linewidth}
    \centering
    \begin{tikzpicture}[thick,scale=0.8]
	\pgfplotsset{compat = 1.3}
	\begin{axis}[
		legend style={nodes={scale=0.7, transform shape}},
		legend cell align={left},
		width = \columnwidth,
		height = 0.7\columnwidth,
		xlabel = {$p$},
		xlabel style={nodes={scale=0.8, transform shape}},
		ylabel = {$R$},
		ylabel style={nodes={scale=0.8, transform shape}},
		xmin = 0,
		xmax = 0.1,
		ymin = 0.0,
		ymax = 0.5,
		legend pos = north east,
		y tick label style={
	        /pgf/number format/.cd,
	            fixed,
	            fixed,
	            precision=2,
	        /tikz/.cd
	    },
	    x tick label style={
	        /pgf/number format/.cd,
	            fixed,
	            fixed,
	            precision=2,
	        /tikz/.cd
	    }
	]

	\addplot[color=blue, mark=none, thick] table {code-and-data/lb_rate_list_rec_q3_ls2_ell2.txt}; 
	\addlegendentry{This work \Cref{thm:lb-listrec}, lower bound}; 

	\addplot[color=red, mark=none, thick] table {code-and-data/ub_rate_list_rec_q3_ls2_ell2.txt};
	\addlegendentry{This work \Cref{thm:eb-list-rec-main}, upper bound};

	\end{axis}
	\end{tikzpicture}
	\caption{$ q=3,\ell=2,L=3 $ (list-recovery, $3$-hashing).}
	\label{fig:this-hash1}
    \end{subfigure}
    \begin{subfigure}[b]{0.49\linewidth}
    \centering
    \begin{tikzpicture}[thick,scale=0.8]
	\pgfplotsset{compat = 1.3}
	\begin{axis}[
		legend style={nodes={scale=0.7, transform shape}},
		legend cell align={left},
		width = \columnwidth,
		height = 0.7\columnwidth,
		xlabel = {$p$},
		xlabel style={nodes={scale=0.8, transform shape}},
		ylabel = {$R$},
		ylabel style={nodes={scale=0.8, transform shape}},
		xmin = 0,
		xmax = 0.03,
		ymin = 0.0,
		ymax = 0.3,
		legend pos = north east,
		y tick label style={
	        /pgf/number format/.cd,
	            fixed,
	            fixed,
	            precision=2,
	        /tikz/.cd
	    },
	    x tick label style={
	        /pgf/number format/.cd,
	            fixed,
	            fixed,
	            precision=3,
	        /tikz/.cd
	    },
	    scaled x ticks = false
	]

	\addplot[color=blue, mark=none, thick] table {code-and-data/lb_rate_list_rec_q4_ls3_ell3.txt}; 
	\addlegendentry{This work \Cref{thm:lb-listrec}, lower bound}; 

	\addplot[color=red, mark=none, thick] table {code-and-data/ub_rate_list_rec_q4_ls3_ell3.txt};
	\addlegendentry{This work \Cref{thm:eb-list-rec-main}, upper bound};

	\end{axis}
	\end{tikzpicture}
	\caption{$ q=4,\ell=3,L=4 $ (list-recovery, $4$-hashing).}
	\label{fig:this-hash2}
    \end{subfigure}
    \begin{subfigure}[b]{0.49\linewidth}
    \centering
    \begin{tikzpicture}[thick,scale=0.8]
	\pgfplotsset{compat = 1.3}
	\begin{axis}[
		legend style={nodes={scale=0.7, transform shape}},
		legend cell align={left},
		width = \columnwidth,
		height = 0.7\columnwidth,
		xlabel = {$p$},
		xlabel style={nodes={scale=0.8, transform shape}},
		ylabel = {$R$},
		ylabel style={nodes={scale=0.8, transform shape}},
		xmin = 0,
		xmax = 0.15,
		ymin = 0.0,
		ymax = 0.6,
		legend pos = north east,
		y tick label style={
	        /pgf/number format/.cd,
	            fixed,
	            fixed,
	            precision=2,
	        /tikz/.cd
	    },
	    x tick label style={
	        /pgf/number format/.cd,
	            fixed,
	            fixed,
	            precision=2,
	        /tikz/.cd
	    }
	]

	\addplot[color=blue, mark=none, thick] table {code-and-data/lb_rate_list_rec_q4_ls2_ell2.txt}; 
	\addlegendentry{This work \Cref{thm:lb-listrec}, lower bound}; 

	\addplot[color=red, mark=none, thick] table {code-and-data/ub_rate_list_rec_q4_ls2_ell2.txt};
	\addlegendentry{This work \Cref{thm:eb-list-rec-main}, upper bound};

	\end{axis}
	\end{tikzpicture}
	\caption{$ q=4,\ell=2,L=3 $ (list-recovery, $(4,3)$-hashing).}
	\label{fig:this-bk-hash}
    \end{subfigure}
    \begin{subfigure}[b]{0.49\linewidth}
    \centering
    \begin{tikzpicture}[thick,scale=0.8]
	\pgfplotsset{compat = 1.3}
	\begin{axis}[
		legend style={nodes={scale=0.7, transform shape}},
		legend cell align={left},
		width = \columnwidth,
		height = 0.7\columnwidth,
		xlabel = {$p$},
		xlabel style={nodes={scale=0.8, transform shape}},
		ylabel = {$R$},
		ylabel style={nodes={scale=0.8, transform shape}},
		xmin = 0,
		xmax = 0.15,
		ymin = 0.0,
		ymax = 0.5,
		legend pos = north east,
		y tick label style={
	        /pgf/number format/.cd,
	            fixed,
	            fixed,
	            precision=2,
	        /tikz/.cd
	    },
	    x tick label style={
	        /pgf/number format/.cd,
	            fixed,
	            fixed,
	            precision=2,
	        /tikz/.cd
	    }
	]

	\addplot[color=blue, mark=none, thick] table {code-and-data/lb_rate_list_rec_q3_ls4_ell2.txt}; 
	\addlegendentry{This work \Cref{thm:lb-listrec}, lower bound}; 

	\addplot[color=red, mark=none, thick] table {code-and-data/ub_rate_list_rec_q3_ls4_ell2.txt};
	\addlegendentry{This work \Cref{thm:eb-list-rec-main}, upper bound};

	\end{axis}
	\end{tikzpicture}
	\caption{$ q=3,\ell=2,L=5 $ (list-recovery).}
	\label{fig:this-listrec}
    \end{subfigure}
	
	\caption{Plots of upper and lower bounds in \cite{blinovsky-1986-ls-lb-binary,blinovsky-2005-ls-lb-qary,gmrsw-2020-sharp-threshold,resch-yuan-arxiv} and this work for various values of $q\ge2,1\le\ell\le q-1,L\ge2$. }
	\label{fig:all}
\end{figure}

\paragraph{Zero-rate thresholds for general adversarial channels}
The problem of locating the zero-rate threshold has been addressed in a much more general context \cite{zhang-2019-list-dec-general}. 
The results in \cite{zhang-2019-list-dec-general} on \emph{general adversarial channel} model can be specialized to the list-recovery setting and read as follows. 
Given $ q,p,\ell,L $, define the \emph{confusability set} $\cK_{(p,\ell,L)_q}$ as the set of types\footnotemark{} (cf.\ \Cref{def:type}) of all ``confusable'' $L$-tuple of codewords in the sense that they can fit into a certain list-recovery ball (cf.\ \Cref{def:metric-list-rec}) of radius $np$. 
\footnotetext{More precisely, the confusability set is the \emph{closure} of the set of types of all confusable codeword tuples, since types are dense in distributions.}
Specifically, 
\begin{align}
    \cK_{(p,\ell,L)_q} &\coloneqq \curbrkt{\sum_{\cY\in\binom{[q]}{\ell}} P_{X_1,\cdots,X_L,Y = \cY}\in\Delta\paren{[q]^L} : 
    \begin{array}{l}
        P_{X_1,\cdots,X_L,Y} \in \Delta\paren{[q]^L\times\binom{[q]}{\ell}} \\
        \forall i\in[L] , 
        \displaystyle\sum_{\substack{(x,\cY)\in[q]\times\binom{[q]}{\ell} \\ x\notin\cY}} P_{X_i, Y}(x,\cY) \le p
    \end{array}} . \notag 
\end{align}
In the above definition, we use the notation $ \sum_{b} P_{A,B = b} $ to denote the marginalization of $ P_{A,B} $ onto the first variable $A$, and use $ P_{X_i,Y} $ to denote the marginal of $ P_{X_1,\cdots,X_i,Y} $ on $ (X_i,Y) $. 
It is not hard to verify that the confusability set is $(i)$ ``increasing'' in $p$ in the sense that $ \cK_{(p,\ell,L)_q}\subset\cK_{(p',\ell,L)_q} $ if $ p\le p' $, and $(ii)$ convex. 
Define also the convex cone of \emph{completely positive (CP) tensors} of order $L$, i.e., tensors that can be written as a sum of \emph{element-wise non-negative} rank-one tensors:
\begin{align}
    \cp_q^{\ot L} &\coloneqq \curbrkt{\sum_{i = 1}^k \vp_i^{\ot L} \in (\bbR_{\ge0}^q)^L : k\in\bbZ_{\ge1}, (\vp_1,\cdots,\vp_k)\in(\bbR_{\ge0}^q)^k} . \notag 
\end{align}
It is proved in \cite{zhang-2019-list-dec-general} that the zero-rate threshold $ p_*(q,\ell,L) $ can be expressed as the smallest $p$ such that all completely positive distributions are confusable:
\begin{align}
    p_*(q,\ell,L) &= \inf\curbrkt{p\in[0,1] : \cp_q^{\ot L}\cap\Delta([q]^L)\subset\cK_{(p,\ell,L)_q}} . \label{eqn:cp-char}
\end{align}
The above characterization is \emph{single-letter} in the sense that it is independent of the blocklength $n$. 
For $ q,\ell,L $ independent of $n$ (which is assumed to be the case in the current paper), the optimization problem on the RHS of \Cref{eqn:cp-char} can be solved in constant time. 
However, it does not immediately provide an explicit formula of $ p_*(q,\ell,L) $ and analytically solving the optimization problem does not appear easy to the authors. 
On the other hand, the characterization $p_*(q,\ell,L) = 1 - \frac{1}{L}\expt{\plur_\ell(X_1,\cdots,X_L)}$ (where the expectation is over $(X_1,\cdots,X_L)\sim\unif([q])^{\ot L}$, cf.\ \Cref{eq:intro-q-ary-plotkin}) in this paper can be seen as the explicit solution to the optimization problem, though the way it is obtained is \emph{not} by solving the latter problem per se. 
Instead, we prove the characterization from the first principle by leveraging specific structures of list-recovery. 
We hope that our characterization can shed light on the geometry of the high-dimensional polytopes -- the confusability set and the set of CP distributions -- involved in the characterization in \Cref{eqn:cp-char}. 

\paragraph{Organization} 
We begin with the necessary preliminaries in \Cref{sec:prelim}, and subsequently state our main results in \Cref{sec:main-results}. 
\Cref{sec:schu-convex-fql} through~\Cref{sec:eb-list-rec-pf} contain proofs of our main technical results. \Cref{sec:list-rec-lb} derives a lower bound on list-recovery capacity via a random coding argument. We summarize our results and state open problems in \Cref{sec:open-q} and provide acknowledgements in \Cref{sec:ack}. 

\section{Preliminaries}
\label{sec:prelim}

\subsection{Notation}
\label{sec:notation}
We use $ \Delta(\Sigma) $ to denote the probability simplex over a finite set $ \Sigma $, whose elements are typically denoted by $P$, perhaps with subscripts. In other cases, 
vectors are denoted by \textbf{boldface} letters, e.g., $ \va $, and the $i$-th entry of $ \va $ is denoted by $ a_i $. 
The symmetric group of order $n$ is denoted by $ S_n $. 
For a finite set $ \Sigma $ and an integer $ 0\le k\le|\Sigma| $, denote by
\begin{align}
\binom{\Sigma}{k} &\coloneqq \curbrkt{\Pi\subset\Sigma:|\Pi| = k} \notag 
\end{align}
the collection of all size-$k$ subsets of $ \Sigma $. 
We use $ \exp(x) $ and $ \exp_b(x) $ (where $b\in\bbR$) to denote $ e^x $ and $ b^x $, respectively. 
For a positive integer $k$, we use $[k]$ to denote the set $\{1,2,\ldots,k\}$.

\subsection{Basic definitions}
\label{sec:basic-def}

Throughout the paper, $ q\in\bbZ_{\ge2} $ denotes the alphabet size, $ L - 1\in\bbZ_{\ge1} $ denotes the list size, 
$ n\in\bbZ_{\ge1} $ denotes the blocklength.
The parameter $ \ell $ used in list-recovery takes integer values between $ 1 $ and $q-1$. 
The $\ell=1$ case reduces to list-decoding and the $\ell=q$ case is trivial (cf.\ \Cref{def:list-rec}).

\subsubsection{List-decodability and list-recoverability}
\label{sec:def-list-dec-rec}

\begin{definition}[Hamming metric]
\label{def:hamming}
We equip the space $ [q]^n $ with the \emph{Hamming metric}:
\begin{align}
\disth(\vx,\vy) &\coloneqq \sum_{i = 1}^n \indicator{x_i \ne y_i} \notag 
\end{align}
for $ \vx,\vy\in[q]^n $. 
The \emph{Hamming weight} of a vector $ \vx\in[q]^n $ is defined as 
\begin{align}
    \wth(\vx) \coloneqq \disth(\vx,\vq) , \notag
\end{align}
where $\vq = (q,q,\dots,q) \in [q]^n$ is the all-$q$ vector.

A \emph{Hamming ball} or \emph{Hamming sphere} centered around $ \vy\in[q]^n $ of radius $ 0\le r\le n $ is defined as 
\begin{align}
\bham(\vy,r) &\coloneqq \curbrkt{\vx\in[q]^n: \disth(\vx,\vy)\le r} . \notag 
\end{align}
or 
\begin{align}
\sham(\vy,r) &\coloneqq \curbrkt{\vx\in[q]^n: \disth(\vx,\vy)=r} , \notag 
\end{align}
respectively.
\end{definition}

\begin{remark}
    Our definition of Hamming weight is nonstandard; typically, the weight is defined to be the distance from the all-$0$ string. However, as our arguments regularly work with probability distributions over the alphabet, and in said arguments it is convenient for us to view these probability distributions as vectors in real-space, it is notationally simplest for us to assume the alphabet is $[q]=\{1,2,\dots,q\}$. We therefore have chosen to make $q$ the ``distinguished'' symbol from the alphabet, i.e., it takes the role that $0$ typically plays. (Note further that as we are not concerned with linear codes, there is no mathematical reason that we should have $0$ playing this distinguished role.)
\end{remark}

\begin{definition}[List-decodability]
\label{def:list-dec}
A \emph{code} $ \cC\subset[q]^n $ is said to be \emph{$ (p,L)_q $-list-decodable} if for every $ \vy\in[q]^n $, $ \card{\cC\cap\bham(\vy,np)} <L  $. 
\end{definition}

\begin{remark}
When we talk about \emph{a} code $ \cC\subset[q]^n $, we always mean \emph{a sequence} of codes $ \{\cC_n\}_n $ for an infinite increasing sequence of $ n\in\bbZ_{\ge1} $. 
\end{remark}

Analogous definitions can be made for list-recovery. 

\begin{definition}
\label{def:metric-list-rec}
Given a vector $\vx \in [q]^n$ and a tuple of sets $\vcY \in \binom{[q]}{\ell}^n$ for $1 \leq \ell \leq q-1$, we define
\begin{align}
    \distlr(\vx,\vcY) &\coloneqq \sum_{i=1}^n \indicator{x_i \notin \cY_i} , \quad 
    \wtlr(\vx) \coloneqq \distlr(\vx,\vcY_\ell) , 
    \notag
\end{align}
where $ \vcY_\ell\coloneqq\{q-\ell+1,\dots,q\}^n\in\binom{[q]}{\ell}^n $. 
We also define the list-recovery ball and sphere centered at $\vcY \in \binom{[q]}{\ell}^n$ of radius $ 0\le r\le n $ as
\begin{align}
    \blr(\vcY,r) &\coloneqq \curbrkt{\vx\in[q]^n: \distlr(\vx,\vcY)\le r} , \quad 
    \slr(\vcY,r) \coloneqq \curbrkt{\vx\in[q]^n: \distlr(\vx,\vcY)= r} ,
    \notag 
\end{align}
respectively. 
\end{definition}

\begin{definition}[List-recoverability]
\label{def:list-rec}
A code $ \cC\subset[q]^n $ is said to be \emph{$ (p,\ell,L)_q $-list-recoverable} if for every $\vcY \in \binom{[q]}{\ell}^n$, $\card{\cC \cap \blr(\vcY,np)} < L$.
\end{definition}

\subsubsection{Radius, capacity and zero-rate threshold}
\label{sec:def-rad}

\begin{definition}[Radius]
\label{def:list-dec-rad}
The \emph{radius} of an $L$-set of vectors $ \vx_1,\dots,\vx_L\in[q]^n $ is defined as the radius of the smallest Hamming ball containing the set $\{\vx_1,\dots,\vx_L\}$:
\begin{align}
    \rad(\vx_1,\dots,\vx_L) &\coloneqq \min_{\vy\in[q]^n} \max_{i\in[L]} \disth(\vx_i,\vy) . \label{eqn:def-rad} 
\end{align}
The \emph{$L$-radius} of a code $ \cC\subset[q]^n $ of size at least $ L $ is defined as the maximum radius of $L$-lists of codewords in $\cC$:
\begin{align}
    \rad_L(\cC) &\coloneqq \max_{\cL\in\binom{\cC}{L}} \rad(\cL) . \label{eqn:def-rad-code} 
\end{align}
\end{definition}

In fact, the radius of a code characterizes its list-decodability. 
Indeed, inspecting \Cref{def:list-dec} and \Cref{eqn:def-rad-code}, one observes that a code $ \cC\subset[q]^n $ is $ (p,L)_q $-list-decodable if and only if $ \rad_L(\cC)>np $. 
A quantity of central interest in the study of combinatorial list-decoding is the \emph{list-decoding capacity} defined as follows. 

\begin{definition}[Rate, list-decoding capacity and zero-rate threshold]
\label{def:cap}
The \emph{rate} of a code $ \cC\subset[q]^n $ is defined as $ R(\cC)\coloneqq\frac{1}{n}\log_q|\cC| $. 
The \emph{$(p,L)_q$-list-decoding capacity} is defined as the maximum rate of $ (p,L)_q $-list-decodable codes:
\begin{align}
    C_{(p,L)_q} &\coloneqq \limsup_{n\to\infty} \max_{\substack{\cC_n\subset[q]^n \\ \rad_L(\cC_n)>np}} R(\cC_n) . \notag 
\end{align}
Since $ C_{(p,L)_q} $ is a decreasing function in $ p $, we define the \emph{zero-rate threshold} for $ (p,L)_q $-list-decoding as the maximum $p$ such that the capacity is strictly positive:
\begin{align}
    p_*(q,L) &\coloneqq \sup\curbrkt{p\ge0 : C_{(p,L)_q}>0} . \notag 
\end{align}
\end{definition}

Again, we extend the above definitions to list-recovery. 
Let us start with \emph{$\ell$-radius} for list-recovery.

\begin{definition}[$\ell$-radius]
\label{def:list-rec-rad}
The \emph{$\ell$-radius} of an $L$-set of vectors $ \vx_1,\dots,\vx_L\in[q]^n $ is defined as the radius of the smallest list-recovery ball containing the set $\{\vx_1,\dots,\vx_L\}$:
\begin{align}
    \rad_\ell(\vx_1,\dots,\vx_L) &\coloneqq \min_{\vcY\in\binom{[q]}{\ell}^n} \max_{i\in[L]} \distlr(\vx_i,\vcY) . \label{eqn:def-rad-lr} 
\end{align}
The \emph{$(\ell,L)$-radius} of a code $ \cC\subset[q]^n $ of size at least $ L $ is defined as the maximum $\ell$-radius of $L$-lists of codewords in $\cC$:
\begin{align}
    \rad_{\ell,L}(\cC) &\coloneqq \max_{\cL\in\binom{\cC}{L}} \rad_\ell(\cL) . \notag 
\end{align}
\end{definition}

Capacity and zero-rate threshold for list-recovery are similarly defined by noting that $(p,\ell,L)_q$-list-recoverability of any $\cC\subset[q]^n$ is equivalent to $ \rad_{\ell,L}(\cC)>np $. 

\begin{definition}[Capacity and zero-rate threshold for list-recovery]
\label{def:cap-lr}
The \emph{$ (p,\ell,L)_q $-list recovery capacity} is defined as the largest rate of any $ (p,\ell,L)_q $-list-recoverable code:
\begin{align}
    C_{(p,\ell,L)_q} &\coloneqq \limsup_{n\to\infty} \max_{\substack{\cC_n\subset[q]^n \\ \rad_{\ell,L}(\cC_n)>np}} R(\cC_n) , \notag
\end{align}
and the \emph{zero-rate threshold} is defined to be the value of $p$ at which $ C_{(p,\ell,L)_q} $ hits zero:
\begin{align}
    p_*(q,\ell,L) &\coloneqq \sup\curbrkt{p\ge0 : C_{(p,\ell,L)_q} > 0} . \notag 
\end{align}
\end{definition}

\subsubsection{Average radius}
\label{sec:def-avg-rad}

Tightly related to the above definition of radius is another widely-studied notion called \emph{average radius}. 
To introduce this concept, let us first define the \emph{plurality} of a sequence of symbols. 

\begin{definition}[Plurality]
\label{def:plur}
For $ L\in\bbZ_{\ge2} $ and $ x_1,\dots,x_L\in[q] $, the \emph{plurality} of $ x_1,\dots,x_L $ is defined to be the number of times the most frequent symbol appears:
\begin{align}
\plur(x_1,\dots,x_L) &\coloneqq \max_{a\in[q]} \sum_{i=1}^L \indicator{x_i = a} . \notag 
\end{align}
The maximizer is denoted by $ \argplur(x_1,\cdots,x_L)\in[q] $. 
\end{definition}

The notion of \emph{average radius} can then be defined using plurality. 

\begin{definition}[Average radius]
\label{def:avg-rad-list-dec}
The \emph{average radius} of an $L$-set of vectors $ \vx_1,\dots,\vx_L\in[q]^n $ is defined as 
\begin{align}
    \ol{\rad}(\vx_1,\dots,\vx_L) &\coloneqq \sum_{j=1}^n \paren{1 - \frac{1}{L}\plur(x_{1,j},\dots,x_{L,j})} . \label{eqn:def-avg-rad}
\end{align}
\end{definition}

We pause to explain the motivation for defining average radius in the above particular way. It is well-known in the literature and not hard to check that the average radius defined in \Cref{eqn:def-avg-rad} is the solution to the following optimization problem:
\begin{align}
    \min_{\vy\in[q]^n} \exptover{I\sim\unif([L])}{\disth(\vx_I,\vy)} , \label{eqn:relax-ld}
\end{align}
and the minimizer is given by $ \vy^*\in[q]^n $ defined as $ y_i^* = \argplur(x_{1,i},\dots,x_{L,i}) $ for every $i\in[n]$. 
\Cref{eqn:relax-ld} is in turn a lower bound on the standard radius in \Cref{eqn:def-rad} where the inner maximum is replaced with average over $ i\in[L] $. 
This is also the reason why the expression in \Cref{eqn:def-avg-rad} is named \emph{average radius}. 

For list-recovery, we define \emph{$\ell$-average radius} through \emph{$\ell$-plurality}.

\begin{definition}[$\ell$-plurality]
\label{def:l-plur}
Let $ L\in\bbZ_{\ge2} $ and $ 1\le\ell\le q-1 $ be an integer. 
The \emph{$ \ell $-plurality} of an $L$-tuple of symbols $ x_1,\dots,x_L\in[q] $ is defined as the total number of times that the top $\ell$ most frequent symbols appear:
\begin{align}
\plur_\ell(x_1,\dots,x_L) &\coloneqq \max_{\Sigma\in\binom{[q]}{\ell}} \sum_{i=1}^L \indicator{x_i\in\Sigma} . \notag 
\end{align}
The maximizer is denoted by $ \argplur_\ell(x_1,\cdots,x_L)\in\binom{[q]}{\ell} $. 
\end{definition}

\begin{remark}
It is easy to see that $ \plur_1(\cdot) = \plur(\cdot) $. 
\end{remark}

\begin{definition}[$\ell$-average radius]
\label{def:avg-rad-list-rec}
The \emph{$\ell$-average radius} of an $L$-set of vectors $ \vx_1,\dots,\vx_L\in[q]^n $ is defined as 
\begin{align}
    \ol{\rad}_\ell(\vx_1,\dots,\vx_L) &\coloneqq \sum_{j=1}^n \paren{1 - \frac{1}{L}\plur_\ell(x_{1,j},\dots,x_{L,j})} . \label{eqn:def-avg-rad-lr}
\end{align}
\end{definition}

Similarly, \Cref{eqn:def-avg-rad-lr} is the solution to the following relaxation of \Cref{eqn:def-rad-lr}:
\begin{align}
    \min_{\vcY\in\binom{[q]}{\ell}^n} \exptover{I\sim\unif([L])}{\distlr(\vx_I,\vcY)} , \notag 
\end{align}
with minimizer $ \vcY^*\in\binom{[q]}{\ell}^n $ given by $ \cY_i^* = \argplur_\ell(x_{1,i},\dots,x_{L,i}) $ for every $i\in[n]$. 

\section{Main results}
\label{sec:main-results}

\subsection{$q$-ary list-decoding}
\label{sec:results-list-dec-qary}

Define $ f_{q,L}\colon\Delta([q])\to\bbR_{\ge0} $ as
\begin{align}
f_{q,L}(P) &\coloneqq \exptover{(X_1,\dots,X_L)\sim P^{\ot L}}{\plur(X_1,\dots,X_L)} \label{eqn:def-fql}
\end{align}
for $ P\in\Delta([q]) $. 

For $ w\in[0,1] $, let $ P_{q,w}\in\Delta([q]) $ denote the following probability vector:
\begin{align}
P_{q,w} &\coloneqq \paren{\frac{w}{q-1},\dots,\frac{w}{q-1},1-w} . \label{eqn:def-pqw} 
\end{align}
Define $ g_{q,L}\colon[0,1]\to\bbR_{\ge0} $ as
\begin{align}
g_{q,L}(w) &\coloneqq f_{q,L}(P_{q,w}) . \label{eqn:def-gql}
\end{align}

\begin{definition}[Majorization]
\label{def:majorization}
Let $ \va,\vb\in\bbR^d $.
Let $ \va^{\downarrow},\vb^{\downarrow}\in\bbR^d $ denote the vectors obtained by sorting the elements in $ \va $ and $ \vb $ in descending order, respectively. 
We say that \emph{$ \va $ majorizes $ \vb $}, written as $ \va\succeq\vb $, if 
\begin{align}
\sum_{i=1}^k a^{\downarrow}_i &\ge \sum_{i=1}^k b^{\downarrow}_i \notag 
\end{align}
for every $ k\in[d] $, and
\begin{align}
\sum_{i=1}^d a_i &= \sum_{i=1}^d b_i . \notag 
\end{align}
\end{definition}

\begin{definition}[Schur convexity]
\label{def:schur-convex}
A function $ f\colon\bbR^d\to\bbR $ is called \emph{Schur-convex} if $ f(\vx)\ge f(\vy) $ for every $ \vx,\vy\in\bbR^d $ such that $ \vx\succeq\vy $ (in the sense of \Cref{def:majorization}). 
\end{definition}

\begin{theorem}[Schur convexity of $ f_{q,L} $]
\label{thm:schur-convex-fql}
For any $ q\in\bbZ_{\ge2} $ and $ L\in\bbZ_{\ge2} $, the function $ f_{q,L}\colon\Delta([q])\to\bbR_{\ge0} $ defined in \Cref{eqn:def-fql} is Schur convex. 
\end{theorem}

\begin{proof}
See \Cref{sec:schu-convex-fql}. 
\end{proof}

\begin{theorem}[Convexity of $ g_{q,L} $]
\label{thm:convex-gql}
For any $ q\in\bbZ_{\ge2} $ and $ L\in\bbZ_{\ge2} $, the function $ g_{q,L}\colon[0,1]\to\bbR_{\ge0} $ defined in \Cref{eqn:def-gql} is convex in the interval $ [0,(q-1)/q] $. 
\end{theorem}

\begin{proof}
See \Cref{sec:convex-gql}. 
\end{proof}

\begin{remark}
\label{rk:binary-easy}
In the binary case (i.e., $ q=2 $), understanding the functions $ f_{2,L} $ and $ g_{2,L} $ is an easier task. 
In fact, $ f_{2,L} $ collapses to a univariate function and coincides with $ g_{2,L} $. 
It can be computed \cite[Eqn.\ (2.15) and (2.16)]{blinovsky-book} that for $ L = 2k,2k+1 $, 
\begin{align}
    p_*(2,L,w) &\coloneqq 1 - \frac{1}{L} g_{2,L}(w) = \sum_{i = 1}^k \frac{\binom{2i-2}{i-1}}{i} (w(1-w))^i , \notag 
\end{align}
and 
\begin{align}
    \frac{\partial^2}{\partial w^2} p_*(2,L,w) &= -k\binom{2k}{k} (w(1-w))^{k-1} . \notag 
\end{align}
The concavity (see also \cite[Lemma 8]{polyanskiy-2016-list-dec}) and monotonicity of $ p_*(2,L,w) $ immediately follow. 
Such explicit computation cannot be performed in the $q>2$ case (and for list-recovery) and we have to work with summations like in \Cref{lemma:derivs}. 
Other approaches to arguing monotonicity such as induction \cite[Lemma 8(d)]{abp-2018} do not seem to work well either for larger $q$. 
\end{remark}

As convexity only holds in the interval $[0,(q-1)/q]$, we will also require the following monotonicity properties, which follow easily from the Schur convexity of $f_{q,L}$. 

\begin{restatable}{lemma}{monotonicitygqL} \label{prop:mono-gqL}
    For any $q \in \bbZ_{\ge2}$, the function $ g_{q,L}\colon[0,1]\to\bbR_{\ge0} $ defined in \Cref{eqn:def-gql} is non-increasing on $[0,(q-1)/q]$ and non-decreasing on $[(q-1)/q,1]$.
\end{restatable}

\begin{proof}
    See \Cref{subsec:mono-gqL}.
\end{proof}

Define 
\begin{align}
p_*(q,L,w) &\coloneqq 1 - \frac{1}{L} g_{q,L}(w) . \label{eqn:def-p-star-q-l-w}
\end{align}

\begin{theorem}[Plotkin bound for $q$-ary list-decoding]
\label{thm:plotkin-qary-list-dec-main}
Let $ \cC\subset[q]^n $ be an arbitrary $ (p,L)_q $-list-decodable code with $ p = p_*\paren{q,L,\frac{q-1}{q}} + \tau $ for any constant $ \tau\in(0,1) $. 
Then there exists a constant $ M_* = M_*(q,L,\tau) $ independent of $n$ such that $ |\cC|\le M_* $. 
As a consequence, in partucular, we have 
\begin{align}
    p_*(q,L) &\le p_*\paren{q,L,\frac{q-1}{q}} = 1 - \frac{1}{L} g_{q,L}\paren{\frac{q-1}{q}} . \notag 
\end{align}
\end{theorem}

\begin{proof}
The proof of this theorem can be found in \Cref{sec:plotkin-list-dec}. 
Specifically, a theorem (cf.\ \Cref{thm:plotkin-list-dec-qary}) of the above kind will be first proved for \emph{approximately constant-weight} codes in which all codewords have approximately the same Hamming weight. 
This theorem can then be used to prove \Cref{thm:plotkin-qary-list-dec-main} above (see \Cref{cor:plotkin-list-dec-qary-unconstr} for a more quantitative version) by partitioning a general (weight-unconstrained) code into a constant number of almost constant-weight subcodes. 
\end{proof}

The lower bound on the zero-rate threshold in \Cref{thm:plotkin-qary-list-dec-main} is in fact sharp. 
It turns out that positive rate $(p,L)_q$-list-decodable codes exist for any $ p $ strictly smaller than the bound $1 - \frac{1}{L} g_{q,L}\paren{\frac{q-1}{q}}$ in \Cref{thm:plotkin-qary-list-dec-main}. 
Indeed, Blinovsky \cite{blinovsky-2005-ls-lb-qary} proved the following lower bound on the $(p,L)_q$-list-decoding capacity which remains the best known to date. 
It can also be implied by our lower bound (\Cref{thm:lb-listrec} below) for list-recovery upon setting $\ell=1$. 

\begin{theorem}[{\cite[Sec.\ 2]{blinovsky-2005-ls-lb-qary}}]
\label{thm:bli-lb-listdec}
The following lower bound on the $(p,L)_q$-list-decoding capacity holds:
\begin{align}
C_{(p,L)_q} &\ge \frac{L}{L-1} - \frac{1}{L-1}\curbrkt{\lambda_*p + \log_q\sqrbrkt{\sum_{\va\in\cA_{q,L}} \binom{L}{\va} \exp_q\paren{-\lambda_*\paren{1 - \frac{1}{L}\max\curbrkt{\va}}}}}, \notag 
\end{align}
where $ \lambda_*=  \lambda_*(q,L,p) $ is the solution to the following equation
\begin{align}
p &= \frac{\sum\limits_{\va\in\cA_{q,L}} \binom{L}{\va} \exp_q\paren{-\lambda_*\paren{1 - \frac{1}{L}\max\curbrkt{\va}}} \paren{1 - \frac{1}{L}\max\curbrkt{\va}}}{\sum\limits_{\va\in\cA_{q,L}} \binom{L}{\va} \exp_q\paren{-\lambda_*\paren{1 - \frac{1}{L}\max\curbrkt{\va}}}}. \notag 
\end{align}
\end{theorem}
Blinovsky's lower bound is plotted in \Cref{fig:bli-qary}. 
It is not hard to verify that the lower bound above vanishes at 
\begin{align}
    p 
    = q^{-L} \sum_{\va\in\cA_{q,L}} \binom{L}{\va} \paren{1 - \frac{1}{L}\max\{\va\}} , \notag
\end{align}
and the corresponding $\lambda_*$ equals $0$. 

\Cref{thm:plotkin-qary-list-dec-main,thm:bli-lb-listdec} together pin down the exact value of $ p_*(q,\ell,L) $ shown in the following corollary. 
\begin{corollary}
\label{cor:plotkin-pt-listdec}
The zero-rate threshold $ p_*(q,\ell,L) $ for $(p,L)_q$-list-decoding is given by 
\begin{align}
    p_*(q,L) &= p_*\paren{q,L,\frac{q-1}{q}} = 1 - \frac{1}{L} g_{q,L}\paren{\frac{q-1}{q}}
    = q^{-L} \sum_{\va\in\cA_{q,L}} \binom{L}{\va} \paren{1 - \frac{1}{L}\max\{\va\}}. \label{eqn:plotkin-pt-listdec}
\end{align}
\end{corollary}

From now on, we will use $ p_*(q,L) $ to denote the RHS of \Cref{eqn:plotkin-pt-listdec}. 

\begin{theorem}[Elias--Bassalygo bound for $q$-ary list-decoding]
\label{thm:eb-qary-list-dec-main}
Suppose $ p<p_*(q,L) $. 
Then the $ (p,L)_q $-list-decoding capacity can be upper bounded as $C_{(p,L)_q} \le 1 - H_q(w_{q,L})$ where $ w_{q,L} $ is the solution to the equation $ p_*(q,L,w) = p $ in $ w\in[0,(q-1)/q] $. 
\end{theorem}

\begin{proof}
The above theorem is implied by \Cref{thm:eb-qary-list-dec} proved in \Cref{sec:eb-list-dec-qary-pf}. 
The latter theorem shows that for any $ (p,L)_q $-list-decodable code $ \cC\subset[q]^n $ with $ p<p_*(q,L) $ and any sufficiently small constant $ \tau>0 $, $ |\cC| $ is at most $ B\cdot n^{1.5}\cdot q^{n(1 - H_q(w_{q,L,\tau}))} $, where $ B = B(q,L,\tau) $ is a constant and $ w_{q,L,\tau} $ is the solution to $ p_*(q,L,w) = p - \tau $. 
Taking $ \tau\to0 $ and neglecting polynomial factors, we obtain the upper bound on the list-decoding capacity. 
\end{proof}

The above upper bound is plotted in \Cref{fig:bli-qary}. 

\subsection{List-recovery}
\label{eqn:results-list-rec}

Define $ f_{q,L,\ell}\colon\Delta([q])\to\bbR_{\ge0} $ as
\begin{align}
f_{q,L,\ell}(P) &\coloneqq \exptover{(X_1,\dots,X_L)\in P^{\ot L}}{\plur_\ell(X_1,\dots,X_L)} \label{eqn:def-fqll}
\end{align}
for $ P\in\Delta([q]) $.
Define $ g_{q,L,\ell}\colon[0,1]\to\bbR_{\ge0} $ as
\begin{align}
g_{q,L,\ell}(w) &\coloneqq f_{q,L,\ell}(P_{q,\ell,w}) , 
\label{eqn:def-gqll}
\end{align}
where the distribution $ P_{q,\ell,w}\in\Delta([q]) $ is defined as
\begin{align}
    P_{q,\ell,w}(i) &= \begin{cases}
    \frac{w}{q-\ell} , & 1\le i\le q-\ell \\
    \frac{1-w}{\ell} , & q-\ell + 1\le i\le q
    \end{cases} . \label{def:p-q-ell-rho}
\end{align}
\begin{theorem}[Schur convexity of $ f_{q,L,\ell} $]
\label{thm:schur-convex-fqll}
For any $ q\in\bbZ_{\ge2} $, $ L\in\bbZ_{\ge2} $ and integer $ 1\le\ell\le q-1 $, the function $ f_{q,L,\ell}\colon\Delta([q])\to\bbR_{\ge0} $ defined in \Cref{eqn:def-fqll} is Schur convex. 
\end{theorem}

\begin{proof}
See \Cref{sec:proof-schur-convex-fqll}. 
\end{proof}

\begin{theorem}[Convexity of $ g_{q,L,\ell} $]
\label{thm:convex-gqll}
For any $ q\in\bbZ_{\ge2} $, $ L\in\bbZ_{\ge2} $ and integer $ 2\le\ell\le q-1 $, the function $ g_{q,L,\ell}\colon\Delta([q])\to\bbR_{\ge0} $ defined in \Cref{eqn:def-gqll} is convex in the interval $ w\in [0,1] $. 
\end{theorem}

\begin{proof}
See \Cref{sec:proof-convex-gqll}. 
\end{proof}

Define 
\begin{align}
p_*(q,\ell,L,w) &\coloneqq 1 - \frac{1}{L} g_{q,L,\ell}(w) . \label{eqn:def-p-star-q-ell-l-w} 
\end{align}

\begin{theorem}[Plotkin bound for list-recovery]
\label{thm:plotkin-list-rec-main}
Let $ \cC\subset[q]^n $ be an arbitrary $ (p,\ell,L)_q $-list-recoverable code with $ p = p_*\paren{q,\ell,L,\frac{q-\ell}{q}} + \tau $ for any constant $ \tau\in(0,1) $. 
Then there exists a constant $ M_* = M_*(q,\ell,\tau) $ independent of $n$ such that $ |\cC|\le M_* $. 
This implies, in particular, 
\begin{align}
    p_*(q,\ell,L) &\le p_*\paren{q,\ell,L,\frac{q-\ell}{q}} = 1 - \frac{1}{L} g_{q,L,\ell}\paren{\frac{q-\ell}{q}} . \notag
\end{align}
\end{theorem}

\begin{proof}
The proof structure is similar to that of \Cref{thm:plotkin-qary-list-dec-main}. 
We first prove the analogous the statement for almost constant-weight codes (in which all codewords have approximately the same list-recovery weight) in \Cref{thm:plotkin-list-rec} and then pass to general codes by weight partitioning (cf.\ \Cref{cor:plotkin-list-rec-unconstr}). 
Since the technical proofs bear many similarities to those in the list-decoding case, we only present proof sketches in \Cref{sec:plotkin-list-rec}. 
\end{proof}

To complement \Cref{thm:plotkin-list-rec-main}, we prove in \Cref{sec:list-rec-lb} the following lower bound on the $(p,\ell,L)_q$-list-recovery capacity. 
To the best of our knowledge, this is the first bound for list-recovery with $q,\ell,L$ all being \emph{constants} (independent of $p$ and $n$). 
We believe that improving it likely requires novel techniques beyond expurgation. 

\begin{theorem}
\label{thm:lb-listrec}
The following lower bound on the $(p,\ell,L)_q$-list-recovery capacity holds:
\begin{align}
C_{(p,\ell,L)_q} &\ge \frac{L}{L-1} - \frac{1}{L-1}\curbrkt{\lambda_*p + \log_q\sqrbrkt{\sum_{\va\in\cA_{q,L}} \binom{L}{\va} \exp_q\paren{-\lambda_*\paren{1 - \frac{1}{L}\maxl\curbrkt{\va}}}}}, \notag 
\end{align}
where $ \lambda_*=  \lambda_*(q,\ell,L,p) $ is the solution to the following equation
\begin{align}
p &= \frac{\sum\limits_{\va\in\cA_{q,L}} \binom{L}{\va} \exp_q\paren{-\lambda_*\paren{1 - \frac{1}{L}\maxl\curbrkt{\va}}} \paren{1 - \frac{1}{L}\maxl\curbrkt{\va}}}{\sum\limits_{\va\in\cA_{q,L}} \binom{L}{\va} \exp_q\paren{-\lambda_*\paren{1 - \frac{1}{L}\maxl\curbrkt{\va}}}}. \notag 
\end{align}
\end{theorem}
Similar to the list-decoding case (\Cref{thm:bli-lb-listdec}), the above lower bound vanishes at 
\begin{align}
    p 
    = q^{-L} \sum_{\va\in\cA_{q,L}} \binom{L}{\va} \paren{1 - \frac{1}{L}\maxl\{\va\}} , \notag
\end{align}
and the corresponding $\lambda_*$ equals $0$. 

\Cref{thm:plotkin-list-rec-main,thm:lb-listrec} jointly determine the value of $ p_*(q,\ell,L) $ shown in the corollary below. 
\begin{corollary}
\label{cor:plotkin-pt-listrec}
The zero-rate threshold $p_*(q,\ell,L)$ for $ (p,\ell,L)_q $-list-recovery is given by 
\begin{align}
    p_*(q,\ell,L) &= p_*\paren{q,\ell,L,\frac{q-\ell}{q}} = 1 - \frac{1}{L} g_{q,L,\ell}\paren{\frac{q-\ell}{q}}
    = q^{-L} \sum_{\va\in\cA_{q,L}} \binom{L}{\va} \paren{1 - \frac{1}{L}\maxl\{\va\}} . \label{eqn:plotkin-pt-listrec}
\end{align}
\end{corollary}

From now on, we use $p_*(q,\ell,L)$ to refer to the same quantity as the RHS of \Cref{eqn:plotkin-pt-listrec}. 

\begin{theorem}[Elias--Bassalygo bound for list-recovery]
\label{thm:eb-list-rec-main}
Suppose $ p<p_*(q,\ell,L) $. 
Then the $ (p,\ell,L)_q $-list-recovery capacity can be upper bounded as $ C_{(p,\ell,L)_q}\le1 - H_{q,\ell}(w_{q,\ell,L}) $ where $ w_{q,\ell,L} $ is the solution to the equation $ p_*(q,\ell,L,w) = p $ in $ w\in[0,(q-\ell)/q] $. 
\end{theorem}

\begin{proof}
Parallel to \Cref{thm:eb-qary-list-dec-main}, the above theorem is immediately implied by a finite-blocklength version \Cref{thm:eb-list-rec} (analogous to \Cref{thm:eb-qary-list-dec}) whose full proof is presented in \Cref{sec:eb-list-rec-pf}. 
\end{proof}

\subsection{Expressions for $f_{q,L,\ell}$, $g_{q,L,\ell}$ and their derivatives}

Before concluding this section, we collect representations for $f_{q,L,\ell}$, $g_{q,L,\ell}$ and their derivatives which will be useful later on. We require the following notations. 

    For integers $q\ge1,m\ge0$,
    \begin{align}
        \cA_{q,m} &\coloneqq \curbrkt{(a_1,\dots,a_q)\in\bbZ_{\ge0}^q:\sum_{i = 1}^q a_i = m} . \notag 
    \end{align}
    For $\va = (a_1,\dots,a_q) \in \cA_{q,m}$, we often simplify notation for the multinomial coefficient and write 
    \begin{align}
        \binom{m}{\va} &\coloneqq \binom{m}{a_1,\dots,a_m} . \notag
    \end{align}
    Next, for any $\va \in \mathbb R^q$, let 
    \begin{align}
        \maxl\curbrkt{\va} &\coloneqq \max_{\cQ\in\binom{[q]}{\ell}} \sum_{i\in \cQ} a_i . \notag 
    \end{align}
    denote the largest partial sum of $\ell$ coordinates from $ \va $.
    
    Lastly, for $i \in [q]$, $\ve_i$ denotes the length-$q$ vector with a $1$ in the $i$-th coordinate and $0$'s elsewhere. 

\begin{lemma} \label{lemma:derivs}
    We have, for all $1\le\ell\le q-1$, $P = (p_1,\dots,p_q) \in \Delta([q])$ and $j,k \in [q]$: 
    \begin{align}
        &f_{q,L,\ell}(P) = \sum_{\va \in \cA_{q,L}} \binom{L}{\va} \maxl\{\va\} \paren{\prod_{i=1}^{q} p_i^{a_i}} \ ; \label{eq:fqL-exp} \\
        &\frac{\partial}{\partial p_j}f_{q,L,\ell}(P) = L \sum_{\va \in \cA_{q,L}} \binom{L}{\va} \maxl\{\va + \ve_j\} \paren{\prod_{i=1}^{q} p_i^{a_i}} \ ; \label{eq:fqL'-exp}\\
        &\frac{\partial^2}{\partial p_k \partial p_j} f_{q,L,\ell}(p_1,\dots,p_q) = L(L-1) \sum_{\va \in \cA_{q,L-2}} \binom{L-2}{\va} \maxl\{\va + \ve_j + \ve_k\} \paren{\prod_{i=1}^{q} p_i^{a_i}} \ . \label{eq:fqL''-exp} 
    \end{align}
    Furthermore, defining 
    \begin{align*}
        G_\ell(\va) &\coloneqq \frac{1}{(q-\ell)^2}\sum_{1\leq i,j \leq q-\ell}\maxl\{\va+\ve_i+\ve_j\} - \frac{2}{(q-\ell)\ell}\sum_{i=1}^{q-\ell}\sum_{j=q-\ell+1}^q\maxl\{\va+\ve_i+\ve_j\} \\
     &\quad +\frac{1}{\ell^2}\sum_{q-\ell+1\leq i,j \leq q}\maxl\{\va+\ve_i+\ve_j\} \ ,
    \end{align*}
    we have 
    \begin{align}
        g_{q,L,\ell}''(w) = L(L-1)\paren{\frac{w}{q-\ell}}^{L-2}\sum_{\va \in \cA_{q,L-2}} \binom{L-2}{\va}\paren{\frac{(q-\ell)(1-w)}{\ell w}}^{a_{q-\ell+1}+\cdots+a_q}G_\ell(\va) \ .\label{eq:exp-gqL''}
    \end{align}
\end{lemma}

\begin{proof}
    See \Cref{subsec:deriv-expressions}. 
\end{proof}

\section{Schur convexity of $ f_{q,L} $: proof of \Cref{thm:schur-convex-fql}}
\label{sec:schu-convex-fql}

First, we provide the criterion for Schur-convexity that we use when proving $f_{q,L}$ (and later, $f_{q,L,\ell}$) is Schur-convex. 

\begin{theorem}[Schur--Ostrowski criterion, \cite{roberts1993convex}]
\label{thm:schur-criterion}
Let $ f\colon\bbR^d\to\bbR $ be a symmetric function, i.e., 
\begin{align}
f(x_1,\dots,x_d) &= f(x_{\pi(1)},\dots,x_{\pi(d)}) \notag 
\end{align}
for every $ \pi\in S_d $. 
Suppose all first partial derivatives of $ f $ exist. 
Then $ f $ is Schur-convex if and only if 
\begin{align}
(x_i - x_j)\cdot\paren{\frac{\partial}{\partial x_i} - \frac{\partial}{\partial x_j}} f(x_1,\dots,x_d) &\ge 0 \notag 
\end{align}
for all $ \vx\in\bbR^d $ and $ 1\le i\ne j\le d $. 
\end{theorem}

\begin{proof}[Proof of~\Cref{thm:schur-convex-fql}]


By \Cref{eq:fqL'-exp} of \Cref{lemma:derivs}, we have 
\begin{align}
    \frac{\partial}{\partial p_j}f_{q,L}(P) = L \sum_{\va \in \cA_{q,L}} \binom{L-1}{\va} \max\{\va + \ve_j\} \paren{\prod_{i=1}^{q} p_i^{a_i}} \ . \label{eq:fqL'-exp-proof}
\end{align}

The proof then proceeds by verifying the Schur--Ostrowski criterion (\Cref{thm:schur-criterion}). 
Fix an arbitrary pair $ k\ne j\in[q] $ and a distribution $ (p_1,\cdots,p_q)\in\Delta([q]) $. 
Without loss of generality, assume $ p_k>p_j $. 
The conclusion follows if one can show $ \paren{\frac{\partial}{\partial p_k} - \frac{\partial}{\partial p_j}}f_{q,L}(p_1,\cdots,p_q)\ge0 $. 
By \Cref{eq:fqL'-exp-proof}, we have 
\begin{align}
\paren{\frac{\partial}{\partial p_k} - \frac{\partial}{\partial p_j}}f_{q,L}(p_1,\cdots,p_q) 
&= L \sum_{\va\in\cA_{q,L}} \binom{L-1}{\va} \paren{\prod_{i = 1}^q p_i^{a_i}} \paren{\max\curbrkt{\va+\ve_k} - \max\curbrkt{\va+\ve_j}}
\end{align}
It is easy to check that 
\begin{align}
\max\curbrkt{\va+\ve_k} - \max\curbrkt{\va+\ve_j} 
&= \begin{cases}
1-1=0, & a_k = \max\curbrkt{\va}, a_j=\max\curbrkt{\va} \\
1-0=1, & a_k = \max\curbrkt{\va}, a_j<\max\curbrkt{\va} \\
0-1=-1, & a_k < \max\curbrkt{\va}, a_j=\max\curbrkt{\va}  \\
0-0=0, & a_k < \max\curbrkt{\va}, a_j<\max\curbrkt{\va} 
\end{cases}
= \begin{cases}
1 , & a_j<a_k = \max\curbrkt{\va} \\
-1 , & a_k<a_j = \max\curbrkt{\va} \\
0 , & \ow 
\end{cases} . \notag 
\end{align}
Therefore 
\begin{align}
& \frac{1}{L}\paren{\frac{\partial}{\partial p_k} - \frac{\partial}{\partial p_j}}f_{q,L}(p_1,\cdots,p_q) \notag \\
&= \sum_{\substack{\va\in\cA_{q,L} \\ \max\curbrkt{\va} = a_k > a_j}} \binom{L-1}{\va} \paren{\prod_{i = 1}^q p_i^{a_i}} - \sum_{\substack{\va\in\cA_{q,L} \\ a_k < a_j = \max\curbrkt{\va}}} \binom{L-1}{\va} \paren{\prod_{i = 1}^q p_i^{a_i}} \notag \\
&= \sum_{\substack{\va\in\cA_{q,L} \\ \max\curbrkt{\va} = a_k > a_j}} \binom{L-1}{\va} \paren{\prod_{i = 1}^q p_i^{a_i}} - \sum_{\substack{\va\in\cA_{q,L} \\ \max\curbrkt{\va} = a_k > a_j}} \binom{L-1}{\pi_{k,j}(\va)} \paren{\prod_{i = 1}^q p_i^{a_{\pi_{k,j}(i)}}} \label{eqn:def-swap} \\
&= \sum_{\substack{\va\in\cA_{q,L} \\ \max\curbrkt{\va} = a_k > a_j}} \binom{L-1}{\va} \paren{\prod_{i\in[q]\setminus\{k,j\}} p_i^{a_i}} \paren{p_k^{a_k}p_j^{a_j} - p_k^{a_j}p_j^{a_k}} . \notag 
\end{align}
In \Cref{eqn:def-swap} we denote by $ \pi_{k,j}\colon[q]\to[q] $ the transposition that swaps $ k $ and $ j $, and we let it act on vectors $\va$ by coordinate permutation. 
Now observe that $ p_k^{a_k}p_j^{a_j} - p_k^{a_j}p_j^{a_k}>0 $. 
Indeed, this follows since $ p_k> p_j $ and $ a_k>a_j $, hence
\begin{align}
\paren{\frac{p_k}{p_j}}^{a_k} &> \paren{\frac{p_k}{p_j}}^{a_j} . \notag
\end{align}
This shows that $ \paren{\frac{\partial}{\partial p_k} - \frac{\partial}{\partial p_j}}f_{q,L}(p_1,\cdots,p_q)>0 $ and the proof of \Cref{thm:schur-convex-fql} is finished. 
\end{proof}

\section{Convexity of $ g_{q,L} $: proof of \Cref{thm:convex-gql}}
\label{sec:convex-gql}

\begin{proof}[Proof of \Cref{thm:convex-gql}] 
We will show that $ g''_{q,L}\ge0 $ on the interval $ [0,(q-1)/q] $. By \Cref{eq:exp-gqL''} of \Cref{lemma:derivs}, we have

\begin{align}
    g_{q,L}''(w) &= L(L-1) \paren{\frac{w}{q-1}}^{L-2}\sum_{\va\in\cA_{q,L-2}}\binom{L-2}{\va}\paren{(q-1)\frac{1-w}{w}}^{a_q} G(\va) \label{eqn:exp-for-g''}
\end{align}
where 
\begin{align}
    G(\va) = \frac{1}{(q-1)^2}\sum_{i,j \in [q-1]}\max\{\va+\ve_i+\ve_j\} - \frac{2}{q-1}\sum_{i \in [q-1]}\max\{\va+\ve_i+\ve_q\} + \max\{\va+2\ve_q\} .  \label{eqn:G-defn}
\end{align}
We now seek to lower bound $G(\va)$ for a given $\va \in \cA_{q,L-2}$. Let $r$ denote the number of distinct $i \in [q]$ for which $a_i = \max\{\va\}$. By symmetry, it is without loss of generality to assume $a_1 \geq a_2 \geq \dots \geq a_{q-1}$. Further, observe that 
\begin{align*}
     \max\{\va\} \paren{ \sum_{i,j \in [q-1]}\frac{1}{(q-1)^2} - \sum_{i \in [q-1]}\frac{2}{q-1} +1 }  = 0 .
\end{align*}
Thus, we have
\begin{align}
    G(\va) = \frac{1}{(q-1)^2}\sum_{i,j \in [q-1]}\paren{\max\{\va+\ve_i+\ve_j\}-\max\{\va\}} &- \frac{2}{q-1}\sum_{i \in [q-1]}\paren{\max\{\va+\ve_i+\ve_q\}-\max\{\va\}} \notag \\
    &+ \paren{\max\{\va+2\ve_q\}-\max\{\va\}} . \label{eqn:G-zero}
\end{align}
We will lower bound each of the three terms separately. We consider two cases.

\paragraph{Case $a_q < \max\{\vec a\}$.} In this case, we claim $G(\va) \geq -\frac{r(r-1)}{(q-1)^2}$. Observe our assumption implies $a_1=\dots=a_r=\max\{\va\}$. Firstly,
\begin{align*}
    \sum_{i,j \in [q-1]}\paren{\max\{\va+\ve_i+\ve_j\}-\max\{\va\}} \geq 2r + r(q-2) + r(q-1-r) = 2r(q-1)-r(r-1) .
\end{align*}
Indeed, for $i,j \in [q-1]$ we note that if $i=j \leq r$ then $\max\{\va+\ve_i+\ve_j\}-\max\{\va\} = 2$; $i \leq r$ and $i \neq j$ then $\max\{\va+\ve_i+\ve_j\}-\max\{\va\} = 1$; and if $i>r$ and $j \leq r$ then again $\max\{\va+\ve_i+\ve_j\}-\max\{\va\} = 1$. Note further that these conditions on $(i,j) \in [q-1]$ define disjoint sets, and that for all other choices of $i,j \in [q-1]$ we have  $\max\{\va+\ve_i+\ve_j\}-\max\{\va\} \geq 0$. 

Next, 
\begin{align*}
    \sum_{i \in [q-1]}\paren{\max\{\va+\ve_i+\ve_q\}-\max\{\va\}} = r
\end{align*}
as if $i \leq r$ then $\max\{\va+\ve_i+\ve_q\}-\max\{\va\} = 1$ and otherwise $\max\{\va+\ve_i+\ve_q\}-\max\{\va\} = 0$.

Finally, it's clear $\max\{\va+2\ve_q\}-\max\{\va\} \geq 0$. Combining these inequalities and plugging them into \eqref{eqn:G-zero} we conclude $G(\va) \geq -\frac{r(r-1)}{(q-1)^2}$, as claimed. 

\paragraph{Case $a_q=\max\{\va\}$.} In this case, we claim $G(\va) \geq \frac{2(r-1)(q-1)-(r-1)(r-2)}{(q-1)^2}$. Observe our assumptions yield $a_i=\max\{\va\}$ if and only if $i\leq r-1$ or $i=q$. Firstly, we have 
\begin{align*}
    \sum_{i,j \in [q-1]}\paren{\max\{\va+\ve_i+\ve_j\}-\max\{\va\}} &\geq 2(r-1) + (r-1)(q-2) + (r-1)(q-1-(r-1)) \\
    &= 2(r-1)(q-1)-(r-1)(r-2) .
\end{align*}
The argument is completely analogous to the previous case, upon replacing $r$ by $r-1$. 

Next, note that for all $i \leq q-1$ we have $\max\{\va+\ve_i+\ve_q\}-\max\{\va\}=1$ so 
\begin{align*}
    \sum_{i \in [q-1]}\paren{\max\{\va+\ve_i+\ve_q\}-\max\{\va\}} = q-1,
\end{align*}
and furthermore $\max\{\va+2\ve_q\}-\max\{\va\}=2$. Plugging these into \Cref{eqn:G-zero} we find $G(\va) \geq \frac{2(r-1)(q-1)-(r-1)(r-2)}{(q-1)^2}$.

We now combine these bounds as follows to lower bound $g''_{q,L}(w)$. Consider the equivalence relation defined on $\cA_{L-2,q}$ obtained by identifying tuples that can be obtained from one another by coordinate permutation. Stated differently, let the symmetric group $S_q$ act on $\cA_{L-2,q}$ by coordinate permutation, and identify tuples lying in the same orbit. Let $\cR$ be a collection of representatives for this equivalence relation, and given $\va^* \in \cR$ let $\cO(\va^*)$ denote the equivalence class to which $\va^*$ belongs. As equivalence classes form a partition, we have 
\begin{align}
     g_{q,L}''(w) &= L(L-1) \paren{\frac{w}{q-1}}^{L-2}\sum_{\va\in\cA_{q,L-2}}\binom{L-2}{\va}\paren{(q-1)\frac{1-w}{w}}^{a_q} G(\va) \notag \\
    &= L(L-1)\paren{\frac{w}{q-1}}^{L-2}\sum_{\va^*\in\cR}\binom{L-2}{\va^*}\sum_{\va\in\cO(\va^*)}\binom{L-2}{\va}\paren{(q-1)\frac{1-w}{w}}^{a_q} G(\va) . \notag
\end{align}
It therefore suffices to fix $\va^* \in \cR$ and show 
\begin{align}
    \sum_{\va\in\cO(\va^*)}\paren{(q-1)\frac{1-w}{w}}^{a_q} G(\va) \ge 0 . \label{eqn:orbit-sum}
\end{align}
Towards this end, let $s \coloneqq \max\{\va^*\}$ and let $r$ denote the number of $i \in [q]$ for which $a_i^*=s$. Define further $\cO_1(\va^*) \coloneqq \{\va \in \cO(\va^*):a_q<s\}$ and $\cO_2(\va^*)\coloneqq \cO(\va^*)\setminus \cO_1(\va^*) = \{\va \in \cO(\va^*):a_q=s\}$. Thus, for $\va \in \cO_1(\va^*)$ we have $G(\va) \geq -\frac{r(r-1)}{(q-1)^2}$ and for $\va \in \cO_2(\va^*)$, $G(\va) \geq \frac{2(r-1)(q-1)-(r-1)(r-2)}{(q-1)^2}$. Further, note that for any $\va\in\cO(\va^*)$ we have that if we sample a random permutation $\Pi \in S_q$ then the probability that the $\va^{\Pi} \in \cO_2(\va^*)$ is $\frac{r}{q}$, where $\va^{\Pi}$ is vector obtained after permuting $\va$'s coordinates according to $\Pi$. Therefore $\frac{|\cO_2(\va^*)|}{|\cO(\va^*)|} = \frac{r}{q}$ and so also $\frac{|\cO_1(\va^*)|}{|\cO(\va^*)|} = \frac{q-r}{q}$. 

Finally, the assumption $w \in \left[0,\frac{q-1}{q}\right]$ is equivalent to the bound $0 \leq \paren{(q-1)\frac{1-w}{w}}\leq 1$, which guarantees that for $t<u$, we have $\paren{(q-1)\frac{1-w}{w}}^t \geq \paren{(q-1)\frac{1-w}{w}}^u$. Thus:
\begin{align*}
    & \sum_{\va\in\cO(\va^*)}\paren{(q-1)\frac{1-w}{w}}^{a_q} G(\va) \\
    &= \sum_{\va \in \cO_1(\va^*)}\paren{(q-1)\frac{1-w}{w}}^{a_q}G(\va) + \sum_{\va \in \cO_2(\va^*)}\paren{(q-1)\frac{1-w}{w}}^{a_q}G(\va) \\
    &\geq \sum_{\va \in \cO_1(\va^*)}\paren{(q-1)\frac{1-w}{w}}^{a_q}\cdot \frac{-r(r-1)}{(q-1)^2} + \sum_{\va \in \cO_2(\va^*)}\paren{(q-1)\frac{1-w}{w}}^{a_q}\cdot \frac{2(r-1)(q-1)-(r-1)(r-2)}{(q-1)^2} \\
    &\geq \sum_{\va \in \cO_1(\va^*)}\paren{(q-1)\frac{1-w}{w}}^{s}\cdot \frac{-r(r-1)}{(q-1)^2} + \sum_{\va \in \cO_2(\va^*)}\paren{(q-1)\frac{1-w}{w}}^{s}\cdot \frac{2(r-1)(q-1)-(r-1)(r-2)}{(q-1)^2} \\
    &= \paren{(q-1)\frac{1-w}{w}}^{s}\cdot|\cO(\va^*)|\cdot \paren{\frac{q-r}{q}\cdot \frac{-r(r-1)}{(q-1)^2} + \frac{r}{q}\cdot\frac{2(r-1)(q-1)-(r-1)(r-2)}{(q-1)^2}} \\
    &= \paren{(q-1)\frac{1-w}{w}}^{s}\cdot|\cO(\va^*)|\cdot \frac{r-1}{q(q-1)^2}\paren{2r(q-1)-r(q-r)-(r-2)}
\end{align*}
As $1 \leq r \leq q$ and $q \geq 2$, it is immediate that the above expression is non-negative. This establishes \Cref{eqn:orbit-sum} and therefore completes the proof of \Cref{thm:convex-gql}. 
\end{proof}

\section{Plotkin bound for $q$-ary list-decoding: proof of \Cref{thm:plotkin-qary-list-dec-main}}
\label{sec:plotkin-list-dec}

Before diving into the proof, let us introduce the concept of \emph{type} which will be used in the current section and \Cref{sec:plotkin-list-rec}. 

\begin{definition}[Type]
\label{def:type}
Let $ \cX_1,\cdots,\cX_k $ be $k$ arbitrary finite sets. 
Then the \emph{type} $ \type_{\vx_1,\cdots,\vx_k}\in\Delta\paren{\prod_{i=1}^k\cX_i} $ of a $k$-tuple of vectors $ (\vx_1,\cdots,\vx_k)\in\cX_1^n\times\cdots\times\cX_k^n $ is defined as its empirical distribution, i.e., 
\begin{align}
    \type_{\vx_1,\cdots,\vx_k} (z_1,\cdots,z_k) &\coloneqq \frac{1}{n} \sum_{j=1}^n \prod_{i=1}^k \indicator{x_{i,j} = z_i} \notag 
\end{align}
for any $ (z_1,\cdots,z_k)\in\cX_1\times\cdots\times\cX_k $. 
\end{definition}

Suppose that we are given a code $ \cC\subset[q]^n $ of size $M$ and therefore rate $ R = \frac{1}{n}\log_qM $. 
Suppose that $ \cC $ is guaranteed to be $ (p,L)_q $-list-decodable for a given $ L\in\bbZ_{\ge2} $. 
The goal is to upper bound $ M $ and therefore $R$. 

\begin{theorem}[Plotkin bound, $ q $-ary list-decoding, approximately constant weight]
\label{thm:plotkin-list-dec-qary}
Let $ \cC\subset[q]^n $ be a $ (p,L)_q $-list-decodable code with $ p = p_*(q,L,w) + \tau $ for $ w\in[0,(q-1)/q] $ and some small constant $ \tau>0 $. 
Let 
\begin{align}
    0<\eps_1\le\frac{L\tau}{8\lip(g_{q,L})} , \notag 
\end{align}
where $ \lip(g_{q,L}) $ depends only on $q,L$ and will be defined in \Cref{eqn:def-lip-gql}. 
Assume that every codeword in $ \cC $ has Hamming weight between $ \max\{0, n(w-\eps_1)\} $ and $ \min\{n(w+\eps_1), n(q-1)/q\} $. 
Then 
\begin{align}
|\cC| &< 
N(c,L,M_0) , \label{eqn:plotkin-ub-quant} 
\end{align}
where
\begin{align}
\begin{split}
c &\coloneqq \paren{\frac{80^2L^6q^{4L-2}}{2\tau^2} + 1}^{q^L} , \\
M_0 &\coloneqq \max\curbrkt{ \frac{2^{11}L^7q^{2L}}{\tau^2} + L-2 , (L-1)L\paren{\frac{p_*(q,L,w)}{\frac{1}{L}\cdot\lip(g_{q,L})\cdot\eps_1} + 2} + 1} , 
\end{split}
\label{eqn:def-c-m0}
\end{align}
and the function $ N(\cdot,\cdot,\cdot) $ denotes the Ramsey number given by \Cref{thm:hypergraph-ramsey}. 
\end{theorem}

\begin{proof}
Suppose $ p = p_*(q,L,w) + \tau $ and we are given a $ (p,L)_q $-list-decodable code $ \cC_1\subseteq[q]^n $ all of whose codewords have Hamming weight between $ n(w-\eps_1) $ and $ n(w+\eps_1) $. 
The goal is to show that $ |\cC_1|\le M_* $ for some $ M_* $ that depends on $ q,L,w,\tau $, but not on $n$. 
Suppose towards a contradiction that \Cref{eqn:plotkin-ub-quant} holds in the reverse direction:
\begin{align}
|\cC_1| &\ge 
N(c,L,M_0) . 
\label{eqn:hyp-contradiction}
\end{align} 

\textbf{The first step} is to pass to an almost ``equi-coupled'' subcode $ \cC_2\subset\cC_1 $. 
This follows from a standard Ramsey argument. 
Let $ \eps_2 $ be a small positive constant defined as
\begin{align}
\eps_2 &\coloneqq \frac{\tau^2}{80^2L^6q^{3L-2}} . \label{eqn:def-eps2} 
\end{align}
We build an $L$-uniform complete hypergraph $ \cH_{L,|\cC_1|} $ in which each codeword in $ \cC_1 $ is a vertex and every $L$-list of codewords is connected by a hyperedge. 
We then assign colours to the hyperedges of $ \cH_{L,|\cC_1|} $. 
To this end, we quantize the probability simplex $ \Delta([q]^L) $. 
By a standard covering argument, there exists an $ \eps_2 $-net $ \cN \subset \Delta([q]^L) $ w.r.t.\ $ \norminf{\cdot} $ of size 
\begin{align}
|\cN| \le \ceil{\frac{q^L}{2\eps_2}}^{q^L} \le \paren{\frac{q^L}{2\eps_2} + 1}^{q^L} \eqqcolon c . \notag 
\end{align}
The covering property of $ \cN $ guarantees that for any $ P_{X_1,\cdots,X_L}\in\Delta([q]^L) $, there exists $ Q_{X_1,\cdots,X_L}\in\cN $ such that $ \norminf{P_{X_1,\cdots,X_L} - Q_{X_1,\cdots,X_L}} \le \eps_2 $. 
Now, viewing distributions in $\cN$ as colours, we assign a colour $ Q_{X_1,\cdots,X_L} $ to a hyperedge\footnotemark $ \{\vx_1,\cdots,\vx_L\}\in\binom{\cC_1}{L} $ iff $ \norminf{\type_{\vx_1,\cdots,\vx_L} - Q_{X_1,\cdots,X_L}}\le\eps_2 $. 
\footnotetext{We enumerate the codewords in $ \cC_1 $ according to an arbitrary fixed order. Any subset of codewords of $\cC_1$ is by default listed according to this order. }
Note that such $ Q_{X_1,\cdots,X_L} $ must exist by the construction of $ \cN $. 
In case it is not unique, pick an arbitrary one. 
In this way, we get a $ c $-colouring of all hyperedges of $ \cH_{L,|\cC_1|} $. 
By \Cref{thm:hypergraph-ramsey}, as long as $ |\cC_1|\ge N(c,L,M) $, $ \cH_{L,|\cC_1|} $ must contain as subgraph a complete $ L $-uniform hypergraph of size at least\footnotemark $ M $. 
\footnotetext{For convenience we assume that the size of this subgraph (later defined to be $\cC_2$) is \emph{exactly} $ M $. Otherwise, we throw away the last few vertices (i.e., codewords). }
The hypothesis \Cref{eqn:hyp-contradiction} guarantees that $M$ is sufficiently large:
\begin{align}
M &\ge M_0 = \max\curbrkt{ \frac{2^{11}L^7q^{2L}}{\tau^2} + L-2 , (L-1)L\paren{\frac{p_*(q,L,w)}{\frac{1}{L}\cdot\lip(g_{q,L})\cdot\eps_1} + 2} + 1} . \label{eqn:m-large}
\end{align}
Furthermore, all hyperedges in the subgraph get the same colour. 
Let us collect vertices (i.e., a subset of codewords in $ \cC_1 $) in this subgraph into a set $ \cC_2 $. 
The monochromatic property of $ \cC_2 $ implies that there exists $ Q_{X_1,\cdots,X_L}\in\cN $,
\begin{align}
\norminf{\type_{\vx_1,\cdots,\vx_L} - Q_{X_1,\cdots,X_L}}\le\eps_2
\label{eqn:equicoupled}
\end{align}
for every $ \{\vx_1,\cdots,\vx_L\}\in\binom{\cC_2}{L} $. 
This $ \cC_2 $ is our desired subcode in which all lists have roughly the same type. 

\textbf{The second step} is to argue that $ Q_{X_1,\cdots,X_L} $ in fact must be approximately \emph{symmetric}, where \emph{symmetric} means $ Q_{X_1,\cdots,X_L} = Q_{X_{\pi(1)},\cdots X_{\pi(L)}} $ for every $ \pi\in S_L $, otherwise $ |\cC_2| $ must be small. 
Let $ \type_{\vx_1,\cdots,\vx_M}\in\Delta[q]^M $ denote the joint type of all codewords in $ \cC_2 $. 
According to the first step of the proof, $\type_{\vx_1,\cdots,\vx_M}$ clearly meets the assumptions of \Cref{eqn:komlos} with distribution $ Q_{X_1,\cdots,X_L} $ and constant $ \eps_2 $. 
Therefore \Cref{eqn:komlos} guarantees that
\begin{align}
\max_{\pi\in S_L} \norminf{Q_{X_1,\cdots,X_L} - Q_{X_{\pi(1)},\cdots,X_{\pi(L)}}} 
&\le 2L^3\sqrt{\frac{2L}{M-(L-2)}} + 4L^3\sqrt{q^{L-2}\eps_2} + L^2\eps_2 \notag \\
&\le 2L^3\sqrt{\frac{2L}{M-(L-2)}} + 5L^3\sqrt{q^{L-2}\eps_2} . \label{eqn:apx-symm} 
\end{align}
In words, as long as $M$ is sufficiently large, $ Q_{X_1,\cdots,X_L} $ must be approximately symmetric. 
More specifically, 
since \Cref{eqn:m-large} guarantees 
\begin{align}
    M\ge \frac{2^{11}L^7q^{2L}}{\tau^2} + L-2 , \notag
\end{align}
by the choice of $ \eps_2 $ (cf.\ \Cref{eqn:def-eps2}), \Cref{eqn:apx-symm} implies that 
\begin{align}
\max_{\pi\in S_L} \norminf{Q_{X_1,\cdots,X_L} - Q_{X_{\pi(1)},\cdots,X_{\pi(L)}}} &\le 2L^3\sqrt{\frac{2L}{2^{11}L^7q^{2L}/\tau^2}} + 5L^3\cdot\frac{\tau}{80L^3q^{L}} 
= \frac{\tau}{8q^L} . \label{eqn:bound-asymm-by-eps3} 
\end{align}
We define the RHS to be $ \eps_3 $, 
\begin{align}
\eps_3 &\coloneqq \frac{\tau}{8q^{L}} . \label{eqn:def-eps3}
\end{align}

At this point, we have got a subcode $ \cC_2\subset\cC_1 $ with the following properties. 
\begin{enumerate}
    \item \label{itm:property1} $ |\cC_2| = M $; 

    \item \label{itm:property2} every codeword in $ \cC_2 $ has Hamming weight between $ n(w-\eps_1) $ and $ n(w+\eps_1) $; 

    \item \label{itm:property3} there exists a distribution $ Q_{X_1,\cdots,X_L}\in\Delta([q]^L) $ satisfying \Cref{eqn:apx-symm} such that \Cref{eqn:equicoupled} holds for every $L$-list in $ \cC_2 $; 

    \item \label{itm:property4} $ \cC_2 $ is $ (p,L)_q $-list-decodable. 
\end{enumerate}

\textbf{The third step} is to apply a well-known double counting trick to $ \cC_2 $. 
Let us bound from both sides the following quantity
\begin{align}
\frac{1}{M^L} \sum_{(i_1,\cdots,i_L)\in[M]^L} \ol{\rad}(\vx_{i_1}, \cdots, \vx_{i_L}) , \label{eqn:double-count-object} 
\end{align}
that is, the average radius averaged over (potentially repeated) lists in $ \cC_2 $. 

A lower bound on \Cref{eqn:double-count-object} follows from list-decodability of $ \cC_2 $. 
Indeed, suppose towards a contradiction that $ \ol{\rad}(\vx_1,\cdots,\vx_L)\le n(p_*(q,L,w) + \eps_4) $ for some list $ \{\vx_1,\cdots,\vx_L\}\in\binom{\cC_2}{L} $. 
Here we choose 
\begin{align}
\eps_4 &\coloneqq \frac{2}{L}\cdot\lip(g_{q,L})\cdot\eps_1 ,
\label{eqn:def-eps4} 
\end{align}
and $ \lip(g_{q,L}) $ will be defined in \Cref{eqn:def-lip-gql}. 
Then by the definition of average radius, 
\begin{align}
\frac{1}{L} \sum_{i = 1}^L \disth(\vx_i,\vy) &\le n(p_*(q,L,w) + \eps_4) , \label{eqn:bound-on-avgrad} 
\end{align}
where $ \vy $ is the centroid of the list whose coordinates attain the plurality of the corresponding coordinates of the list: $ y_i = \argplur(x_{1,i},\cdots,x_{L,i}) $ for every $ i\in[n] $. 
If $ \argplur(x_{1,i},\cdots,x_{L,i}) $ is not unique, take an arbitrary one. 
Now observe that for any $ 1\le i\ne j\le L $, $ \disth(\vx_i,\vy) $ and $ \disth(\vx_j,\vy) $ have approximately the same value which does not depend on the choice of $i$ and $j$. 
\begin{align}
& \frac{1}{n}\abs{\disth(\vx_i, \vy) - \disth(\vx_j, \vy)} \notag \\
&= \abs{\sum_{(x_1,\cdots,x_L,y)\in[q]^{L+1}} \type_{\vx_1,\cdots,\vx_L,\vy}(x_1,\cdots,x_L,y) \paren{\indicator{x_i \ne y} - \indicator{x_j \ne y}}} \notag \\
&= \abs{\sum_{(x_1,\cdots,x_L,y)\in[q]^{L+1}} \paren{\type_{\vx_1,\cdots,\vx_L,\vy}(x_1,\cdots,x_L,y) - \type_{\sigma_{i,j}(\vx_1,\cdots,\vx_L),\vy}(x_1,\cdots,x_L,y)} \indicator{x_i \ne y}} \label{eqn:dc-apply-swap} \\
&\le \sum_{(x_1,\cdots,x_L)\in[q]^{L}} \abs{\sum_{y\in[q]} \paren{\type_{\vx_1,\cdots,\vx_L,\vy}(x_1,\cdots,x_L,y) - \type_{\sigma_{i,j}(\vx_1,\cdots,\vx_L),\vy}(x_1,\cdots,x_L,y)}} \notag \\
&\le \normone{\type_{\vx_1,\cdots,\vx_L} - Q_{X_1,\cdots,X_L}}
    + \normone{Q_{X_1,\cdots,X_L} - Q_{\sigma_{i,j}(X_1,\cdots,X_L)}}
    + \normone{\type_{\sigma_{i,j}(\vx_1,\cdots,\vx_L)} - Q_{\sigma_{i,j}(X_1,\cdots,X_L)}} \notag \\
&\le q^{L} (2\eps_2 + \eps_3) . \label{eqn:apply-asymm} 
\end{align}
In \Cref{eqn:dc-apply-swap}, we use $ \sigma_{i,j}(\vx_1,\cdots,\vx_L) $ to denote the sequence obtained by swapping $ \vx_i $ and $ \vx_j $ in $ \vx_1,\cdots,\vx_L $. 
\Cref{eqn:apply-asymm} above follows from Property \ref{itm:property3} of $ \cC_2 $. 
This bound together with \Cref{eqn:bound-on-avgrad} implies an upper bound on $ \rad(\vx_1,\cdots,\vx_L) $. 
For $ j\in[L] $, 
\begin{align}
\frac{1}{L} \sum_{i = 1}^L \disth(\vx_i,\vy) &\ge \disth(\vx_j, \vy) - nq^{L} (2\eps_2 + \eps_3) . \notag 
\end{align}
Therefore, 
\begin{align}
\rad(\vx_1,\cdots,\vx_L) &\le \max_{i\in[L]} \disth(\vx_i,\vy) \notag \\
&\le \frac{1}{L} \sum_{i = 1}^L \disth(\vx_i,\vy) + nq^{L} (2\eps_2 + \eps_3) \notag \\
&\le n(p_*(q,L,w) + \eps_4 + q^{L} (2\eps_2 + \eps_3)) \label{eqn:use-bound-on-avgrad} \\
&= n(p_*(q,L,w) + \tau/4 + \tau/4) \label{use-eps234} \\
&< n(p_*(q,L,w) + \tau) . \label{eqn:dc-contradict} 
\end{align}
\Cref{eqn:use-bound-on-avgrad} is due to \Cref{eqn:bound-on-avgrad}.
\Cref{use-eps234} follows from the choices of $\eps_2$ (cf.\ \Cref{eqn:def-eps2}), $\eps_3$ (cf.\ \Cref{eqn:def-eps3}), $\eps_4$ (cf.\ \Cref{eqn:def-eps4}) and that 
\begin{align}
    \eps_2 &= \paren{\frac{\tau}{80L^3q^{L}}}^2 \cdot \frac{1}{q^{L-2}}
    \le \paren{\frac{\tau}{16q^{L}}}^2
    \le \frac{\tau}{16q^{L}} . \notag
\end{align}
That is, we have found a list $ \{\vx_1,\cdots,\vx_L\}\in\binom{\cC_2}{L} $ whose Chebyshev radius is at most the RHS of \Cref{eqn:dc-contradict} which is \emph{less than} $ np $. 
This contradicts $ (p,L)_q $-list-decodability of $ \cC_2 $ (cf.\ Property \ref{itm:property4}). 
Hence we can conclude that 
\begin{align}
\ol{\rad}(\vx_{i_1}, \cdots, \vx_{i_L}) &\ge n(p_*(q,L,w) + \eps_4) \label{eqn:dc-lb-term}
\end{align}
for every $ \{i_1,\cdots,i_L\}\in\binom{[M]}{L} $. 
For those lists $ (i_1,\cdots,i_L)\in[M]^L $ of indices which are not all distinct, we simply lower bound $ \ol{\rad}(\vx_{i_1},\cdots,\vx_{i_L})\ge0 $. 
Therefore, we have the following lower bound\footnotemark on \Cref{eqn:double-count-object}:
\footnotetext{A list of distinct indices $ (i_1,\cdots,i_L)\in[M]^L $ may not obey the order of codewords in $\cC_2$ (inherited from $\cC_1$). However, note that the summand $ \ol{\rad}(\vx_{i_1},\cdots,\vx_{i_L}) $ of \Cref{eqn:double-count-object} is a symmetric function and therefore does not depend on the order of $ (i_1,\cdots,i_L) $. So the bound in \Cref{eqn:dc-lb-term} also applies to distinct indices $(i_1,\cdots,i_L)$ that do not obey the order of $\cC_2$. }
\begin{align}
\frac{1}{M^L} \sum_{(i_1,\cdots,i_L)\in[M]^L} \ol{\rad}(\vx_{i_1}, \cdots, \vx_{i_L}) &\ge \frac{M(M-1)\cdots(M-L+1)}{M^L} n(p_*(q,L,w) + \eps_4) \notag \\
&\ge \frac{(M-L+1)^L}{M^L}n(p_*(q,L,w) + \eps_4) \notag \\
&= \paren{1 - \frac{L-1}{M}}^L n(p_*(q,L,w) + \eps_4) \notag \\
&\ge \paren{1 - \frac{(L-1)L}{M}} n(p_*(q,L,w) + \eps_4) . \label{eqn:dc-lb} 
\end{align}

Next, we upper bound \Cref{eqn:double-count-object}. 
\begin{align}
& \frac{1}{M^L} \sum_{(i_1,\cdots,i_L)\in[M]^L} \ol{\rad}(\vx_{i_1}, \cdots, \vx_{i_L}) \notag \\
&= \frac{1}{M^L} \sum_{(i_1,\cdots,i_L)\in[M]^L} \sum_{k\in[n]} \paren{1 - \frac{1}{L}\plur(x_{1,k},\cdots,x_{L,k})} \notag \\
&= \frac{1}{M^L} \sum_{(i_1,\cdots,i_L)\in[M]^L} \sum_{(x_1,\cdots,x_L)\in[q]^L} \indicator{x_{1,k} = x_1,\cdots,x_{L,k} = x_L} \sum_{k\in[n]} \paren{1 - \frac{1}{L}\plur(x_{1},\cdots,x_{L})} \notag \\
&= \sum_{k\in[n]} \sum_{(x_1,\cdots,x_L)\in[q]^L} \sqrbrkt{\prod_{j=1}^L \paren{\frac{1}{M}\sum_{i_j\in[M]}\indicator{x_{i_j,k} = x_j}}} \paren{1 - \frac{1}{L}\plur(x_{1},\cdots,x_{L})} \notag \\
&= \sum_{k\in[n]} \sum_{(x_1,\cdots,x_L)\in[q]^L} \paren{\prod_{j=1}^L P_k(x_j)} \paren{1 - \frac{1}{L}\plur(x_{1},\cdots,x_{L})} \label{eqn:emp-col-distr} \\
&= \sum_{k\in[n]} \sum_{\va\in\cA_{q,L}} \binom{L}{\va} \paren{\prod_{i=1}^q P_k(i)^{a_i}} \paren{1 - \frac{1}{L}\max\curbrkt{\va}} \notag \\
&= \sum_{k\in[n]} \paren{1 - \frac{1}{L}f_{q,L}(P_k)} . \label{eqn:interchange-sum-list-dec}
\end{align}
In \Cref{eqn:emp-col-distr}, we define the empirical distribution $ P_k \in\Delta([q]) $ of the $k$-th ($ k\in[n] $) column of $ \cC_2 $ as
\begin{align}
P_k(x) &= \frac{1}{M} \sum_{\vx\in\cC_2} \indicator{x_k = x} \notag 
\end{align}
for any $ x\in[q] $. 
In \Cref{eqn:interchange-sum-list-dec}, we use the multinomial theorem
\begin{align}
    \paren{\sum_{i = 1}^q p_i}^L &= \sum_{\va\in\cA_{q,L}} \binom{L}{\va} \paren{\prod_{i = 1}^q p_i^{a_i}} \notag 
\end{align}
and the fact that $ P_k $ is a probability vector.

Let $ w_k $ ($ 1\le k\le n $) denote the (fractional) weight of the $k$-th column of $ \cC_2 $, i.e., 
$w_k \coloneqq 1 - P_k(q)$. 
By Schur convexity of $ f_{q,L} $ (cf.\ \Cref{thm:schur-convex-fql}), $ f_{q,L}(P_k)\ge f_{q,L}(P_{q,w_k}) $ for every $ 1\le k\le n $. 
Therefore, 
\begin{align}
\frac{1}{n} \sum_{k\in[n]} \paren{1 - \frac{1}{L}f_{q,L}(P_k)}
\le \frac{1}{n} \sum_{k\in[n]} \paren{1 - \frac{1}{L}f_{q,L}(P_{q,w_k})} 
= \frac{1}{n} \sum_{k\in[n]} \paren{1 - \frac{1}{L}g_{q,L}(w_k)} . \label{eqn:db-ub-schur} 
\end{align}
Denote by 
\begin{align}
\ol{w} \coloneqq \frac{1}{n} \sum_{k\in[n]} w_k \label{eqn:def-w-bar} 
\end{align}
the average (fractional) weight of $ \cC_2 $. 
At this point, we would like to use convexity of $ g_{q,L} $ (cf.\ \Cref{thm:convex-gql}) to deduce 
\begin{align}
\frac{1}{n} \sum_{k\in[n]} g_{q,L}(w_k) 
\ge g_{q,L}(\ol{w}) . \label{eqn:use-gql-convex} 
\end{align}
However, convexity of $ g_{q,L} $ only holds on the interval $ [0,(q-1)/q] $ whereas it is \emph{not} guaranteed that $ w_k\in[0,(q-1)/q] $ \emph{for every} $k\in[n]$. 
Fortunately, this is not going to cause an issue due to the monotonicity property of $ g_{q,L} $ (cf. \Cref{prop:mono-gqL}). 
We can perform the following surgery on $ (w_k)_{k\in[n]} $ after which the value of $\sum_{k\in[n]} g_{q,L}(w_k)$ will increase and the value of $ \sum_{k\in[n]} w_k $ will not change. 
Suppose that not all $w_k$ lie in the interval $ [0,(q-1/q)] $. 
Then we must be able to find a pair of weights $ w_i,w_j $ ($1\le i\ne j\le n$) such that $ w_i<(q-1)/q,w_j>(q-1)/q $. 
We then replace $w_i,w_j$ with $ w_i + \eps, w_j-\eps $, respectively, where $ \eps \coloneqq \min\{(q-1)/q - w_i, w_j - (q-1)/q\} $. 
One of the new $ w_i, w_j$ becomes equal to $ (q-1)/q $ and the other one may not be equal. 
We then repeat such operations until no $ w_k>(q-1)/q $ can be found. 
This process will terminate since each codeword has Hamming weight at most $ n(q-1)/q $, so $ \frac{1}{n} \sum_{k\in[n]} w_k \le (q-1)/q $. 
With slight abuse of notation, we denote the weights after the surgery still by $ (w_k)_{k\in[n]} $. 
With the new weights, convexity of $ g_{q,L} $ can be applied and \Cref{eqn:use-gql-convex} is valid. 

Note that by Property \ref{itm:property2} of $ \cC_2 $, 
\begin{align}
\ol{w} = \frac{1}{M} \sum_{\vx\in\cC_2} \frac{1}{n}\wth(\vx) &\in [\max\{0,w-\eps_1\},\max\{w+\eps_1,(q-1)/q\}] . \notag 
\end{align}
Let
\begin{align}
\lip(g_{q,L}) &\coloneqq \max_{w\in[0,(q-1)/q]} g_{q,L}'(w) \label{eqn:def-lip-gql}
\end{align}
be the Lipschitz constant of $ g_{q,L} $. 
Then 
\begin{align}
g_{q,L}(\ol{w}) &\ge g_{q,L}(w) - \lip(g_{q,L})\cdot\eps_1 . \notag 
\end{align}
Therefore, continuing with \Cref{eqn:db-ub-schur}, we have shown that \Cref{eqn:double-count-object} is upper bounded by
\begin{align}
\frac{1}{n} \sum_{k\in[n]} \paren{1 - \frac{1}{L}g_{q,L}(w_k)}
&\le 1 - \frac{1}{L} g_{q,L}(\ol{w}) \notag \\
&\le 1 - \frac{1}{L} g_{q,L}(w) + \frac{1}{L}\cdot\lip(g_{q,L})\cdot\eps_1 \notag \\
&= p_*(q,L,w) + \frac{1}{L}\cdot\lip(g_{q,L})\cdot\eps_1 . \label{eqn:dc-ub} 
\end{align}

Finally, combining the upper (\Cref{eqn:dc-ub}) and lower (\Cref{eqn:dc-lb}) bounds on \Cref{eqn:double-count-object}, we get 
\begin{align}
\paren{1 - \frac{(L-1)L}{M}} (p_*(q,L,w) + \eps_4) 
\le p_*(q,L,w) + \frac{1}{L}\cdot\lip(g_{q,L})\cdot\eps_1 . \notag 
\end{align}
However, it is straightforward to check that the above inequality holds in the \emph{reverse} direction for any $M$ satisfying
\begin{align}
M &> \frac{(L-1)L(p_*(q,L,w) + \eps_4)}{\eps_4 - \frac{1}{L}\cdot\lip(g_{q,L})\cdot\eps_1}
= (L-1)L\paren{\frac{p_*(q,L,w)}{\frac{1}{L}\cdot\lip(g_{q,L})\cdot\eps_1} + 2}  \label{eqn:plotkin-contrad-concl}
\end{align}
which holds by the hypothesis \Cref{eqn:m-large}. 
We have reached a contradiction which implies that \Cref{eqn:hyp-contradiction} must \emph{not} hold and therefore the conclusion of \Cref{thm:plotkin-list-dec-qary} is true. 
\end{proof}

By tuning the slack factors in \Cref{thm:plotkin-list-dec-qary}, one can easily obtain an analogous bound for \emph{exactly} constant-weight list-decodable codes. 

\begin{corollary}[Plotkin bound, $ q $-ary list-decoding, constant weight]
\label{cor:plotkin-list-dec-qary-const-wt}
Let $ \cC\subset[q]^n $ be a $ (p,L)_q $-list-decodable code with $ p = p_*(q,L,w) + \tau $ for $ w\in[0,(q-1)/q] $ and some small constant $ \tau>0 $. 
Assume that every codeword in $ \cC $ has Hamming weight $ nw $. 
Then 
\begin{align}
|\cC| &< 
N(c,L,M_0) , \notag 
\end{align}
where
\begin{align}
c &\coloneqq \paren{\frac{80^2L^6q^{4L-2}}{2\tau^2} + 1}^{q^L} , \quad 
M_0 \coloneqq \max\curbrkt{ \frac{2^{11}L^7q^{2L}}{\tau^2} + L-2 , (L-1)L\paren{\frac{4p_*(q,L,w)}{\tau} + 1} + 1} , \notag 
\end{align}
and the function $ N(\cdot,\cdot,\cdot) $ denotes the Ramsey number given by \Cref{thm:hypergraph-ramsey}. 
\end{corollary}

\begin{proof}
The corollary follows by setting $ \eps_1 = 0 $ and $ \eps_4 = \tau/4 $ in \Cref{eqn:plotkin-contrad-concl} in the proof of \Cref{thm:plotkin-list-dec-qary}. 
\end{proof}

By partitioning the code into approximately constant-weight subcodes, one can obtain a bound analogous to that in \Cref{thm:plotkin-list-dec-qary} for codes \emph{without} weight constraints. 

\begin{corollary}[Plotkin bound, $ q $-ary list-decoding, weight unconstrained]
\label{cor:plotkin-list-dec-qary-unconstr}
Let $ \cC\subset[q]^n $ be a $ (p,L)_q $-list-decodable code with $ p = p_*(q,L) + \tau $ for some small constant $ \tau>0 $. 
Then 
\begin{align}
|\cC| &< 
q\paren{\frac{4\lip(g_{q,L})}{L\tau} + 1} 
N(c,L,M_0) , \notag 
\end{align}
where $ \lip(g_{q,L}) $ depends only on $q,L$ and has been defined in \Cref{eqn:def-lip-gql}, 
\begin{align}
    \begin{split}
        c &\coloneqq \paren{\frac{80^2L^6q^{4L-2}}{2\tau^2} + 1}^{q^L} , \\
        M_0 &\coloneqq \max\curbrkt{ \frac{2^{11}L^7q^{2L}}{\tau^2} + L-2 , (L-1)L\paren{\frac{8p_*(q,L)}{\tau} + 2} + 1} , 
    \end{split}
    \label{eqn:def-m0-unconstr} 
\end{align}
and the function $ N(\cdot,\cdot,\cdot) $ denotes the Ramsey number given by \Cref{thm:hypergraph-ramsey}. 
\end{corollary}

\begin{proof}
We slice $ \cC $ into a sequence subcodes each of which is almost constant-weight for some weight at most $n(q-1)/q$, so that \Cref{thm:plotkin-list-dec-qary} can be applied. 
Let $ \eps_1 $ be a small positive constant defined as
\begin{align}
\eps_1 &\coloneqq \frac{L\tau}{8\lip(g_{q,L})} , \label{eqn:def-eps1}
\end{align}
where $ \lip(g_{q,L}) $ only depends on $ q,L $ and has been defined in \Cref{eqn:def-lip-gql}. 
Let $ \cW \coloneqq 2\eps_1\bbZ\cap[0,(q-1)/q] $ be an $\eps_1$-net of the interval $ [0,(q-1)/q] $. 
For $w \in \cW $ and $u \in [q]$ define 
\begin{align*}
    \cX_{w,u} \coloneqq \curbrkt{\vx \in [q]^n: \max\{0,(w-\eps_1)n\} \leq \disth(\vx, \vu) \leq \min\{(w+\eps_1)n, n(q-1)/q\}} \ ,
\end{align*}
where $\vu = (u,u,\dots,u)$ denotes the all-$u$ vector. First, we note that  
\begin{align}
    [q]^n \subseteq \bigcup_{w\in\cW} \bigcup_{u\in[q]} \cX_{w,u} . \notag
\end{align}
Indeed, given any $\vx \in [q]^n$ we have that for some $u \in [q]$, $\disth(\vx, \vu) \leq n\frac{q-1}{q}$, and thus for some $w \in \cW$ we have 
\begin{align}
    \max\{0,(w-\eps_1)n\} \leq \disth(\vx, \vu) \leq \min\{(w+\eps_1)n, n(q-1)/q\} . 
    \label{eqn:dist-to-u}
\end{align}
Define $\cC_{w,u} \coloneqq \cC \cap \cX_{w,u}$.
For each subcode $ \cC_{w,u} $ with $ u\ne q $, we interchange $u$ and $q$ in each codeword and denote, with slight abuse of notation, the resulting subcode again by $ \cC_{w,u} $. 
This does not affect the $(p,L)_q$-list-decodability of $ \cC_{w,u} $ since list-decodability is preserved under permutation on the alphabet. 
However, \Cref{eqn:dist-to-u} then becomes a guarantee on the \emph{Hamming weight} of each codeword in the subcode:
\begin{align}
    \max\{0,(w-\eps_1)n\} \leq \wth(\vx) \leq \min\{(w+\eps_1)n, n(q-1)/q\} . 
    \label{eqn:wt-guarantee-subcode}
\end{align}

To summarize, at this point, we have obtained $K$ subcodes $ \cC_1,\cC_2,\cdots,\cC_K\subset\cC $ such that 
\begin{enumerate}
    \item $ \{\cC_i\}_{i=1}^K $ forms a covering of $\cC$, i.e., $ \bigcup_{i=1}^K \cC_i \supset \cC $;
    \item for each $ i\in[K] $, the codewords in $ \cC_i $ all have Hamming weight between $ \max\{0,n(w_i-\eps_1)\} $ and $ \min\{n(w_i+\eps_1), n(q-1)/q\} $ for some $ w_i\in[0,(q-1)/q] $; 
    \item for each $ i\in[K] $, $ \cC_i $ is $ (p,L)_q $-list-decodable; 
    \item $ K=|\cW|\cdot q\le q\ceil{\frac{(q-1)/q}{2\eps_1}}\le q\paren{\frac{1}{2\eps_1} + 1} = q\paren{\frac{4\lip(g_{q,L})}{L\tau} + 1} $. 
\end{enumerate}
Therefore, all subcodes satisfy the assumptions of \Cref{thm:plotkin-list-dec-qary} and we can apply the theorem to each of them and get
\begin{align}
    |\cC_i| &\le N(c,L,M_i) , \notag 
\end{align}
for each $ i\in[K] $, where $ c $ is defined in \Cref{eqn:def-c-m0} and 
\begin{align}
    M_i &\coloneqq \max\curbrkt{ \frac{2^{11}L^7q^{2L}}{\tau^2} + L-2 , (L-1)L\paren{\frac{p_*(q,L,w_i)}{\frac{1}{L}\cdot\lip(g_{q,L})\cdot\eps_1} + 2} + 1} . \notag
\end{align}
We conclude that 
\begin{align}
    |\cC| &\le \sum_{i=1}^K |\cC_i| 
    \le \sum_{i=1}^K N(c,L,M_i)
    \le K\cdot N(c,L,M_0) , \notag 
\end{align}
where $ M_0 $ is defined in \Cref{eqn:def-m0-unconstr}. 
In the last inequality, we use the fact that $ g_{q,L}(w) $ is minimized at $ w = (q-1)/q $ (cf.\ \Cref{prop:mono-gqL}), so $ p_*(q,L,w) = 1 - \frac{1}{L} g_{q,L}(w) $ is maximized at $ w = (q-1)/q $ which gives $ 1 - \frac{1}{L} g_{q,L}((q-1)/q) = p_*(q,L) $ (cf.\ \Cref{cor:plotkin-pt-listdec}). 
This finishes the proof. 
\end{proof}

\section{Elias--Bassalygo bound for $q$-ary list-decoding: proof of \Cref{thm:eb-qary-list-dec-main}}
\label{sec:eb-list-dec-qary-pf}

\begin{theorem}[Elias--Bassalygo bound, $q$-ary list-decoding]
\label{thm:eb-qary-list-dec}
Let $ \cC\subset[q]^n $ be a $ (p,L)_q $-list-decodable code with $ p < p_*(q,L) $.
For any sufficiently small constant $ \tau>0 $, let $ w_{q,L,\tau} $ be the unique solution to the following equation in $ w\in[0, (q-1)/q] $, 
\begin{align}
    p_*(q,L,w) &= p - \tau . \label{eqn:subcode-above-plotkin} 
\end{align}
Then 
\begin{align}
    |\cC| &\le q\paren{\frac{4\lip(g_{q,L})}{L\tau} + 1}^2 N(c,L,M_0) \cdot \paren{n\ln(q)\cdot\sqrt{8nw_{q,L,\tau}(1-w_{q,L,\tau})}\cdot q^{n(1-H_q(w_{q,L,\tau}))} + 1} , \notag 
\end{align}
where $ c,M_0 $ are defined in \Cref{cor:plotkin-list-dec-qary-unconstr}. 
\end{theorem}

\begin{proof}
Let $ \cC\subset[q]^n $ be a $ (p,L)_q $-list-decodable code with $ p < p_*(q,L) $. 
Our goal is to upper bound the size of $ |\cC| $ and therefore upper bound the list-decoding capacity. 

Let $ \tau>0 $ be a sufficiently small constant. 
Let $ w_{q,L,\tau} $ be the solution to \Cref{eqn:subcode-above-plotkin} in $ w\in[0, (q-1)/q] $. 
Now we cover the Hamming space $ [q]^n $ using Hamming balls of equal radius $ nw_{q,L,\tau} $. 
Specifically, choose $ \vc_1,\cdots,\vc_K $ for some $K$ such that 
\begin{align}
    \bigcup_{i=1}^K \bham(\vc_i,nw_{q,L,\tau}) &= [q]^n . \label{eqn:cov-prop} 
\end{align}
By a simple random construction (cf.\ \Cref{lem:cov} in \Cref{app:aux-lem}), there exists such a set of covering centers of size at most
\begin{align}
    K &\le n\ln(q)\cdot\sqrt{8nw_{q,L,\tau}(1-w_{q,L,\tau})}\cdot q^{n(1-H_q(w_{q,L,\tau}))} + 1 . \notag 
\end{align}
For each $ i\in[K] $, define a subcode
\begin{align}
    \cC_i &\coloneqq \cC\cap\bham(\vc_i,nw_{q,L,\tau}) , \notag 
\end{align}
that is, the subcode obtained by restricting $\cC$ to the Hamming ball centered around $\vc_i$ of radius $ nw_{q,L,\tau} $. 
Note that different subcodes may not be disjoint. 
Let $ \eps_1\coloneqq\frac{L\tau}{8\lip(g_{q,L})} $. 
Take an $\eps_1$-net $\cW \coloneqq 2\eps_1\bbZ\cap[0,w_{q,L,\tau}]$ of the interval $ [0,w_{q,L,\tau}] $. 
Note that $ |\cW|\le\frac{1}{2\eps_1} + 1 = \frac{4\lip(g_{q,L})}{L\tau} + 1 $. 
For each $i\in[K]$, further define
\begin{align}
    \cC_{i,w} &\coloneqq \cC_i \cap (\bham(\vc_i,\min\{n(w+\eps_1),nw_{q,L,\tau}\}) \setminus \bham(\vc_i,\max\{0,n(w-\eps_1)\})) , \notag 
\end{align}
for each $ w\in\cW $, 
that is, the subcode obtained by collecting codewords in $\cC$ whose Hamming distance to $\vc_i$ is within $n(w\pm\eps_1)$. 
We then re-center each $ \cC_{i,w} $ by permuting the symbols in $\cC_{i,w}$ so that $\cC_{i,w}\subset\bham(\vq,\min\{n(w+\eps_1),nw_{q,L,\tau}\}) \setminus \bham(\vq,\max\{0,n(w-\eps_1)\})$.
Denote, with slight abuse of notation, the resulting subcode again by $\cC_{i,w}$.
Now each codeword in $\cC_{i,w}$ has Hamming weight between $ \max\{0,n(w-\eps_1)\} $ and $ \min\{n(w+\eps_1),nw_{q,L,\tau}\} $. 
This operation does not affect list-decodability of $\cC_{i,w}$. 
For any $w\in\cW$, 
\begin{align}
    p_*(q,L,w) + \tau &\le p_*(q,L,w_{q,L,\tau}) + \tau = p . \notag 
\end{align}
The inequality follows since $(i)$ by \Cref{prop:mono-gqL}, $ p_*(q,L,w) $ is non-decreasing on $ [0,(q-1)/q] $ and non-increasing on $ [(q-1)/q,1] $, and $(ii)$ $ w\le w_{q,L,\tau}\le(q-1)/q $. 
The equality is by the choice of $ w_{q,L,\tau} $ (cf.\ \Cref{eqn:subcode-above-plotkin}). 
Therefore, each $ \cC_{i,w} $ satisfies the assumptions of \Cref{thm:plotkin-list-dec-qary} and has size at most
\begin{align}
    |\cC_{i,w}| &\le q\paren{\frac{4\lip(g_{q,L})}{L\tau} + 1} N(c,L,M_0) \eqqcolon M_* , \notag 
\end{align}
where $ c,M_0 $ are defined in \Cref{eqn:def-m0-unconstr} and we define $ M_* $ as the constant on the RHS. 
The covering property \Cref{eqn:cov-prop} implies
\begin{align}
    \bigcup_{i=1}^K \bigcup_{w\in\cW} \cC_{i,w} &= \cC . \notag 
\end{align}
Consequently, 
\begin{align}
    |\cC| &\le \sum_{i=1}^K \sum_{w\in\cW} |\cC_{i,w}| \le K\cdot|\cW|\cdot M_* \notag 
\end{align}
which finishes the proof. 
\end{proof}

\section{Schur convexity of $ f_{q,L,\ell} $: proof of \Cref{thm:schur-convex-fqll}}
\label{sec:proof-schur-convex-fqll}
The proof of \Cref{thm:schur-convex-fqll} follows similar ideas to those used in \Cref{thm:schur-convex-fql} with suitable adjustments.
\begin{proof}[Proof of \Cref{thm:schur-convex-fqll}]
 
The proof of \Cref{thm:schur-convex-fqll} again follows from verifying the Schur--Ostrowski criterion (\Cref{thm:schur-criterion}). 

First, by \Cref{eq:fqL'-exp} of \Cref{lemma:derivs}, we have 
\begin{align}
\frac{\partial}{\partial p_j} f_{q,L,\ell}(p_1,\cdots,p_q)
&= L \sum_{\va\in\cA_{q,L}} \binom{L-1}{\va} \maxl\curbrkt{\va + \ve_j} \paren{\prod_{i = 1}^qp_i^{a_i}} . \notag 
\end{align}
for all $j \in [q]$.

Let now $ 1\le k\ne j\le q $ and assume $ p_k > p_j $. 
The goal is to show $ \paren{\frac{\partial}{\partial p_k} - \frac{\partial}{\partial p_j}}f_{q,L,\ell}\ge0 $. 
We have 
\begin{align}
\paren{\frac{\partial}{\partial p_k} - \frac{\partial}{\partial p_j}}f_{q,L,\ell}(p_1,\cdots,p_q) 
&= L \sum_{\va\in\cA_{q,L}} \binom{L-1}{\va} \paren{\prod_{i = 1}^qp_i^{a_i}} \paren{\maxl\curbrkt{\va + \ve_k} - \maxl\curbrkt{\va + \ve_j}} . \notag 
\end{align}
Now observe that
\begin{align}
\maxl\curbrkt{\va + \ve_k} - \maxl\curbrkt{\va + \ve_j} &= \begin{cases}
1, & k\in\argmaxl\curbrkt{\va}, j\notin\argmaxl\curbrkt{\va} \\
-1, & k\notin\argmaxl\curbrkt{\va}, j\in\argmaxl\curbrkt{\va} \\
0, & \ow
\end{cases}
\end{align}
where we have defined 
\begin{align}
\argmaxl\curbrkt{\va} &\coloneqq \argmax_{\cQ\in\binom{[q]}{\ell}} \sum_{i\in\cQ} a_i \notag
\end{align}
to be the $\ell$-subset that achieves $\maxl\{\va\}$. 

It follows, via similar manipulations to \Cref{eqn:def-swap}, that 
\begin{align}
\frac{1}{L}\paren{\frac{\partial}{\partial p_k} - \frac{\partial}{\partial p_j}}f_{q,L,\ell}(p_1,\cdots,p_q) 
&= \sum_{\substack{\va\in\cA_{q,L} \\ k\in\argmaxl\curbrkt{\va} \\ j\notin\argmaxl\curbrkt{\va}}} \binom{L-1}{\va} \paren{\prod_{i\in[q]\setminus\{k,j\}}p_i^{a_i}} \paren{p_k^{a_k} p_j^{a_j} - p_k^{a_j} p_j^{a_k}} . \notag 
\end{align}
Finally, the above expression can be seen positive due to the assumption $ p_k>p_j $ and the condition $ a_k\ge a_j $. 
\end{proof}

\section{Convexity of $ g_{q,L,\ell} $: proof of \Cref{thm:convex-gqll}}
\label{sec:proof-convex-gqll}

\Cref{thm:convex-gqll} can be proved in an analogous way to \Cref{thm:convex-gql} by showing $ g_{q,L,\ell}''\ge0 $ with the adjustment of replacing $ \max $ with $ \maxl $. And indeed, quite pleasingly, the argument actually becomes \emph{simpler} when $\ell \geq 2$.
\begin{proof}[Proof of \Cref{thm:convex-gqll}]
By \Cref{eq:exp-gqL''} from \Cref{lemma:derivs}, we have
\begin{align*}
    g_{q,L,\ell}''(w) = L(L-1)\paren{\frac{w}{q-\ell}}^{L-2}\sum_{\va \in \cA_{q,L-2}} \binom{L-2}{\va}\paren{\frac{(q-\ell)(1-w)}{\ell w}}^{a_{q-\ell+1}+\cdots+a_q}G_\ell(\va) 
\end{align*}
where
\begin{align*}
    G_\ell(\va) &= \frac{1}{(q-\ell)^2}\sum_{1\leq i,j \leq q-\ell}\maxl\{\va+\ve_i+\ve_j\} - \frac{2}{(q-\ell)\ell}\sum_{i=1}^{q-\ell}\sum_{j=q-\ell+1}^q\maxl\{\va+\ve_i+\ve_j\} \\
    &\quad +\frac{1}{\ell^2}\sum_{q-\ell+1\leq i,j \leq q}\maxl\{\va+\ve_i+\ve_j\}.
\end{align*}
In this case, it turns out that $G_\ell(\va) \geq 0$ for all $\va \in \cA_{q,L-2}$; this clearly suffices to derive the non-negativity of $g_{q,L,\ell}''(w)$ for all $w \in \left[0,1\right]$.

We first observe that
$$
\sum_{1\leq i,j\leq q-
\ell}\frac{1}{(q-\ell)^2}-\sum_{i=1}^{q-\ell}\sum_{j=q-\ell+1}^q \frac{2}{(q-\ell)\ell}+\sum_{q-\ell+1\leq i,j\leq q}\frac{1}{\ell^2}=0 \ .
$$
Therefore
\begin{align}
    G_\ell(\va)&=\frac{1}{(q-\ell)^2}\sum_{i,j \in [q-\ell]}\paren{\maxl\{\va+\ve_i+\ve_j\}-\maxl\{\va\}} \notag \\
    &- \frac{2}{(q-\ell)(\ell)}\sum_{i=1 }^{q-\ell}\sum_{j=q-\ell+1}^q\paren{\maxl\{\va+\ve_i+\ve_j\}-\maxl\{\va\}} \notag \\
    &+ \frac{1}{\ell^2}\sum_{q-\ell+1\leq i,j \leq q}\paren{\maxl\{\va+\ve_i+\ve_j\}-\maxl\{\va\}}; \label{eqn:Gell-zero}
\end{align}
as before, we lower bound each of the three terms in the above expression. 

Fix $j \in [q]$ such that $a_j$ is an $\ell$-th largest element in $\va$. Let $S=\{i\in [q]: a_i\geq a_j, a_i\in \va\}$. By the definition of $a_j$, we have $|S|\geq \ell$. Let $S_1=S\cap [q-\ell]$, $S_2=S-S_1$. Assume that the size of $S_1,S_2$ are $r_1,r_2$ respectively.

First, we claim
\begin{align}
    \sum_{1\leq i,j \leq q-\ell}\paren{\maxl\{\va+\ve_i+\ve_j\} - \maxl\{\va\}} \geq 2r_1^2 + 2r_1(q-\ell-r_1) = 2r_1(q-\ell) .
\end{align}
Indeed, if $i,j\in S_1$ then $\maxl\{\va+\ve_i+\ve_j\} - \maxl\{\va\}=2$; if $i\in S_1$ and $j\notin S_1$ then $\maxl\{\va+\ve_i+\ve_j\} - \maxl\{\va\}=1$; and if $i\notin S_1$ and $j\in S_1$ then $\maxl\{\va+\ve_i+\ve_j\} - \maxl\{\va\}=1$. As we have $\maxl\{\va+\ve_i+\ve_j\} - \maxl\{\va\}\geq 0$ for all $i,j$, the claimed lower bound follows.

Next,
\begin{align*}
    \sum_{i=1 }^{q-\ell}\sum_{j=q-\ell+1}^q\paren{\maxl\{\va+\ve_i+\ve_j\}-\maxl\{\va\}} = 2r_1r_2+r_1(\ell-r_2)+r_2(q-\ell-r_1)=r_1\ell+r_2(q-\ell),
\end{align*}
as $\maxl\{\va+\ve_i+\ve_j\}-\maxl\{\va\}=2$ if $i\in S_1,j\in S_2$; $\paren{\maxl\{\va+\ve_i+\ve_j\}-\maxl\{\va\}}=1$ if $i\in S_1,j\notin S_2$ or $i\notin S_1,j\in S_2$; and otherwise $\paren{\maxl\{\va+\ve_i+\ve_j\}-\maxl\{\va\}}=0$.
Finally,
\begin{align}
    \sum_{q-\ell+1\leq i,j \leq q}\paren{\maxl\{\va+\ve_i+\ve_j\} - \maxl\{\va\}} \geq 2r_2^2 + 2r_2(\ell-r_2) = 2r_2\ell.
\end{align}
This is because if $i,j\in S_2$ then $\maxl\{\va+\ve_i+\ve_j\} - \maxl\{\va\}=2$; if $i\in S_2$ and $j\notin S_2$ then $\maxl\{\va+\ve_i+\ve_j\} - \maxl\{\va\}=1$; and if $i\notin S_2$ and $j\in S_2$ then $\maxl\{\va+\ve_i+\ve_j\} - \maxl\{\va\}=1$.
Thus, $G_\ell(\va) \geq \frac{1}{(q-\ell)^2}2r_1(q-\ell) - \frac{2r_1\ell+2r_2(q-\ell)}{(q-\ell)\ell}+\frac{1}{\ell^2}2r_2\ell =0$, as claimed.
\end{proof}

\begin{remark} \label{rem:ell-2-convexity}
    For $\ell \geq 2$, note that we had an easier time understanding how $\max_\ell\{\va + \ve_i + \ve_j\}$ can differ from $\max_\ell\{\va + \ve_i + \ve_j\}$ than we did for the $\ell=1$ case. This makes sense as $\va+\ve_i+\ve_j$ increments the count of up to 2 alphabet symbols, so (if necessary) both can be included in a set achieving  $\argmaxl\{\va+\ve_i+\ve_j\}$. This is responsible for the simplification in the argument, and more importantly yields convexity on the entire interval $[0,1]$ (which we recall is false in general for the $\ell=1$ case). 
\end{remark}

\section{Plotkin bound for list-recovery: proof of \Cref{thm:plotkin-list-rec-main}}
\label{sec:plotkin-list-rec}

\begin{theorem}[Plotkin bound, list-recovery]
\label{thm:plotkin-list-rec}
Assume $ 2\le\ell\le q-1 $ is an integer. 
Let $ \cC\subset[q]^n $ be a $ (p,\ell,L)_q $-list-recoverable code with $ p = p_*(q,\ell,L,w) + \tau $ for $ w\in[0,1] $ and some small constant $ \tau\in(0,1) $. Let 
\begin{align}
    0<\eps_1<\frac{L\tau}{8\lip(g_{q,L,\ell})} , \notag 
\end{align}
where $ \lip(g_{q,L,\ell}) $ denotes the Lipschitz constant of $ g_{q,L,\ell} $:
\begin{align}
    \lip(g_{q,L,\ell}) &\coloneqq \max_{w\in[0,1]} g'_{q,\ell,L}(w) . \label{eqn:lip-gqll} 
\end{align}
Assume that every codeword in $\cC$ has list-recovery weight between $ n(w - \eps_1) $ and $ n(w + \eps_1) $. 
Then 
\begin{align}
    |\cC| &< N(c,L,M_0) 
    \label{eqn:plotkin-bound-list-rec}
\end{align}
where 
\begin{align}
\begin{split}
    c &\coloneqq \paren{\frac{80^2L^6q^{4L-2}}{2\tau^2} + 1}^{q^L} , \\ 
    M_0 &\coloneqq \max\curbrkt{ \frac{2^{11}L^7q^{2L}}{\tau^2} + L-2 , (L-1)L\paren{\frac{p_*(q,\ell,L,w)}{\frac{1}{L}\cdot\lip(g_{q,L,\ell})\cdot\eps_1} + 2} + 1} ,
\end{split}
\label{eqn:def-c-m0-list-rec}
\end{align}
and the function $ N(\cdot,\cdot,\cdot) $ denotes the Ramsey number given by \Cref{thm:hypergraph-ramsey}. 
\end{theorem}

\begin{remark}
Unlike the list-decoding version (cf.\ \Cref{thm:plotkin-list-dec-qary}), $w$ here is \emph{not} required to lie in $ [0,(q-\ell)/q] $ which seems to be a natural analogue to the condition in \Cref{thm:plotkin-list-dec-qary}. 
Instead, the proof below works for any $w\in[0,1]$ thanks to \Cref{thm:convex-gqll} which guarantees convexity of $g_{q,L,\ell}$ on the whole $[0,1]$ (see \Cref{eqn:use-conv-lr} in the proof for more details). 
Note that the latter statement is not true for list-decoding (cf.\ \Cref{thm:convex-gql}). 
\end{remark}

\begin{proof}[Proof of \Cref{thm:plotkin-list-rec-main}.]
The proof largely follows the structure of the proof of \Cref{thm:plotkin-list-dec-qary}. 
Suppose $ p = p_*(q,\ell,L,w) + \tau $ and we are given a $ (p,\ell,L)_q $-list-recoverable code $ \cC_1\subset[q]^n $ all of whose codewords have list-recovery weight between $ n(w - \eps_1) $ and $ n(w + \eps_1) $. 
The goal is to show that $ |\cC_1|\le M_* $ for some $ M_* $ depending only on $ q,\ell,L,w,\tau $, but not on $n$. 
Suppose towards a contradiction that \Cref{eqn:plotkin-bound-list-rec} holds in the reverse direction. 
Then according to the \textbf{first} and \textbf{second} steps in the proof of \Cref{thm:plotkin-list-dec-qary}, we can extract a subcode $ \cC_2\subset\cC_1 $ satisfying the following properties
\begin{enumerate}
    \item $ |\cC_2| = M \ge M_0 $ where $M_0$ is defined in \Cref{eqn:def-c-m0-list-rec}. 
    
    \item every codeword in $ \cC_2 $ has list-recovery weight between $ n(w-\eps_1) $ and $ n(w + \eps_1) $; 
    
    \item there exists a distribution $ Q_{X_1,\cdots,X_L}\in\Delta([q]^L) $ satisfying \Cref{eqn:apx-symm,eqn:bound-asymm-by-eps3,eqn:def-eps3} such that \Cref{eqn:equicoupled} holds for every $L$-list in $ \cC_2 $ with $\eps_2,\eps_3$ to be determined in \Cref{eqn:list-rec-determine-eps234}; 
    \item $ \cC_2 $ is $ (p,\ell,L)_q $-list-recoverable. 
\end{enumerate}
The rest of the proof is devoted to a pair of upper and lower bounds on 
\begin{align}
    \frac{1}{M^L} \sum_{(i_1,\cdots,i_L)\in[M]^L} \ol{\rad}_\ell(\vx_{i_1}, \cdots, \vx_{i_L}) , 
    \label{eqn:list-rec-dc}
\end{align}
which in turn implies an upper bound on $ M $. 

A lower bound on \Cref{eqn:list-rec-dc} follows from the list-recoverability guarantee of $ \cC_2 $. 
Indeed, suppose towards a contradiction that $ \ol{\rad}_\ell(\vx_1,\cdots,\vx_L)\le n(p_*(q,\ell,L,w) + \eps_4) $ for some list $ \{\vx_1,\cdots,\vx_L\}\in\binom{\cC_2}{L} $ and $ \eps_4 $ will be determined in \Cref{eqn:list-rec-determine-eps234}. 
Then by the definition of $\ell$-average radius, 
\begin{align}
    \frac{1}{L} \sum_{i = 1}^L \distlr(\vx_i,\vcY) &\le n(p_*(q,\ell,L,w) + \eps_4) \label{eqn:upper-bound-avgrad-lr} 
\end{align}
where $ \vcY\in\binom{[q]}{\ell}^n $ is defined as $ \cY_i \coloneqq \argplur_\ell(x_{1,i},\cdots,x_{L,i}) $ for every $ i\in[n] $. 
Due to the equicoupledness property and the approximate symmetry of $ Q_{X_1,\cdots,X_L} $, the distance from any codeword $ \vx_i $ to $ \vcY $ is approximately the same. 
\begin{align}
    & \frac{1}{n} \abs{ \distlr(\vx_i,\vcY) - \distlr(\vx_j,\vcY) } \notag \\
    &= \abs{\sum_{(x_1,\cdots,x_L)\in[q]^L} \sum_{\cY\in\binom{[q]}{\ell}} \type_{\vx_1,\cdots,\vx_L,\vcY}(x_1,\cdots,x_L,\cY) \paren{\indicator{x_i\notin\cY} - \indicator{x_j\notin\cY}}} \notag \\
    &= \abs{\sum_{(x_1,\cdots,x_L)\in[q]^L} \sum_{\cY\in\binom{[q]}{\ell}} \paren{\type_{\vx_1,\cdots,\vx_L,\vcY}(x_1,\cdots,x_L,\cY) - \type_{\sigma_{i,j}(\vx_1,\cdots,\vx_L),\vcY}(x_1,\cdots,x_L,\cY)} \indicator{x_i\notin\cY}} \notag \\
    &\le \sum_{(x_1,\cdots,x_L)\in[q]^L}  \abs{\sum_{\cY\in\binom{[q]}{\ell}} \paren{\type_{\vx_1,\cdots,\vx_L,\vcY}(x_1,\cdots,x_L,\cY) - \type_{\sigma_{i,j}(\vx_1,\cdots,\vx_L),\vcY}(x_1,\cdots,x_L,\cY)}} \notag \\
    &\le \normone{\type_{\vx_1,\cdots,\vx_L} - Q_{X_1,\cdots,X_L}}
    + \normone{Q_{X_1,\cdots,X_L} - Q_{\sigma_{i,j}(X_1,\cdots,X_L)}}
    + \normone{\type_{\sigma_{i,j}(\vx_1,\cdots,\vx_L)} - Q_{\sigma_{i,j}(X_1,\cdots,X_L)}} \notag \\
    &\le q^L (2\eps_2 + \eps_3) . \label{eqn:lr-dist-close} 
\end{align}
We can then deduce an upper bound on the \emph{maximum} radius of $ \{\vx_1,\cdots,\vx_L\} $. 
\begin{align}
    \rad(\vx_1,\cdots,\vx_L) &\le \max_{i\in[L]} \distlr(\vx_i,\vcY) \notag \\
    &\le \frac{1}{L} \sum_{i=1}^L \distlr(\vx_i,\vcY) + nq^L(2\eps_2 + \eps_3) \label{eqn:use-dist-close} \\
    &\le n(p_*(q,\ell,L,w) + \eps_4 + q^L(2\eps_2 + \eps_3)) \label{eqn:use-upper-bound-avgrad-lr} \\
    &< n(p_*(q,\ell,L,w) + \tau) . \label{eqn:atmost-tau}
\end{align}
\Cref{eqn:use-dist-close,eqn:use-upper-bound-avgrad-lr} are due to \Cref{eqn:lr-dist-close} and \Cref{eqn:upper-bound-avgrad-lr}, respectively. 
\Cref{eqn:atmost-tau} follows from the definitions of $ \eps_2,\eps_3,\eps_4 $ in \Cref{eqn:list-rec-determine-eps234}. 
However, the existence of a list with radius strictly smaller than $np$ violates the list-recoverability guarantee of $ \cC_2 $. 
We therefore have 
\begin{align}
    \ol{\rad}_\ell(\vx_{i_1},\cdots,\vx_{i_L}) &\ge n(p_*(q,\ell,L,w) + \eps_4) \notag 
\end{align}
for every $ \{i_1,\cdots,i_L\}\in\binom{[M]}{L} $. 
The desired lower bound on \Cref{eqn:list-rec-dc} follows from the same calculations leading to \Cref{eqn:dc-lb}:
\begin{align}
    \frac{1}{M^L} \sum_{(i_1,\cdots,i_L)\in[M]^L} \ol{\rad}(\vx_{i_1}, \cdots, \vx_{i_L})
    &\ge \paren{1 - \frac{(L-1)L}{M}} n(p_*(q,\ell,L,w) + \eps_4) . \label{eqn:dc-lr-lb} 
\end{align}

We then turn to deriving an upper bound on \Cref{eqn:list-rec-dc}. 
Repeating the calculations leading to \Cref{eqn:interchange-sum-list-dec} with $ \plur $ and $ \max $ replaced by $ \plur_\ell $ and $ \max_\ell $, respectively, one has
\begin{align}
    & \frac{1}{M^L} \sum_{(i_1,\cdots,i_L)\in[M]^L} \frac{1}{n}\cdot\ol{\rad}_\ell(\vx_{i_1}, \cdots, \vx_{i_L}) \notag \\
    &= \frac{1}{n} \sum_{k\in[n]} \sum_{(x_1,\cdots,x_L)\in[q]^L} \paren{\prod_{j=1}^L P_k(x_j)} \paren{1 - \frac{1}{L}\plur_\ell(x_{1},\cdots,x_{L})} \notag \\
    &= \frac{1}{n} \sum_{k\in[n]} \sum_{\va\in\cA_{q,L}} \binom{L}{\va} \paren{\prod_{i=1}^q P_k(i)^{a_i}} \paren{1 - \frac{1}{L}\maxl\curbrkt{\va}} \notag \\
    &= \frac{1}{n} \sum_{k=1}^n \paren{1 - \frac{1}{L} f_{q,\ell,L}(P_k)} \notag \\
    &\le \frac{1}{n} \sum_{k=1}^n \paren{1 - \frac{1}{L} f_{q,\ell,L}(P_{k,\ell,w_k})} \label{eqn:use-schur-conv-lr} \\
    &= \frac{1}{n} \sum_{k=1}^n \paren{1 - \frac{1}{L} g_{q,L,\ell}(w_k)} \notag \\
    &\le 1 - \frac{1}{L} g_{q,L,\ell}(\ol{w}) \label{eqn:use-conv-lr} \\
    &\le p_*(q,\ell,L,w) + \frac{1}{L}\cdot\lip(g_{q,L,\ell})\cdot\eps_1 . \label{eqn:dc-lr-ub}
\end{align}
where $ P_k\in\Delta([q]) $ denotes the empirical distribution (a.k.a.\ type) of the $k$-th column of $ \cC_2 $, 
\begin{align}
    w_k &\coloneqq \sum_{x = 1}^{q-\ell} P_k(x) \notag 
\end{align}
denotes the (relative) list-recovery weight of the $k$-th column of $ \cC_2 $ and $ \ol{w} $ is defined in \Cref{eqn:def-w-bar}. 
In \Cref{eqn:use-schur-conv-lr,eqn:use-conv-lr} we use the Schur convexity of $ f_{q,\ell,L} $ (guaranteed by \Cref{thm:schur-convex-fqll}) and convexity of $ g_{g,\ell,L} $ (guaranteed by \Cref{thm:convex-gqll}), respectively. 
(Here it is helpful to recall the definitions of $ P_{q,\ell,w} $ and $ g_{q,L,\ell} $, see \Cref{def:p-q-ell-rho,eqn:def-gqll}). 
Note that here we do not need to further massage the weights $ (w_k)_{k\in[n]} $ (see the paragraph following \Cref{eqn:use-gql-convex}) thanks to the fact that convexity of $ g_{q,L,\ell} $ holds on the whole interval $ [0,1] $ for $ \ell\ge2 $ (cf.\ \Cref{thm:convex-gqll}). 

Finally, choosing
\begin{align}
    \eps_2 &= \frac{\tau^2}{80^2L^6q^{3L-2}} , \quad 
    \eps_3 = \frac{\tau}{8q^L} , \quad 
    \eps_4 = \frac{2}{L}\lip(g_{g,\ell,L})\eps_1 \label{eqn:list-rec-determine-eps234} 
\end{align}
allows us to conclude from \Cref{eqn:dc-lr-lb,eqn:dc-lr-ub} that 
\begin{align}
    M &\le (L-1)L\paren{\frac{p_*(q,\ell,L,w)}{\frac{1}{L}\cdot\lip(g_{q,L,\ell})\cdot\eps_1} + 2} \notag
\end{align}
which contradicts our assumption that \Cref{eqn:plotkin-bound-list-rec} holds in the reverse direction. 
This finishes the proof. 
\end{proof}

\begin{corollary}[Plotkin bound, list-recovery, weight unconstrained]
\label{cor:plotkin-list-rec-unconstr}
Assume $ 2\le\ell\le q-1 $ is an integer. 
Let $ \cC\subset[q]^n $ be a $ (p,\ell,L)_q $-list-recoverable code with $ p = p_*(q,\ell,L) + \tau $ for some small constant $ \tau>0 $. 
Then 
\begin{align}
    |\cC| &< \paren{\frac{4\lip(g_{q,L,\ell})}{L\tau} + 1} N(c,L,M_0) \notag 
\end{align}
where $ \lip(g_{q,L,\ell}) $ is defined in \Cref{eqn:lip-gqll}, 
\begin{align}
    c &\coloneqq \paren{\frac{80^2L^6q^{4L-2}}{2\tau^2} + 1}^{q^L} , \notag \\
    M_0 &\coloneqq \max\curbrkt{ \frac{2^{11}L^7q^{2L}}{\tau^2} + L-2 , (L-1)L\paren{\frac{8p_*(q,\ell,L,w)}{\tau} + 2} + 1} , \notag 
\end{align}
and the function $N(\cdot,\cdot,\cdot)$ denotes the Ramsey number given by \Cref{thm:hypergraph-ramsey}. 
\end{corollary}

\begin{proof}
The proof follows the same idea used in the proof of \Cref{cor:plotkin-list-dec-qary-unconstr} with minor modifications.
We partition the code according to list-recovery weight and apply \Cref{thm:plotkin-list-rec} to each subcode. 
A proof sketch is presented below. 

Let $\eps_1 = \frac{L\tau}{8\lip(g_{q,L,\ell})}$ and $ \cW \coloneqq 2\eps_1\bbZ\cap[0,1] $.
Define for each $w\in\cW$,
\begin{align}
    \cC_w &\coloneqq \curbrkt{\vx\in\cC : \max\{0,(w-\eps_1)n\} \le \wtlr(\vx) \le \min\{(w+\eps_1)n, n\}} . \notag 
\end{align}
Each $\cC_w$ satisfies the condition in \Cref{thm:plotkin-list-rec} and therefore has size at most $ |\cC_w|\le N(c,L,M_w) $ where $M_w$ equals
\begin{align}
    M_w &\coloneqq \max\curbrkt{ \frac{2^{11}L^7q^{2L}}{\tau^2} + L-2 , (L-1)L\paren{\frac{p_*(q,\ell,L,w)}{\frac{1}{L}\cdot\lip(g_{q,L,\ell})\cdot\eps_1} + 2} + 1} . \notag 
\end{align}
We conclude that 
\begin{align}
    |\cC| &\le \sum_{w\in\cW} |\cC_w| \le |\cW|\cdot N(c,L,M_0)
    \le \paren{\frac{1}{2\eps_1} + 1} N(c,L,M_0) , \notag 
\end{align}
where the middle inequality follows from the monotonicity of $ g_{q,L,\ell}(w) $. 
Specifically, $ g_{q,L,\ell}(w) $ is non-increasing on $ [0,(q-\ell)/q] $ and non-decreasing on $ [(q-\ell)/q,1] $. 
This follows from the Schur convexity of $ f_{q,L,\ell} $. 
The argument is analogous to the proof of \Cref{prop:mono-gqL} and is therefore omitted. 
The proof of \Cref{cor:plotkin-list-rec-unconstr} is then finished. 
\end{proof}

\section{Elias--Bassalygo bound for list-recovery: proof of \Cref{thm:eb-list-rec-main}}
\label{sec:eb-list-rec-pf}

\begin{theorem}[Elias--Bassalygo bound, list-recovery]
\label{thm:eb-list-rec}
Let $ \cC\subset[q]^n $ be a $ (p,\ell,L)_q $-list-recoverable code with $ p<p_*(q,\ell,L) $. 
For any sufficiently small constant $ \tau\in(0,1) $, let $ w_{q,\ell,L,\tau} $ be the unique solution to the following equation in $ w\in[0,(q-\ell),q] $, 
\begin{align}
    p_*(q,\ell,L,w) = p - \tau . 
    \label{eqn:rho-eb-lr}
\end{align}
Then 
\begin{align}
    |\cC| &\le \paren{\frac{4\lip(g_{q,L,\ell})}{L\tau}}^2 N(c,L,M_0) \cdot \paren{n\ln(q) \cdot \sqrt{8nw_{q,\ell,L,\tau}(1-w_{q,\ell,L,\tau})} \cdot q^{n(1-H_{q,\ell}(w_{q,\ell,L,\tau}))} + 1} , \notag 
\end{align}
where $ c,M_0 $ are defined in \Cref{cor:plotkin-list-rec-unconstr}. 
\end{theorem}

The proof follows verbatim that of \Cref{thm:eb-qary-list-dec} with the application of \Cref{lem:cov} and \Cref{thm:plotkin-list-dec-qary} replaced with \Cref{lem:cov-list-rec} (proved in \Cref{app:aux-lem}) and \Cref{thm:plotkin-list-rec}, respectively. 
The details are omitted. 

\begin{remark}
\label{rk:eb-hashing}
We specialize the bound in \Cref{thm:eb-list-rec-main} to the case $ p = 0, \ell \le q-1, L = \ell+1 $ corresponding to $ (q,\ell+1) $-hashing. 
The fixed point equation $ p_*(q,\ell,L,w) = p $ becomes
\begin{align}
    0 &= 1 - \frac{1}{\ell+1} \exptover{(X_1,\cdots,X_{\ell+1})\sim P_{q,\ell,w}^{\ot (\ell+1)}}{\plur_\ell(X_1,\cdots,X_{\ell+1})} , \notag 
\end{align}
or, 
\begin{align}
    \exptover{(X_1,\cdots,X_{\ell+1})\sim P_{q,\ell,w}^{\ot (\ell+1)}}{\plur_\ell(X_1,\cdots,X_{\ell+1})} &= \ell+1 , \label{eqn:fixed-pt-eqn-hashing} 
\end{align}
where $ P_{q,\ell,w}\in\Delta([q]) $ (recall \Cref{def:p-q-ell-rho}) is a probability vector whose first $q - \ell$ entries are all equal to $ \frac{w}{q - \ell} $ and rest entries are all equal to $ \frac{1-w}{\ell} $. 
We see that the only way to ensure the expected $\ell$-plurality of $\ell+1$ (random) symbols to be $ \ell+1 $ is to set $ w = 0 $ so that each $X_i$ can only take $\ell$ possible values (specifically, $ q-\ell + 1,\cdots,q $) with positive probability. 
Therefore, the unique solution to \Cref{eqn:fixed-pt-eqn-hashing} is $ w_{q,\ell,\ell+1} = 0 $. 
The capacity upper bound in \Cref{thm:eb-list-rec-main} then becomes
\begin{align}
    1 - H_{q,\ell}(0) = 1 - \log_q\ell = \log_q\frac{q}{\ell} . \notag 
\end{align}
Further setting $ \ell = q-1 $ yields $ \log_q\frac{q}{q-1} $ which recovers the upper bound due to \korner and Marton for $q$-hashing \cite{KM} (see also \Cref{eqn:korner-marton-upper}). 
\end{remark}


\section{Lower bound on list-recovery capacity: proof of \Cref{thm:lb-listrec}}
\label{sec:list-rec-lb}

We prove a slightly rephrased version of \Cref{thm:lb-listrec} below. 

\begin{theorem}[Random coding with expurgation lower bound]
\label{thm:lb}
Let $ q,\ell,L $ be positive integers such that $ q,L\ge2,\ell\le q $. 
Suppose $ p< p_*(q,\ell,L) $. 
Let $ ( X_1,\cdots, X_L)\sim\unif([q])^{\ot L} $. 
Then
\begin{align}
C_{(p,\ell,L)_q} &\ge \frac{1}{L-1}\paren{-\lambda_*p - \log_q\expt{\exp_q\paren{-\lambda_*\cdot\ol{\rad}_\ell( X_1,\cdots, X_L)}}}, \label{eqn:lb-lr} 
\end{align}
where $ \lambda_* = \lambda_*(q,\ell,L,p) $ is the solution to the following equation
\begin{align}
p &= \frac{\expt{\exp_q\paren{-\lambda_*\cdot\ol{\rad}_\ell( X_1,\cdots, X_L)}\cdot\ol{\rad}_\ell( X_1,\cdots, X_L)}}{\expt{\exp_q\paren{-\lambda_*\cdot\ol{\rad}_\ell( X_1,\cdots, X_L)}}} , \label{eqn:lb-lr-fp-eqn}
\end{align}
and 
\begin{align}
    \ol{\rad}_\ell(x_1,\cdots,x_L) &\coloneqq 1 - \frac{1}{L} \plur_\ell(x_1,\cdots,x_L) \notag 
\end{align}
for any $ (x_1,\cdots,x_L)\in[q]^L $. 
\end{theorem}

\begin{remark}
\label{rk:list-rec-lb-recover}
Expanding out the expectation, we can write the above bound in a more explicit (though less informative) way. 
\begin{align}
C_{(p,\ell,L)_q} &\ge \frac{L}{L-1} - \frac{1}{L-1}\curbrkt{\lambda_*p + \log_q\sqrbrkt{\sum_{\va\in\cA_{q,L}} \binom{L}{\va} \exp_q\paren{-\lambda_*\paren{1 - \frac{1}{L}\maxl\curbrkt{\va}}}}}, \notag 
\end{align}
where $ \lambda_*=  \lambda_*(q,\ell,L,p) $ is the solution to the following equation
\begin{align}
p &= \frac{\sum\limits_{\va\in\cA_{q,L}} \binom{L}{\va} \exp_q\paren{-\lambda_*\paren{1 - \frac{1}{L}\maxl\curbrkt{\va}}} \paren{1 - \frac{1}{L}\maxl\curbrkt{\va}}}{\sum\limits_{\va\in\cA_{q,L}} \binom{L}{\va} \exp_q\paren{-\lambda_*\paren{1 - \frac{1}{L}\maxl\curbrkt{\va}}}}. \notag 
\end{align}
\end{remark}

\begin{remark}
\label{rk:lb-lr-hashing}
We specialize the bound \Cref{eqn:lb-lr} to the case $ p = 0, L = \ell + 1,\ell\le q-1 $ corresponding to $(q,\ell + 1)$-hashing. 
\begin{align}
    & - \frac{1}{\ell} \log_q \sqrbrkt{q^{-(\ell + 1)} \sum_{\va\in\cA_{q,\ell + 1}} \binom{\ell + 1}{\va} \exp_q\paren{-\lambda_*\paren{1 - \frac{1}{\ell + 1}\maxl\curbrkt{\va}}}} \label{eqn:expand-expectation-lr-lb} \\
    & - \frac{1}{\ell} \log_q \sqrbrkt{q^{-(\ell + 1)} \paren{\sum_{\substack{\va\in\cA_{q,\ell + 1} \\ \maxl\curbrkt{\va} = \ell}} \binom{\ell + 1}{\va} q^{-\frac{\lambda_*}{\ell+1}} + \sum_{\substack{\va\in\cA_{q,\ell + 1} \\ \maxl\curbrkt{\va} = \ell+1}} \binom{\ell+1}{\va} } } \label{eqn:either-or} \\
    &= - \frac{1}{\ell} \log_q \sqrbrkt{q^{-(\ell + 1)} \paren{\sum_{\substack{\va\in\cA_{q,\ell + 1} \\ \maxl\curbrkt{\va} = \ell}} \binom{\ell + 1}{\va} q^{-\frac{\lambda_*}{\ell+1}} + q^{\ell+1} - \sum_{\substack{\va\in\cA_{q,\ell + 1} \\ \maxl\curbrkt{\va} = \ell}} \binom{\ell+1}{\va} } } \label{eqn:use-multinom-lb-lr} \\
    &= - \frac{1}{\ell} \log_q \sqrbrkt{q^{-(\ell+1)} \paren{ \binom{q}{\ell+1} \binom{\ell + 1}{\one_{\ell+1}} q^{-\frac{\lambda_*}{\ell+1}} + q^{\ell+1} - \binom{q}{\ell+1} \binom{\ell + 1}{\one_{\ell+1}} } } \label{eqn:intro-all-one} \\
    &= - \frac{1}{\ell} \log_q \sqrbrkt{q^{-(\ell+1)} \paren{ \binom{q}{\ell+1}(\ell+1)! \paren{q^{-\frac{\lambda_*}{\ell+1}} - 1} + q^{\ell+1} } } \notag \\
    &= - \frac{1}{\ell} \log_q \sqrbrkt{ 1 - \frac{\binom{q}{\ell+1}(\ell+1)!}{q^{\ell+1}} \paren{ 1 - q^{-\frac{\lambda_*}{\ell+1}} } } . \notag 
\end{align}
\Cref{eqn:expand-expectation-lr-lb} is obtained by expanding the expectation in \Cref{eqn:lb-lr}. 
\Cref{eqn:either-or} follows since for $ \va\in\cA_{q,\ell+1} $, $ \maxl\curbrkt{\va} $ can be either $\ell$ or $\ell+1$. 
We then use the multinomial theorem $ \sum_{\va\in\cA_{q,m}}\binom{q}{\va} = q^m $ to get \Cref{eqn:use-multinom-lb-lr}. 
Observe that the former case ($\maxl\curbrkt{\va} = \ell$) occurs if and only if $ \va $ is a length-$q$ vector with $ \ell+1 $ ones and $q-\ell-1$ zeros. 
Otherwise, the latter case ($\maxl\curbrkt{\va} = \ell + 1$) occurs. 
For the former case, there are $\binom{q}{\ell+1}$ ways to choose the locations of ones, which leads to \Cref{eqn:intro-all-one} where we use $ \one_{\ell+1} $ to denote the all-one vector of length $\ell+1$. 

By similar manipulations, the fixed point equation \Cref{eqn:lb-lr-fp-eqn} with $ p = 0,L = \ell + 1 $ becomes: 
\begin{align}
    0 &= \frac{\binom{q}{\ell+1}\binom{\ell+1}{\one_{\ell+1}} q^{-\frac{\lambda_*}{\ell+1}}\cdot\frac{1}{\ell+1}}{\binom{q}{\ell+1}\binom{\ell+1}{\one_{\ell+1}} q^{-\frac{\lambda_*}{\ell+1}} + q^{\ell+1} - \binom{q}{\ell+1}\binom{\ell+1}{\one_{\ell+1}}} 
    = \frac{q^{-\frac{\lambda_*}{\ell+1}}\cdot\frac{1}{\ell+1}}{q^{-\frac{\lambda_*}{\ell+1}} + \frac{q^{\ell+1}}{\binom{q}{\ell+1}(\ell+1)!} - 1} . \notag 
\end{align}
It is obvious that the unique solution to the above equation is $ \lambda_* = \infty $. 
Therefore, \Cref{eqn:lb-lr} evaluates to 
\begin{align}
    \frac{1}{\ell} \log_q\frac{1}{1 - \frac{\binom{q}{\ell+1}(\ell+1)!}{q^{\ell+1}}} , \notag 
\end{align}
which recovers the lower bound due to Fredman and \komlos \cite{FK} (see also \Cref{eqn:fredman-komlos-bk}). 
Further setting $\ell = q-1$ yields 
\begin{align}
    \frac{1}{q-1} \log_q\frac{1}{1 - \frac{q!}{q^q}} , \notag 
\end{align}
which recovers the Fredman--\komlos lower bound for $q$-hashing \cite{KM} (see also \Cref{eqn:fredman-komlos-lower}). 
\end{remark}

\begin{proof}[Proof of \Cref{thm:lb}.]
Let $ \cC\in[q]^{M\times n} $ be a codebook consisting $M$ codewords from $ [q]^n $. 
Each entry of each codeword is sampled independently and uniformly from $ [q] $. 
Denote $ \cC = \{\vX_1,\cdots,\vX_M\} $. 

For an arbitrary $L$-list $ \{i_1,\cdots,i_L\}\in\binom{[M]}{L} $ of indices, we have
\begin{align}
\expt{\frac{1}{n}\ol{\rad}_\ell(\vbfX_{i_1},\cdots,\vbfX_{i_L})} 
&= \frac{1}{n}\sum_{j=1}^n \expt{1-\frac{1}{L}\plur_\ell(X_{i_1,j},\cdots,X_{i_L,j})} \notag \\
&= \exptover{(X_1,\cdots,X_L)\sim\unif([q])^{\ot L}}{1-\frac{1}{L}\plur_\ell(X_1,\cdots,X_L)} \notag \\
&= \sum_{(x_1,\cdots,x_L)\in[q]^L} q^{-L} {\paren{1 - \frac{1}{L}\plur_\ell(x_1,\cdots,x_L)}} \notag \\
&= q^{-L} \sum_{\va\in\cA_{q,L}} \binom{L}{\va} \paren{1 - \frac{1}{L}\maxl\curbrkt{\va}} > p. \notag 
\end{align}
Therefore, by \cramer's large deviation theorem (\Cref{thm:cramer}), we have
\begin{align}
E &\coloneqq \lim_{n\to\infty}-\frac{1}{n}\log_q\prob{\ol{\rad}_\ell(\vbfX_{i_1},\cdots,\vbfX_{i_L})\le np} \notag \\
&= \lim_{n\to\infty}-\frac{1}{n}\log_q\prob{\sum_{j = 1}^n\paren{1-\frac{1}{L}\plur_\ell(X_{i_1,j},\cdots,X_{i_L,j})}\le np} \notag \\
&= \sup_{\lambda\ge0}\curbrkt{-\lambda p - \log_q\expt{\exp_q\paren{-\lambda\paren{1-\frac{1}{L}\plur_\ell(X_1,\cdots,X_L)}}}} , \notag 
\end{align}
where the expectation is taken over $ (X_1,\cdots,X_L)\sim\unif([q])^{\ot L} $.
For any $ \lambda\in\bbR $, we then compute the moment generating function $ \expt{\exp_q\paren{-\lambda\paren{1-\frac{1}{L}\plur_\ell(X_1,\cdots,X_L)}}} $ in a similar way. 
\begin{align}
& \exptover{(X_1,\cdots,X_L)\sim\unif([q])^{\ot L}}{\exp_q\paren{-\lambda\paren{1-\frac{1}{L}\plur_\ell(X_1,\cdots,X_L)}}} \notag \\
&= q^{-L} \sum_{\va\in\cA_{q,L}} \binom{L}{\va} \exp_q\paren{\lambda\paren{1 - \frac{1}{L}\maxl\curbrkt{\va}}} \eqqcolon \sfM(\lambda). \notag 
\end{align}
Note that 
\begin{align}
\sfM'(\lambda) &= q^{-L} \sum_{\va\in\cA_{q,L}} \binom{L}{\va} \exp_q\paren{\lambda\paren{1 - \frac{1}{L}\maxl\curbrkt{\va}}} \paren{1 - \frac{1}{L}\maxl\curbrkt{\va}}. \notag 
\end{align}
Then 
\begin{align}
E &
= \sup_{\lambda\ge0}\curbrkt{-\lambda p - \log_q \sfM(-\lambda)}. \notag 
\end{align}
The maximizer $ \lambda_* $ is given by the solution to the following equation
\begin{align}
- p + \frac{\sfM'(-\lambda_*)}{\sfM(-\lambda_*)} = 0, \notag 
\end{align}
i.e., 
\begin{align}
p = \frac{\sfM'(-\lambda_*)}{\sfM(-\lambda_*)}. \notag 
\end{align}
Or more explicitly, 
\begin{align}
p &= \frac{\sum\limits_{\va\in\cA_{q,L}} \binom{L}{\va} \exp_q\paren{-\lambda_*\paren{1 - \frac{1}{L}\maxl\curbrkt{\va}}} \paren{1 - \frac{1}{L}\maxl\curbrkt{\va}}}{\sum\limits_{\va\in\cA_{q,L}} \binom{L}{\va} \exp_q\paren{-\lambda_*\paren{1 - \frac{1}{L}\maxl\curbrkt{\va}}}}. \notag 
\end{align}
Finally, by a standard random coding with expurgation argument, the following rate $R$ can be achieved
\begin{align}
R &= \frac{E}{L-1} \label{eqn:choose-r-lr-lb} \\
&= \frac{1}{L-1}(-\lambda_* p - \log_q \sfM(-\lambda_*)) \notag \\
&= \frac{1}{L-1}\curbrkt{-\lambda_*p - \log_q\sqrbrkt{q^{-L} \sum_{\va\in\cA_{q,L}} \binom{L}{\va} \exp_q\paren{-\lambda_*\paren{1 - \frac{1}{L}\maxl\curbrkt{\va}}}}} \notag \\
&= \frac{L}{L-1} - \frac{1}{L-1}\curbrkt{\lambda_*p + \log_q\sqrbrkt{\sum_{\va\in\cA_{q,L}} \binom{L}{\va} \exp_q\paren{-\lambda_*\paren{1 - \frac{1}{L}\maxl\curbrkt{\va}}}}}. \notag 
\end{align}
Indeed, we discard one codeword per \emph{bad} list, i.e., a list $\cL\in\binom{\cC}{L}$ satisfying $ \ol\rad_\ell(\cL)\le np $. 
Since there are $ \binom{M}{L}\doteq q^{nRL} $ lists, by linearity of expectation, the number of codewords to be discarded is expected to be $ \exp_q(nRL - nE) $. 
As long as $ \exp_q(nRL - nE)\ll q^{nR} $, we are left with a $(p,\ell,L)_q$-list-recoverable code of rate essentially $R$. 
The proof is then finished by noting that the requirement is met by the choice of $R$ in \Cref{eqn:choose-r-lr-lb}. 
\end{proof}


\section{Conclusion}
\label{sec:open-q}
In this work, we addressed the basic question of determining the maximum achievable decoding radius for positive rate list-recoverable codes, i.e., we pinned down the list-recovery zero-rate threshold. We then adapted known techniques to show that codes correcting more errors must in fact have \emph{constant} size. Subsequently, we transferred this bound to give upper bounds on the rate of list-recoverable codes for all values of decoding radius. 

As we apply general Ramsey-theoretic tools in bounding the size of list-recoverable codes in the zero-rate regime, our dependence on the corresponding parameters is quite poor, and indeed, we made no efforts to optimize these constants. However, for list-decodable binary codes in the zero-rate, a recent work of Alon, Bukh and Polyanskiy~\cite{abp-2018} derived new (and, in some cases, tight) upper bounds on their size. Obtaining similarly improved size upper bounds for $q$-ary list-decodable/-recoverable codes in the zero-rate regime therefore appears to be a natural next step.

\section{Acknowledgement}
\label{sec:ack}
Part of this work was done while NR was affiliated with the Centrum Wiskunde \& Informatica and supported in part by ERC H2020 grant No.74079 (\mbox{ALGSTRONGCRYPTO}). CY is supported in part by the National Natural Science Foundation of China under Grant 12101403. YZ is grateful to Shashank Vatedka, Diyuan Wu and Fengxing Zhu for inspiring discussions. 

\bibliographystyle{alpha}
\bibliography{ref} 

\appendix

\section{Discussion of Blinovsky's results \cite{blinovsky-2005-ls-lb-qary,blinovsky-2008-ls-lb-qary-supplementary}}
\label{sec:discussion-blinovsky}

As mentioned in \Cref{sec:intro}, part of the motivation of this work is to fill in the gaps in the proofs in  \cite{blinovsky-2005-ls-lb-qary,blinovsky-2008-ls-lb-qary-supplementary} for $q$-ary list-decoding. 
We discuss in detail below the issues therein. 
The main result in \cite{blinovsky-2005-ls-lb-qary} is a Plotkin bound (as our \Cref{thm:plotkin-qary-list-dec-main}) for an arbitrary $q$-ary list-decodable code $\cC\subset[q]^n$. 
For the sake of brevity, we assume in the proceeding discussion that $\cC$ is $w$-constant weight. 
Additional bookkeeping is needed to handle small deviations in the weight, as we did in the proof of \Cref{thm:plotkin-list-dec-qary}. 
The skeleton of the proof in \cite{blinovsky-2005-ls-lb-qary} follows Blinovsky's proof in the \emph{binary} case \cite{blinovsky-1986-ls-lb-binary} which we adopt here as well: $ (i) $ pass to an (approximately) equi-coupled subcode $\cC' = \{\vx_1,\cdots,\vx_M\}\subset\cC$ using a Ramsey reduction; $(ii)$ handle asymmetric coupling using \komlos's argument (and its order-$L$ generalization \cite{bondaschi-dalai}); $ (iii) $ prove an upper bound on the size $M$ of the subcode $\cC'$ using a double-counting argument. 
In completing the double-counting argument, one is required to upper bound the average radius (averaged over all $L$-lists in the subcode) by the zero-rate threshold $ p_*(q,L,w) = 1 - \frac{1}{L}g_{q,L}(w) $:
\begin{align}
    \frac{1}{M^L} \sum_{(i_1,\cdots,i_L)\in[M]^L} \ol{\rad}(\vx_{i_1}, \cdots, \vx_{i_L})
    &= \sum_{k = 1}^n \paren{1 - \frac{1}{L} f_{q,L}(P_k)}
    \le n\paren{1 - \frac{1}{L} g_{q,L}(w)} , \label{eqn:desiderata-ineq}
\end{align}
where $P_k\in\Delta([q])$ is the empirical distribution of the $k$-th column of $\cC'\in[q]^{M\times n}$. 
The equality in \Cref{eqn:desiderata-ineq} is by elementary algebraic manipulations (see \Cref{eqn:interchange-sum-list-dec} for details). 
To show the inequality in \Cref{eqn:desiderata-ineq}, we need the following properties of the functions $ f_{q,L} $ and $ g_{q,L} $:
\begin{enumerate}
    \item \label{itm:desiderata-1}
    For any $ P = (p_1,\cdots,p_q) \in\Delta([q]) $, we have $ f_{q,L}(P) \ge g_{q,L}(1-p_q) $. 
    In words, uniformizing $P$ except one entry will only make $ f_{g,L} $ no larger. 
    
    \item \label{itm:desiderata-2} $g_{q,L}$ is convex as a univariate real-valued function on $[0,(q-1)/q]$. 
\end{enumerate}
If these properties hold, one can deduce \Cref{eqn:db-ub-schur,eqn:use-gql-convex} from which \Cref{eqn:desiderata-ineq} follows. 
However, we observe that the proofs in \cite{blinovsky-2005-ls-lb-qary,blinovsky-2008-ls-lb-qary-supplementary} for both properties above are problematic. 

To show \Cref{itm:desiderata-1} above, the idea in \cite{blinovsky-2005-ls-lb-qary} is to show instead monotonicity of $f_{q,L}$ under the so-called \emph{Robin Hood operation} which averages two distinct entries of $P$. 
Specifically, \cite{blinovsky-2005-ls-lb-qary} attempts to show 
\begin{align}
    f_{q,L}\paren{p_1,\cdots,p_i,\cdots,p_j,\cdots,p_q} &\ge f_{q,L}\paren{p_1,\cdots,\frac{p_i + p_j}{2},\cdots,\frac{p_i + p_j}{2},\cdots,p_q} , \label{eqn:robin-hood}
\end{align}
for any $ 1\le i<j\le q $. 
This suffices since a sequence of Robin Hood operations can turn $P$ into $ P_{q,1-p_q} $ (defined in \Cref{eqn:def-pqw}). 
\cite{blinovsky-2005-ls-lb-qary} then proceeds to show \Cref{eqn:robin-hood} by checking the derivative of a certain function related to the Robin Hood operation. 
Specifically, fix $ (p_k)_{k\in[q]\setminus\{i,j\}} $ and assume $ p_i+p_j = c $ (or equivalently $ \sum_{k\in[q]\setminus\{i,j\}} p_i = 1-c $) for some constant $0\le c\le1$. 
Consider the function $ F_{q,L}\colon[0,c]\to\bbR $ defined as:
\begin{align}
    F_{q,L}(p) &= f_{q,L}\paren{p_1,\cdots,p,\cdots,c-p,\cdots,p_q} , \notag 
\end{align}
i.e., $ f_{q,L} $ evaluated at $P$ with $p_i = p,p_j = c-p$. 
The proof of \Cref{eqn:robin-hood} is reduced to proving $ F_{q,L}'(p)\le0 $ for $ p\in[0,c/2] $ and $ F_{q,L}'(p)\ge0 $ for $p\in[c/2,c]$. 
If true, it implies that $ f_{q,L}(P) $ is minimized at $ p_i = p_j = c/2 $ with fixed $ (p_k)_{k\in[q]\setminus\{i,j\}} $. 
However, we note that the expression of $ F_{q,L}'(p) $ (see the second displayed equation on page 27 of \cite{blinovsky-2005-ls-lb-qary}) is incorrect. 
Upon correcting it, we do not see an easy way to argue its non-positivity/-negativity. 
In particular, the claim in \cite{blinovsky-2005-ls-lb-qary} that $F_{q,L}'(p)$, as a sum of multiple terms, is \emph{term-wise} non-positive/-negative can be in general falsified by counterexamples. 

The proof (attempt) of \Cref{itm:desiderata-2} is deferred to a subsequent paper \cite{blinovsky-2008-ls-lb-qary-supplementary}. 
The methodology thereof is similar to ours, i.e., verifying $g_{q,L}''\ge0$. 
However, the expression of $ g_{q,L}'' $ in \cite{blinovsky-2008-ls-lb-qary-supplementary} is not exactly correct (see the first displayed equation on page 36 of \cite{blinovsky-2008-ls-lb-qary-supplementary} and compare it with ours in \Cref{eqn:exp-for-g''}\footnotemark{}) and we have trouble verifying the case analysis of the values of $G(\cdot)$ (see \Cref{eqn:G-defn} in our notation, denoted by $\gamma(\cdot)$ in \cite{blinovsky-2008-ls-lb-qary-supplementary}) following that expression. 
\footnotetext{Note that the function considered in \cite{blinovsky-2008-ls-lb-qary-supplementary} is, in our notation, $ 1 - \frac{1}{L}g_{q,L}(w) $ instead of $g_{q,L}(w)$ per se as considered in \Cref{thm:convex-gql}.  }

In contrast to Blinovsky's approach \cite{blinovsky-2005-ls-lb-qary,blinovsky-2008-ls-lb-qary-supplementary}, we deduce the monotonicity property of $f_{q,L}$ (cf.\ \Cref{itm:desiderata-1} above) from a stronger property: Schur convexity (cf.\ \Cref{thm:schur-convex-fql}). 
Also, we believe that our proof of the convexity of $g_{q,L}$ (cf.\ \Cref{itm:desiderata-2} above) is cleaner, more transparent and easier to verify. 
Both results can be extended to list-recovery setting. 
Another advantage is that the monotonicity property of $ g_{q,L} $ (specifically, $ g_{q,L} $ is non-increasing in $ [0,(q-1)/q] $ and non-decreasing in $ [(q-1)/q,1] $) which is needed in the proof of the Plotkin bound appears to be a simple consequence of the Schur convexity of $ f_{q,L} $ (see \Cref{prop:mono-gqL}). 
In \cite{blinovsky-2008-ls-lb-qary-supplementary}, this is proved by checking the first derivative of $g_{q,L}$ which involves somewhat cumbersome calculations and case analysis. 

\section{Monotonicity properties of $g_{q,L}$: proof of \Cref{prop:mono-gqL}} \label{subsec:mono-gqL}

In this section we prove \Cref{prop:mono-gqL}, which we recall here for convenience.

\monotonicitygqL*

\begin{proof}
    As $g_{q,L}(w) = f_{q,L}(P_{q,w})$ and $f_{q,L}$ is Schur convex, the proposition follows from the following pair of facts:
    \begin{itemize}
        \item if $w \leq u \leq \frac{q-1}{q}$ then $P_{q,w} \succeq P_{q,u}$;
        \item if $w \geq u \geq \frac{q-1}{q}$ then $P_{q,w} \succeq P_{q,u}$.
    \end{itemize}
    Looking at the first item, if $w \leq u \leq \frac{q-1}{q}$ then $P_{q,w}^{\downarrow} = (1-w, \frac{w}{q-1},\dots,\frac{w}{q-1})$ and $P_{q,u}^{\downarrow} = (1-u,\frac{u}{q-1},\dots,\frac{u}{q-1})$. It therefore suffices to show that 
    \[
        \forall k \in \{0,1,\dots,q-1\},~~1-w + k\cdot\frac{w}{q-1} \geq 1-u + k\cdot\frac{u}{q-1},
    \]
    which is rearranges to 
    \[
        \forall k \in \{0,\dots,q-1\},~~ u-w \geq (u-w)\frac{k}{q-1}.
    \]
    As $u-w \geq 0$, the above is indeed true for all $0 \leq k \leq q-1$.
    
    Considering now the second item, if $w \geq u \geq \frac{q-1}{q}$ then $P_{q,w}^{\downarrow} = (\frac{w}{q-1},\dots,\frac{w}{q-1},1-w)$ and $P_{q,u}^{\downarrow} = (\frac{u}{q-1},\dots,\frac{u}{q-1},1-u)$. Thus, the fact that $P_{q,w} \succeq P_{q,u}$ follows immediately from the fact that 
    \[
        \forall k \in \{1,\dots,q-1\},~~ k\cdot\frac{w}{q-1} \geq k\cdot\frac{u}{q-1}. \qedhere
    \]
\end{proof}

\section{Expressions for $f_{q,L,\ell},g_{q,L,\ell}$ and their derivatives: proof of \Cref{lemma:derivs}} \label{subsec:deriv-expressions}

First, by applying the multinomial theorem, we derive that
\begin{align}
    f_{q,L}(P) &= \exptover{(X_1,\dots,X_L)\sim P^{\ot L}}{\plur(X_1,\dots,X_L)} \notag \\
    &= \sum_{(x_1,\dots,x_L) \in [q]^L}\pl(x_1,\dots,x_L)\probover{(X_1,\dots,X_L)\sim P^{\ot L}}{X_1=x_1,\dots,X_L=x_L} \notag \\
    &= \sum_{(x_1,\dots,x_L) \in [q]^L}\max\{|\{i:x_i=j\}|:j \in [q]\} \prod_{i=1}^n p_{x_i} \notag \\
    &= \sum_{(a_1,\dots,a_q)\in\cA_{q,L}} \binom{L}{a_1,\dots,a_q} \max\curbrkt{a_1,\dots,a_q} \paren{\prod_{i = 1}^q p_i^{a_i}} \notag \\
    &= \sum_{\va\in\cA_{q,L}} \binom{L}{\va} \max\curbrkt{\va} \paren{\prod_{i = 1}^q p_i^{a_i}} . \notag
\end{align}
Next, for any $ j\in[q] $, one can then compute $ \frac{\partial f_{q,L}}{\partial p_j} $ as follows. 
\begin{align}
    &\frac{\partial}{\partial p_j} f_{q,L}(p_1,\dots,p_q) \notag \\
    &= \sum_{\substack{(a_1,\dots,a_q)\in\cA_{q,L} \\ a_j\ge1}} \binom{L}{a_1,\dots,a_q} \max\curbrkt{a_1,\dots,a_q} \paren{\prod_{i\in[q]\setminus\{j\}}p_i^{a_i}} a_j p_j^{a_j - 1} \notag \\
    &= \sum_{\substack{(a_1,\dots,a_{j}-1,\dots,a_q)\in\bbZ_{\ge0}^q \\ a_1 + \cdots + a_j + \cdots + a_q = L}} \frac{L!}{a_1!\cdots(a_j - 1)!\cdots a_q!} \max\curbrkt{a_1,\dots,a_q} \paren{\prod_{i\in[q]\setminus\{j\}}p_i^{a_i}} p_j^{a_j - 1} \notag \\
    &= L\sum_{\substack{(a_1,\dots,a_{j},\dots,a_q)\in\bbZ_{\ge0}^q \\ a_1 + \dots + a_j + \dots + a_q = L-1}} \frac{(L-1)!}{a_1!\cdots a_j!\cdots a_q!} \max\curbrkt{a_1,\dots,a_j+1,\dots,a_q} \paren{\prod_{i\in[q]}p_i^{a_i}} \label{eqn:change-variable} \\
    &= L\sum_{(a_1,\cdots,a_q)\in\cA_{q,L}} \binom{L-1}{a_1,\dots,a_q} \max\curbrkt{a_1,\dots,a_j + 1,\dots,a_q} \paren{\prod_{i=1}^qp_i^{a_i}} \notag \\
    &= L \sum_{\va\in\cA_{q,L}} \binom{L-1}{\va} \max\curbrkt{\va + \ve_j} \paren{\prod_{i = 1}^q p_i^{a_i}} . \notag 
\end{align}
In \Cref{eqn:change-variable}, we apply the change of variable $ (a_j - 1) \to a_j $. 

Now, for $j,k \in [q]$, similar reasoning leads to 
\begin{align}
    & \frac{\partial}{\partial p_k\partial p_j} f_{q,L}(p_1,\dots,p_q) \notag \\
    &= L\sum_{\substack{(a_1,\dots,a_q)\in\cA_{q,L} \\ a_k\ge1}} \binom{L-1}{\va} \max\curbrkt{\va + \ve_j} \paren{\prod_{i\in[q]\setminus\{k\}}p_i^{a_i}} a_k p_k^{a_k - 1} \notag \\
    &= L\sum_{\substack{(a_1,\dots,a_{k}-1,\dots,a_q)\in\bbZ_{\ge0}^q \\ a_1 + \cdots + a_k + \cdots + a_q = L-1}} \frac{(L-1)!}{a_1!\cdots(a_k - 1)!\cdots a_q!} \max\curbrkt{\va + \ve_j} \paren{\prod_{i\in[q]\setminus\{k\}}p_i^{a_i}} p_k^{a_k - 1} \notag \\
    &= L(L-1)\sum_{\substack{(a_1,\dots,a_{j},\dots,a_q)\in\bbZ_{\ge0}^q \\ a_1 + \dots + a_j + \dots + a_q = L-1}} \frac{(L-2)!}{a_1!\cdots a_k!\cdots a_q!} \max\curbrkt{\va + \ve_j + \ve_k} \paren{\prod_{i\in[q]}p_i^{a_i}} \\
    &= L(L-1)\sum_{(a_1,\cdots,a_q)\in\cA_{q,L-2}} \binom{L-2}{a_1,\dots,a_q} \max\curbrkt{\va + \ve_j + \ve_k} \paren{\prod_{i=1}^qp_i^{a_i}} \notag \\
    &= L(L-1) \sum_{\va\in\cA_{q,L-2}} \binom{L-2}{\va} \max\curbrkt{\va + \ve_j + \ve_k} \paren{\prod_{i = 1}^q p_i^{a_i}} . \notag 
\end{align}

Finally, to compute $g_{q,L,\ell}''(w)$ for $w \in [0,1]$,we first apply the chain rule.
\begin{align}
    g_{q,L,\ell}(P_{q,w,\ell}) &= \frac{d}{dw}f_{q,L,\ell}(P_{q,\ell,w}) \notag \\
    &= \sum_{1\leq i,j\leq q-\ell} \paren{\left.\frac{\partial^2}{\partial p_i\partial p_j}\right|_{(p_1,\cdots,p_q) = P_{q,\ell,w}} f_{q,L,\ell}(p_1,\cdots,p_q)} \frac{1}{(q-\ell)^2} \notag \\
    & \quad - 2\sum_{ i=1}^{q-\ell}\sum_{j=q-\ell+1}^q \paren{\left.\frac{\partial^2}{\partial p_i\partial p_j}\right|_{(p_1,\cdots,p_q) = P_{q,\ell,w}} f_{q,L,\ell}(p_1,\cdots,p_q)} \frac{1}{(q-\ell)\ell} \notag \\
    &\quad +\sum_{q-\ell+1 \leq i,j\leq q} \paren{\left.\frac{\partial^2}{\partial p_i\partial p_j}\right|_{(p_1,\cdots,p_q) = P_{q,\ell,w}} f_{q,L,\ell}(p_1,\cdots,p_q)} \frac{1}{\ell^2} \notag .
\end{align}
Next, note that for any $\va \in \cA_{q,L-2}$, if $P_{q,w,\ell} = (p_1,\dots,p_q)$,
\begin{align*}
    \prod_{i=1}^q p_i^{a_i} &= \paren{\frac{w}{q-\ell}}^{a_1+\cdots+a_{q-\ell}}\paren{\frac{1-w}{\ell}}^{a_{q-\ell+1}+\cdots+a_{q}} \\
    &= \paren{\frac{w}{q-\ell}}^{L-2}\paren{\frac{(q-\ell)(1-w)}{\ell w}}^{a_{q-\ell+1}+\cdots+a_q} . \
\end{align*}
Thus, we conclude
\begin{align}
    g_{q,L,\ell}(P_{q,w,\ell})
    &= \sum_{1\leq i,j\leq q-\ell} L(L-1)\sum_{\va \in \cA_{q,L-2}}\binom{L-2}{\va} \maxl\{\va + \ve_i + \ve_j\} \notag \\
    &\quad\quad \cdot \paren{\frac{w}{q-\ell}}^{L-2}\paren{\frac{(q-\ell)(1-w)}{\ell w}}^{a_{q-\ell+1}+\cdots+a_q} \frac{1}{(q-\ell)^2} \notag \\
    & \quad - 2\sum_{ i=1}^{q-\ell}\sum_{j=q-\ell+1}^q L(L-1)\sum_{\va \in \cA_{q,L-2}}\binom{L-2}{\va} \maxl\{\va + \ve_i + \ve_j\} \notag \\
    &\quad\quad \cdot \paren{\frac{w}{q-\ell}}^{L-2}\paren{\frac{(q-\ell)(1-w)}{\ell w}}^{a_{q-\ell+1}+\cdots+a_q} \frac{1}{(q-\ell)\ell} \notag \\
    &\quad +\sum_{q-\ell+1 \leq i,j\leq q} L(L-1)\sum_{\va \in \cA_{q,L-2}}\binom{L-2}{\va} \maxl\{\va + \ve_i + \ve_j\} \notag \\
    &\quad\quad \cdot \paren{\frac{w}{q-\ell}}^{L-2}\paren{\frac{(q-\ell)(1-w)}{\ell w}}^{a_{q-\ell+1}+\cdots+a_q} \frac{1}{\ell^2} \notag \\
    &= L(L-1)\paren{\frac{w}{q-\ell}}^{L-2}\sum_{\va \in \cA_{q,L-2}} \paren{\frac{(q-\ell)(1-w)}{\ell w}}^{a_{q-\ell+1}+\cdots+a_q} \notag \\
    &\quad \cdot \left(\frac{1}{(q-\ell)^2}\sum_{1 \leq i,j \leq q-\ell} \maxl\{\va + \ve_i + \ve_j\} - \frac{2}{(q-\ell)\ell}\sum_{i=1}^{q-\ell}\sum_{j=q-\ell+1}^{q} \max_{\ell}\{\va + \ve_i + \ve_j\}\right. \notag \\
    &\quad\quad \left.+ \frac{1}{\ell^2}\sum_{q-\ell+1 \leq i,j \leq q} \maxl\{\va + \ve_i + \ve_j\} \right) \ . 
\end{align}
Recalling the definition of $G_\ell(\va)$, the result follows. 


\section{Auxiliary lemmas}
\label{app:aux-lem}

\begin{theorem}[\cramer]
\label{thm:cramer}
Let $ X_1,\cdots,X_n $ be a sequence of $n$ i.i.d.\ real-valued random variables
Define $ S_n \coloneqq \frac{1}{n}\sum_{i=1}^n X_i $. 
Fix any $ q>1 $. 
Suppose $ \log_q\expt{\exp_q\paren{\lambda X_1}}<\infty $ for every $ \lambda\in\bbR $. 
Then for any $ t>\expt{X_1} $, 
\begin{align}
    \limsup_{n\to\infty} -\frac{1}{n} \log_q \prob{S_n \ge t} &\le \sup_{\lambda\in\bbR} \curbrkt{\lambda t - \log_q\expt{\exp_q\paren{\lambda X_1}}} . \notag 
\end{align}
\end{theorem}

\begin{theorem}[Hypergraph Ramsey \cite{ramsey,erdos-rado-hypergraph-ramsey}]
\label{thm:hypergraph-ramsey}
Let $ c, k, M\in\bbZ_{\ge2} $. 
Then there exists a positive integer $ N = N(c,k,M) $ such that if the hyperedges of a complete $k$-uniform hypergraph $ \cH_{k,N} $ (i.e., a hypergraph in which every subset of $k$ vertices is connected by an hyperedge) of size (i.e., number of vertices) $ N $ are coloured with $c$ different colours, then $ \cH_{k,N} $ must contain a subgraph which $(i)$ is a complete $k$-uniform hypergraph of size at least $M$, and $(ii)$ has the same colour on all of its hyperedges. 
\end{theorem}

\begin{theorem}[\cite{blinovsky-2005-ls-lb-qary,bondaschi-dalai}]
\label{eqn:komlos}
Let $ M\ge2L-2 $. 
Let $ X_1,\cdots,X_M $ be a sequence of $[q]$-valued random variables. 
Let $ P_{X_1,\cdots,X_M} $ denote their joint distribution. 
Suppose there exists a distribution $ Q_{\ol{X}_1,\cdots,\ol{X}_L}\in\Delta([q]^L) $ and a constant $ \eps>0 $ such that\footnote{Here we use $ P_{X_{i_1},\cdots,X_{i_L}} $ to denote the joint distribution of $ X_{i_1},\cdots,X_{i_L} $ which can be obtained by taking the corresponding marginal of the full joint distribution $ P_{X_1,\cdots,X_M} $.} 
\begin{align}
\norminf{P_{X_{i_1},\cdots,X_{i_L}} - Q_{\ol{X}_1,\cdots,\ol{X}_L}} \le \eps \notag 
\end{align}
for every $ 1\le i_1<\cdots<i_L\le M $.  
Then 
\begin{align}
\max_{\pi\in S_L} \norminf{Q_{\ol{X}_1,\cdots,\ol{X}_L} - Q_{\ol{X}_{\pi(1)},\cdots,\ol{X}_{\pi(L)}}} 
&\le 2L^3\sqrt{\frac{2L}{M-(L-2)}} + 4L^3\sqrt{q^{L-2}\eps} + L^2\eps . \notag 
\end{align}
\end{theorem}

\begin{lemma}[{\cite[pp.\ 308-310]{mcwilliams-sloane-book}}]
\label{lem:binom-coeff-bound}
Suppose $w\in[0,1]$ satisfies $ 1\le nw\le n-1 $. 
Then
\begin{align}
    \frac{2^{nH(w)}}{\sqrt{8nw(1-w)}} &\le \binom{n}{nw} \le 2^{nH(w)} . \notag 
\end{align}
\end{lemma}

\begin{lemma}
\label{lem:vol-bound}
For any $ q\in\bbZ_{\ge2} $ and $w\in[0,1]$ such that $ 1\le nw\le n-1 $, 
\begin{align}
    \frac{1}{\sqrt{8nw(1-w)}} \cdot q^{nH_q(w)} \le |\bham(\vq,nw)| 
    &\le nw \cdot q^{nH_q(w)} . 
    \notag 
\end{align}
Furthermore, suppose $ 1\le\ell\le q $ is an integer and 
let $ \vcY_\ell \coloneqq \{q-\ell+1,\ldots,q\}^n\in\binom{[q]}{\ell}^n $. 
Then 
\begin{align}
    \frac{1}{\sqrt{8nw(1-w)}} \cdot q^{nH_{q,\ell}(w)} \le |\blr(\vcY_\ell,nw)| 
    \le nw\cdot q^{nH_{q,\ell}(w)} . \notag
\end{align}
\end{lemma}

\begin{lemma}[Covering]
\label{lem:cov}
There exists a set $ \{\vc_1,\cdots,\vc_K\}\subset[q]^n $ of vectors such that 
\begin{enumerate}
    \item $ \bigcup_{i=1}^K \bham(\vc_i,nw) = [q]^n $;
    \item $ K\le n\ln(q)\cdot\sqrt{8nw(1-w)}\cdot q^{n(1-H_q(w))} + 1 $. 
\end{enumerate}
\end{lemma}

\begin{proof}
The lemma follows from a simple random construction. 
Let $ \vC_1,\cdots,\vC_K $ be points independent and uniformly distributed in $ [q]^n $, for $K$ satisfying
\begin{align}
    K &\ge n\ln(q)\cdot\sqrt{8nw(1-w)}\cdot q^{n(1-H_q(w))} + 1 . 
    \label{eqn:k-bound}
\end{align}
Let us upper bound the probability that such a set of vectors does \emph{not} cover the whole space $ [q]^n $. 
\begin{align}
    & \prob{\exists \vx\in[q]^n,\; \forall i\in[K],\; \disth(\vx,\vC_i)>nw} \notag \\
    &\le q^n \paren{1 - \frac{|\bham(\vq, nw)|}{q^n}}^K \notag \\
    &\le q^n \paren{1 - \frac{q^{nH_q(w)}}{\sqrt{8nw(1-w)}}\cdot\frac{1}{q^n}}^K \notag \\
    &\le q^n \exp\paren{-\frac{1}{\sqrt{8nw(1-w)}} \cdot q^{-n(1-H_q(w))} \cdot K} \label{eqn:use-e-bound} \\
    &= \exp\paren{n\ln(q)-\frac{1}{\sqrt{8nw(1-w)}} \cdot q^{-n(1-H_q(w))} \cdot K} \notag \\
    &< 1 . \label{eqn:use-k-bound}
\end{align}
In \Cref{eqn:use-e-bound}, we use the elementary inequality $ (1 - x/n)^n\le e^{-x} $ for any $ |x|\le n $. 
\Cref{eqn:use-k-bound} follows from the choice of $K$ (cf.\ \Cref{eqn:k-bound}).

We have just shown that by taking a random set $ \{\vC_i\}_{i=1}^K $, the probability that it covers $ [q]^n $ is strictly positive, provided that $ K $ satisfies \Cref{eqn:k-bound}. 
This finishes the proof. 
\end{proof}

\begin{lemma}[Covering, list-recovery]
\label{lem:cov-list-rec}
There exists a set $ \{\vcC_1,\cdots,\vcC_K\}\subset\binom{[q]}{\ell}^n $ of vectors such that
\begin{enumerate}
    \item $ \bigcup_{i=1}^K \blr(\vcC_i,nw) = [q]^n $;
    \item $ K\le n\ln(q) \cdot \sqrt{8nw(1-w)} \cdot q^{n(1-H_{q,\ell}(w))} + 1 $. 
\end{enumerate}
\end{lemma}

\begin{proof}
The lemma again follows from a simple random construction. 
Let $ \vcC_1,\cdots,\vcC_K $ be sequences of subsets independent and uniformly distributed in $ \binom{[q]}{\ell}^n $, for $K$ satisfying
\begin{align}
    K &\ge n\ln(q) \cdot \sqrt{8nw(1-w)} \cdot q^{n(1-H_{q,\ell}(w))} + 1 .  \label{eqn:def-k-cov-lr}
\end{align}
Let us upper bound the probability that the union of the list-recovery balls centered around $ \vcC_1,\cdots,\vcC_K $ does \emph{not} cover $ [q]^n $. 
\begin{align}
    \prob{\exists\vx\in[q]^n,\; \forall i\in[K],\; \distlr(\vx,\vcC_i) > nw} 
    &\le q^n \paren{1 - \probover{\vcC\sim\binom{[q]}{\ell}^n}{\distlr(\vq,\vcC)\le nw}}^K . \label{eqn:cov-lr-tobound}
\end{align}
The probability in the parentheses above can be calculated using a simple counting argument. 
\begin{align}
    \probover{\vcC\sim\binom{[q]}{\ell}^n}{\distlr(\vq,\vcC)\le nw}
    &= \frac{1}{\binom{q}{\ell}^n} \sum_{i=0}^{nw} \binom{n}{i} \binom{q-1}{\ell}^i \binom{q-1}{\ell-1}^{n-i} \notag \\
    &\ge \frac{1}{\binom{q}{\ell}^n} \binom{n}{nw} \binom{q-1}{\ell}^{nw} \binom{q-1}{\ell-1}^{n(1-w)} \notag \\
    &\ge \frac{1}{\sqrt{8nw(1-w)}} \exp_q\left(n\left[-\log_q\binom{q}{\ell} + w\log_q\frac{1}{w} + (1-w)\log_q\frac{1}{1-w} \right.\right. \notag \\
    &\qquad\qquad\qquad\qquad\qquad \left.\left. + w\log_q\binom{q-1}{\ell} + (1-w)\log_q\binom{q-1}{\ell-1} \right]\right) \notag \\
    &= \frac{1}{\sqrt{8nw(1-w)}} \exp_q\left(-n\left[w\log_q\frac{\binom{q}{\ell}}{\binom{q-1}{\ell}} + (1-w)\log_q\frac{\binom{q}{\ell}}{\binom{q-1}{\ell-1}} \right.\right. \notag \\
    &\qquad\qquad\qquad\qquad\qquad \left.\left. - w\log_q\frac{1}{w} - (1-w)\log_q\frac{1}{1-w}\right]\right) \notag \\
    &= \frac{1}{\sqrt{8nw(1-w)}} \exp_q\paren{-n\sqrbrkt{w\log_q\frac{qw}{q-\ell} + (1-w)\log_q\frac{q(1-w)}{\ell}}} \notag \\
    &= \frac{1}{\sqrt{8nw(1-w)}} q^{-n(1-H_{q,\ell}(w))} . \notag
\end{align}
Bounding \Cref{eqn:cov-lr-tobound} in a similar way to the proof of \Cref{lem:cov}, we have
\begin{align}
    \prob{\exists\vx\in[q]^n,\; \forall i\in[K],\; \distlr(\vx,\vcC_i) > nw}
    \le q^n \exp\paren{-\frac{q^{-n(1-H_{q,\ell}(w))}}{\sqrt{8nw(1-w)}}K}
    <1 , \notag
\end{align}
where the last strict inequality follows from the choice of $K$ (cf.\ \Cref{eqn:def-k-cov-lr}). 
According to the probabilistic method, the lemma follows. 
\end{proof}

\begin{remark}
\label{rk:tight-cov}
Though not needed in this paper, a converse statement can be shown that the exponents $ 1-H_q(w) $ in \Cref{lem:cov} and $ 1-H_{q,\ell}(w) $ in \Cref{lem:cov-list-rec} can not be further reduced. 
That is, \emph{any} covering (not necessarily the random ones used in \Cref{lem:cov,lem:cov-list-rec}) of $[q]^n$ must have size at least $ q^{n(1-H_q(w) - \eps)} $ w.r.t.\ the Hamming metric or at least $ q^{n(1-H_{q,\ell}(w) - \eps)} $ w.r.t.\ the list-recovery metric, for any constant $\eps>0$. 
\end{remark}

\end{document}